\documentclass[10pt,journal,twocolumn]{IEEEtran}
\usepackage{amsmath,etoolbox}
\usepackage{amsfonts,xcolor,svg}
\usepackage{amssymb,balance,tikz,multirow,makecell}
\usetikzlibrary{arrows}

\usepackage{xspace}
\usepackage{amsthm}
\usepackage{cite}
\usepackage{enumerate}

\usepackage{graphicx}
\usepackage{color}
\usepackage{bbm}

\newtoggle{paper}
\togglefalse{paper}

\newtheorem{theorem}{Theorem}
\newtheorem{lemma}{Lemma}
\newtheorem{proposition}{Proposition}
\newtheorem{corollary}{Corollary}
\newtheorem{definition}{Definition}
\theoremstyle{definition}\newtheorem{remark}{Remark}
\theoremstyle{definition} \newtheorem{example}{Example}
\newtheorem{claim}{Claim}

\newcommand{\comment}[1]{}

\newcommand{\RRB}[1]{\textcolor{black}{#1}}

\def\cS{\mbox{$\cal{S}$}}

\newcommand{\C}{\ensuremath{\mathcal{C}}\xspace}
\newcommand{\X}{\ensuremath{\mathcal{X}}\xspace}
\newcommand{\Y}{\ensuremath{\mathcal{Y}}\xspace}
\newcommand{\Z}{\ensuremath{\mathcal{Z}}\xspace}
\newcommand{\U}{\ensuremath{\mathcal{U}}\xspace}
\newcommand{\V}{\ensuremath{\mathcal{V}}\xspace}

\newcommand{\E}{\ensuremath{\mathcal{E}}\xspace}
\newcommand{\G}{\ensuremath{\mathcal{G}}\xspace}

\title{Interactive Secure Function Computation}
\author{
	Deepesh Data$^\dag$, Gowtham~R.~Kurri$^\dag$, Jithin~Ravi, Vinod~M.~Prabhakaran
	\thanks{\dag \ Equal contribution.
	
		D.~Data is with University of California, Los Angeles, USA.
		G.~R.~Kurri and V.~M.~Prabhakaran are with School of Technology \& Computer Science at Tata Institute of Fundamental Research, Mumbai, India.
		J.~Ravi is with the Signal Theory and Communications Departmentat, Universidad Carlos III de Madrid, Legan\'es, Spain. Email: deepeshdata@ucla.edu, k.raghunath@tifr.res.in, rjithin@tsc.uc3m.es, vinodmp@tifr.res.in.
		
	}
}

\begin{document}
\maketitle
 \begin{abstract}
We consider interactive computation of randomized functions between two users with the following privacy requirement: the interaction should not reveal to either user any extra information about the other user's input and output other than what can be inferred from the user's own input and output. We also consider the case where privacy is required against only one of the users. For both cases, we give single-letter expressions for feasibility and optimal rates of communication. Then we discuss the role of common randomness and interaction in both privacy settings. \RRB{We also study perfectly secure non-interactive computation
when only one of the users computes a randomized function based on a single transmission from the other user. 
We characterize randomized functions which can be perfectly securely computed in this model and obtain tight bounds on the optimal message lengths in all the privacy settings.}
\end{abstract}
\iftoggle{paper}
{\textit{A full version of this paper is accessible at: \url{http://www.tifr.res.in/~k.raghunath/ISIT2018/twoparty.pdf}}}{\ignorespaces}

\section{Introduction}
\label{sec_intro}

In his seminal work on cryptography, Shannon considered secure communication between two users in the presence of an eavesdropper~\cite{Shannon49}. 
How mutually trusting users can collaborate while revealing as little information about their data as possible to untrusted eavesdroppers has been studied extensively in the pages of this journal. Collaboration in the form of secure communication over a wiretap channel~\cite{Wyner1975C,CsiszarK78},
 secret-key agreement via public discussion~\cite{AhlswedeC93,Maurer93,CsiszarN00}, computation and channel simulation in the presence of an eavesdropper~\cite{TyagiNG11,GohariYA12}, secure source coding~\cite{PrabhakaranR07,GunduzEP08,VillardP13}, and secure network coding~\cite{CaiY11,ElRouayhebSoljaninSprintson2012} %
%S. El Rouayheb, E. Soljanin and A. Sprintson, "Secure Network Coding for Wiretap Networks of Type II," in IEEE Transactions on Information Theory, vol. 58, no. 3, pp. 1361-1371, March 2012.
 have received considerable attention. For a more exhaustive set of references, see~\cite{LiangPS11,Bloch11}. 
%@article{LiangPS11,
%year = {2009},
%volume = {5},
%journal = {Foundations and Trends® in Communications and Information Theory},
%title = {Information Theoretic Security},
%number = {4–5},
%pages = {355-580},
%author = {Yingbin Liang and H. Vincent Poor and Shlomo Shamai (Shitz)}
%}
The object of study in this work is how mutually {\em distrusting} users may collaborate while preserving privacy. 
Specifically, we consider a function computation problem between two users
(Figure~\ref{fig_two_user}). They observe memoryless sources (inputs) $X$ and
$Y$, respectively, and communicate interactively over a noiseless communication
link to compute {\em randomized functions} (outputs) $Z_1$ and $Z_2$,
respectively.  Common and private randomness which is independent of the inputs
$X$ and $Y$ is available to both of them. They want to compute the functions in
such a way that neither of them learn any extra information about the other
user's input and output other than what their own input and output reveal. We
assume that the users are {\em honest-but-curious}, i.e., they faithfully
follow the given protocol, but will try to infer extra information at the end
of the protocol. This secure computation problem is specified by a pair
$(q_{XY},q_{Z_1Z_2|XY})$, where $q_{XY}$ is the input distribution and
$q_{Z_1Z_2|XY}$ specifies the randomized function. Our goal is to determine
whether a function is securely computable and, when it is computable, determine
the optimal rates of interactive communication required.
 \begin{figure}[htbp]
  \centering
   \includegraphics[scale=1.09]{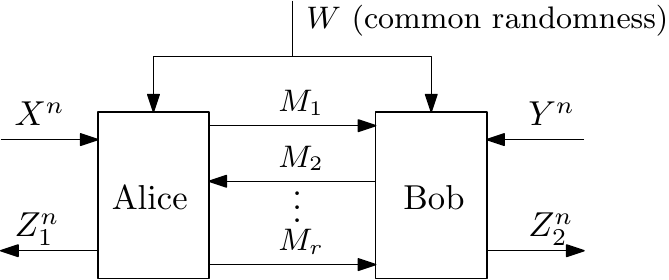}
   \caption{Secure interactive randomized function computation. The case where Alice starts the communication is shown. Privacy against Alice requires that $M_1, \cdots, M_r$ should not reveal anything about $(Y^n,Z_2^n)$ other than what can be inferred from $(X^n,Z_1^n)$. Similar conditions should hold for privacy against Bob.}
   \label{fig_two_user}
  \end{figure}
  
There is a vast literature on two-user interactive computation (with no privacy) in computer science (see, e.g., \cite{KushilevitzN97,
BravermanGPW13} and references therein) and in information theory~\cite{Kaspi85,OrlitskyR01,MaI11,YassaeeGA15,TyagiVVW17}.
Kaspi~\cite{Kaspi85} obtained the rate-distortion trade-off for two-user
interactive source coding. The rate region for optimal two-round interactive lossless computation was obtained by Orlitsky and Roche~\cite{OrlitskyR01}.  Ma and
Ishwar~\cite{MaI11} derived the optimal rate region for an arbitrary number of rounds. Using an example, they also showed that three rounds of interactions (which is clearly not needed for computation to be feasible) can strictly out-perform two rounds. 
%Interactive computation of deterministic functions was addressed by  and also by. 
Optimal interactive schemes for computing randomized functions
were studied by Yassaee et al.~\cite{YassaeeGA15} and Tyagi et
al.~\cite{TyagiVVW17}. 

%{\color{black} There is also a large body of work on secure communication and
%secret-key agreement (see, e.g., \cite{Wyner1975C, AhlswedeC93, Maurer93,
%CsiszarN00, Bloch11, CaiY11}).  In these problems privacy is required against a
%third-party eavesdropper. In contrast, in the secure computation problem in
%this paper, privacy is required against the legitimate users themselves.

Secure computation (with privacy against the users themselves) has been studied
in cryptography under computational as well as information theoretic secrecy
(see~\cite{CramerDBNB15} and references therein). It is known that not all
functions are information theoretically securely computable by two users
interacting over a noiseless link, whereas, under computational secrecy, all
functions are securely computable \cite{Yao82,Yao86}\footnote{Note that if the
two users have access to correlated random variables or a noisy channel, a
larger class of functions may be securely computed~\cite{CREK88}. A
characterization of such stochastic resources which allow any function to be
securely computed is given in~\cite{MajiPR12}. In this paper, we do not
consider such resources.}. A combinatorial characterization of two-user
securely computable deterministic functions was given by
Beaver~\cite{Beaver89} and Kushilevitz~\cite{Kushelvitz92}. An alternative characterization using the common
randomness generated by interactive deterministic protocols was provided
by Narayan et al. ~\cite{NarayanTW15}. A combinatorial characterization of two-user securely
computable randomized functions is still not known; some partial results were
obtained by Maji et al.~\cite{MajiPR12}, and Data and Prabhakaran~\cite{DataP17}. A characterization of two-user securely
computable output distributions with no inputs and no common randomness was
given by Wang and Ishwar~\cite{WangI11}. In contrast to these, function computation with
privacy against an {\em eavesdropper} who has access to the interactive
communication was studied by Tyagi et al.~\cite{TyagiNG11} and Gohari et al.~\cite{GohariYA12}.

Secure computation among $n$ users ($n>2$) with pairwise, private, error-free
communication links was studied by Ben-Or, Goldwasser, and Wigderson~\cite{BGW88} and Chaum, Cr{\'e}peau, and Damg\aa rd~\cite{CCD88}. It was shown that any
function can be securely computed even if any $t<\frac{n}{2}$
honest-but-curious users collude\footnote{For malicious cases where colluding
users may arbitrarily deviate from the protocol, secure computation of all
functions is feasible under a stricter threshold of $t<\frac{n}{3}$ colluding
users~\cite{BGW88,CCD88}.}. The problem of finding optimal protocols from a
communication point-of-view remains mostly open~\cite{DPP16}.

\RRB{Communication rate for secret-key agreement through public discussion was studied by Tyagi~\cite{Tyagi13} for the two-terminal scenario, and Mukherjee et al.~\cite{MukherjeeKS16} and Chan et al.~\cite{ChanMKZ2017} for the multi-terminal scenario. Communication versus rounds trade-off for secret key and common randomness generation was studied by Bafna et al.~\cite{SudanGGB2019}.}

We consider two-user interactive secure function computation
%by two honest-but-curious users 
in two privacy settings: (i) when privacy is required against both users, and
(ii) when privacy is required against only one of the users. The users are
assumed to be honest-but-curious, i.e., they follow the protocol faithfully,
but attempt to obtain information about the other user's input and output from
everything they have access to at the end of the protocol. The function maybe
randomized. Our main contributions are as follows: 
\begin{itemize}

\item For each of the privacy settings mentioned above, we show that the set of
\emph{asymptotically} securely computable functions (see
Definition~\ref{defn:asymp_secur_comp}) is the same as the set of
\emph{one-shot, perfectly} securely computable functions.
We also show that such a property is not true for all secure function computation problems, in general. Function computation with privacy (only) against an eavesdropper~\cite{TyagiNG11} turns out to be a counterexample (see Remark~\ref{remark:eavesdropper}).
Furthermore, we give
single-letter expressions for the asymptotic rate regions
(Theorems~\ref{Feasibility_AB}, \ref{Thm_Rate_Region_AB},
\ref{Thm_Rate_Region_A} and \ref{Thm_Rate_Region_B}).  

\item  We show that for a class of functions (including \RRB{all} deterministic
functions), checking secure computability (with privacy requirement against
both users) is equivalent to checking whether \emph{cut-set} lower bounds for
computation with no privacy requirements can be met (Theorem~\ref{cutset}).

\item When privacy is required against both users, we give a necessary and
sufficient condition for  common randomness to be helpful in improving the
communication rate (see Proposition~\ref{prop:CR}).

\item When no privacy is required, any function can be computed in two rounds
by exchanging the inputs. However, Ma and Ishwar~\cite{MaI11} showed that there
are functions for which more rounds of interaction can strictly improve the
communication rate.  For a function securely computable with privacy against
both users, (depending on the function) a certain minimum number of rounds of
interaction may be required to securely compute it. In contrast with function
computation without any privacy required, we show that further interaction does
not improve the communication rate (see Proposition~\ref{prop:Int}). However,
when privacy is required against only one user, we show that both interaction and common randomness may help to improve the communication rate (see Propositions~\ref{prop3} and \ref{newprop}). We also give a condition under which neither interaction nor common randomness helps to enlarge the rate region (see Proposition~\ref{prop5}).

%\item For (one-shot) perfectly secure computation, we derive upper and lower
%bounds on the optimal expected length of transmission in the {\em
%non-interactive setting} where only Bob produces an output using a {single
%transmission} from Alice (Theorems~\ref{thm:ps_rate-1}, \ref{thm:ps_rate-2} and
%\ref{thm:ps_rate-3}).
\item \RRB{For \emph{perfectly secure} computation, we derive matching upper and lower bounds on the optimal expected length of transmission in several \emph{non-interactive} privacy settings where only Bob produces an output using a \emph{single transmission} from Alice. We study these in a one-shot setting (i.e., with only one copy of inputs), as well as a setting where sequence of i.i.d. samples of inputs are given as inputs. The different scenarios include $(i)$~when privacy is required against both users; our results are stated in Theorem~\ref{thm:ps_rate-1} (one-shot) and Corollary~\ref{cor:ps_rate-1-n} (sequence of inputs), $(ii)$ when privacy is required only against Alice; our results are stated in Theorem~\ref{thm:ps_rate-2} (one-shot) and Corollary~\ref{cor:ps_rate-2-n} (sequence of inputs), $(iii)$ when privacy is required only against Bob; our results are stated in Theorem~\ref{thm:ps_rate-3} (one-shot) and Corollary~\ref{cor:ps_rate-3-n} (sequence of inputs). Since not all randomized functions are securely computable in the first and the third scenarios, we also give a characterization of securely computable randomized functions in these settings; see Lemma~\ref{lem:characterization} in Appendix~\ref{appendixD}\footnote{\RRB{When there is no input distribution, Kilian~\cite{Kilian00} already gave a characterization of securely computable randomized functions when privacy is required against both the users. We extend his proof to the case when we have a probability distribution on the inputs.}} and Theorem~\ref{thm:one-round-characterization}.} 
\end{itemize}

The paper is organized as follows.  We present our problem formulation in
Section~\ref{sec_problm_formlt} and provide results on feasibility and rate
region characterization in Section~\ref{section:results}. The role of
interaction and common randomness is discussed in
Section~\ref{sec_role_CR_Intr}. We illustrate the communication cost of
security using an example in Section~\ref{section:cost}. Bounds on the optimal
expected length of transmission for perfectly secure non-interactive
computation are provided in Section~\ref{Sec_nointeraction}. Open questions and future directions are discussed in Section~\ref{section:discussion}.

\section{Problem Formulation and Definitions}
\label{sec_problm_formlt}
 
A secure randomized function computation problem is specified by a pair $\left(q_{XY},q_{Z_1Z_2|XY}\right)$, where $X,Y,Z_1$ and $Z_2$ take values in finite sets $\mathcal{X},\mathcal{Y},\mathcal{Z}_1$ and $\mathcal{Z}_2$, respectively. Inputs to Alice and Bob are $X^n$ and $Y^n$, respectively, where $(X_i,Y_i)$, $i=1,\dots , n$, are independent and identically distributed (i.i.d.) with distribution $q_{XY}$. Both users have access to a common random variable $W$, which is independent of $(X^n,Y^n)$ and uniformly distributed over its alphabet $\mathcal{W}=[1:2^{nR_0}]$. The users interactively communicate in $r$ rounds over a noiseless bidirectional link. Their goal is to \emph{securely} compute the randomized function $q_{Z_1Z_2|XY}$, i.e., to output $Z_1^n$ and $Z_2^n$,  respectively, such that they are (approximately) distributed according to $q^{(n)}_{Z_1Z_2|XY}(z_1^n,z_2^n|x^n,y^n):=\Pi_{i=1}^n q_{Z_1Z_2|XY}(z_{1i},z_{2i}|x_i,y_i)$ while {\color{black}(approximately)} preserving privacy in the sense that {\color{black}the average amount of additional information that a user learns about the input and output of the other user goes to zero as $n$ tends to infinity.} 
%{\color{red}a user does not learn any additional information about the other user's input and output other than what can be inferred from the user's own input and output.}
The users are assumed to be honest-but-curious. We consider this problem in two different cases: (i) when privacy is required against both users, and (ii) when privacy is required against only one of the users. In both cases we wish to determine whether secure computation is feasible in any arbitrary $r\in\mathbb{N}$ number of rounds and when feasible, characterize the set of achievable rates. Next we formally state the problem assuming that Alice starts the communication. 
\begin{definition}
A \emph{protocol} $\Pi_n$ with $r$ interactive rounds of communication consists of 
\begin{itemize}
\item a set of $r$ randomized encoders with p.m.f.'s $p^{\emph{E}_1}(m_i|x^n,w,m_{[1:i-1]})$ for odd numbers $i\in[1:r]$ 
and $p^{\emph{E}_2}(m_i|y^n,w,m_{[1:i-1]})$ for even numbers $i\in[1:r]$, where $M_i$ is the message transmitted in the $i^{\emph{th}}$ round,
\item two randomized decoders $p^{\emph{D}_1}(z_1^n|x^n,w,m_{[1:r]})$ and $p^{\emph{D}_2}(z_2^n|y^n,w,m_{[1:r]})$.
\end{itemize}
\end{definition}
Let $p^{\text{(induced)}}_{W,X^n,Y^n,M_{[1:r]},Z_1^n,Z_2^n}$ denote the induced distribution of the protocol $\Pi_n$.
\begin{align}
&p^{\text{(induced)}}(w,x^n,y^n,m_{[1:r]},z_1^n,z_2^n)=\frac{1}{2^{nR_0}}\prod_{i=1}^n q(x_i,y_i) \times\nonumber\\
&\left[\prod_{i:\text{odd}}p^{\text{E}_1}(m_i|x^n,w,m_{[1:i-1]})\prod_{j:\text{even}}p^{\text{E}_2}(m_j|y^n,w,m_{[1:j-1]})\right]\nonumber\\
&\times p^{\text{D}_1}(z_1^n|x^n,w,m_{[1:r]})p^{\text{D}_2}(z_2^n|y^n,w,m_{[1:r]})\label{eqn:asymptoticdist}.
\end{align}
\begin{definition}\label{defn:asymp_secur_comp}
$(q_{XY},q_{Z_1Z_2|XY})$ is \emph{asymptotically securely computable} in $r$ rounds with privacy against both users if there exists a sequence of protocols $\Pi_n$ 
such that, for every $\epsilon>0$, there exists a large enough $n$ such that 
\begin{align}
\left\lVert p^{\emph{(induced)}}_{X^n,Y^n,Z_1^n,Z_2^n}-q^{(n)}_{X,Y,Z_1,Z_2}\right\rVert_1&\leq \epsilon,\label{eqn:asymptotic_1}\\
I(M_{[1:r]},W;Y^n,Z_2^n|X^n,Z_1^n)&\leq n\epsilon,\label{eqn:asymptotic_2}\\
I(M_{[1:r]},W;X^n,Z_1^n|Y^n,Z_2^n)&\leq n\epsilon, \label{eqn:asymptotic_3}
\end{align}
where \begin{align*}
&q^{(n)}_{X,Y,Z_1,Z_2}(x^n,y^n,z_1^n,z_2^n):=\nonumber\\
&\hspace{12pt}\Pi_{i=1}^n \left[q_{XY}(x_i,y_i)q_{Z_1Z_2|XY}(z_{1i},z_{2i}|x_i,y_i)\right].
\end{align*}
%\hspace{2.1cm}$\Pi_{i=1}^n \left[q_{XY}(x_i,y_i)q_{Z_1Z_2|XY}(z_{1i},z_{2i}|x_i,y_i)\right]$.
\end{definition}
Note that \eqref{eqn:asymptotic_2} {\color{black}is the} privacy condition against Alice. It requires that the rate of additional information that Alice learns about Bob's input and output other than what can be inferred from her own input and output is negligible. Similarly, \eqref{eqn:asymptotic_3} is the privacy condition against Bob.
\begin{definition}
$(q_{XY},q_{Z_1Z_2|XY})$ is \emph{perfectly securely computable} in $r$ rounds with privacy against both users if there exists a protocol $\Pi_n$ with $n=1$ such that \eqref{eqn:asymptotic_1}-\eqref{eqn:asymptotic_3} are satisfied with $\epsilon=0$.
\end{definition}
\begin{definition}
An $(n,R_0,R_{12},R_{21})$ \emph{protocol} is a protocol $\Pi_n$ such that the alphabet of $W$ is $\mathcal{W}=[1:2^{nR_0}]$ and
\begin{align*}
R_{12}&=\frac{1}{n}\sum_{i:\emph{odd}}\log|\mathcal{M}_i|,\\
R_{21}&=\frac{1}{n}\sum_{i:\emph{even}}\log|\mathcal{M}_i|,
\end{align*}
where $\mathcal{M}_i$ is the alphabet of $M_i$, $i\in[1:r]$.
\end{definition}
For a given function $(q_{XY},q_{Z_1Z_2|XY})$, a rate triple $(R_0,R_{12},R_{21})$ 
is said to be \emph{achievable} in $r$ rounds with privacy against both users if there exists a sequence 
of $(n,R_0,R_{12},R_{21})$ protocols such that, for every $\epsilon>0$, there exists a large enough $n$ 
satisfying \eqref{eqn:asymptotic_1}-\eqref{eqn:asymptotic_3}. 
The \emph{rate region} $\mathcal{R}^{AB-\text{pvt}}_A(r)$ with privacy against both users is the closure of all the achievable rate triples $(R_0,R_{12},R_{21})$. The subscript $A$ in $\mathcal{R}^{AB-\text{pvt}}_A(r)$ denotes that Alice starts the communication and the superscript {AB-\text{pvt}} denotes that privacy is against Alice and Bob. 
$\mathcal{R}^{AB-\text{pvt}}_B(r)$ can be defined in a similar fashion for the scenario when Bob starts the communication. We are interested in the region $\mathcal{R}^{AB-\text{pvt}}(r):=\mathcal{R}^{AB-\text{pvt}}_A(r)\bigcup \mathcal{R}^{AB-\text{pvt}}_B(r)$ {\color{black}and the minimum sum-rate $R^{AB-\text{pvt}}_{\text{sum}}(r,R_0):=\min\{R_{12}+R_{21}: (R_0,R_{12},R_{21})\in \mathcal{R}^{AB-\text{pvt}}(r)\}$}. Let $\mathcal{R}^{AB-\text{pvt}}:=\bigcup_{r=1}^\infty \mathcal{R}^{AB-\text{pvt}}(r)$. Notice that the above definitions are for the case when privacy is required against both users. $\mathcal{R}_A^{A-\text{pvt}}(r)$, $\mathcal{R}_A^{B-\text{pvt}}(r)$ and so on can also be defined in a similar fashion for the cases when privacy is required only against Alice and privacy is required only against Bob, respectively. For example, for the case when privacy is required only against Alice, the definitions will require \eqref{eqn:asymptotic_1}-\eqref{eqn:asymptotic_2} only and not \eqref{eqn:asymptotic_3}.
\section{Feasibility and Rate Region}\label{section:results}

We present single-letter characterizations of securely computable randomized functions and the rate regions. Detailed proofs \iftoggle{paper}{can be found in an extended version of this paper.}{can be found in Appendix~\ref{appendix:proofs_omitted}.} {\color{black}The following theorem is about feasibility. Part~$(i)$ states that asymptotic secure computability (in $r$ rounds) of a function  implies one-shot (i.e., $n=1$) perfectly secure computability (in the same number of rounds). Part~$(ii)$ shows that asymptotic secure computability depends on the input distribution $q_{XY}$ only through its support, $\text{supp}(q_{XY}):=\{(x,y): q_{XY}(x,y)>0\}$. In fact, asymptotic secure computability of a  function $(q_{XY},q_{Z_1Z_2|XY})$ is preserved even with another input distribution $\tilde{q}_{XY}$ whose support is a subset of $\text{supp}(q_{XY})$.}

\begin{theorem}
  \label{Feasibility_AB}
 % \begin{enumerate}[(i)]
   (i) $(q_{XY},q_{Z_1Z_2|XY})$ is asymptotically securely computable with privacy against both users using an $r$-round protocol in which
   Alice starts the communication if and only if there exists a conditional p.m.f. $p(u_{[1:r]}|x,y,z_1,z_2)$ satisfying
   \begin{align}
     & U_i-(U_{[1:i-1]},X)-Y, \emph{if}\  i\  \emph{is} \ \emph{odd},\label{Eq_AB_Markov1} \\
     & U_i-(U_{[1:i-1]},Y)-X, \emph{if} \ i \ \emph{is} \ \emph{even},\label{Eq_AB_Markov2}\\
     & Z_1-(U_{[1:r]},X)-(Y,Z_2),\label{Eq_AB_Markov_Decod1} \\
   & Z_2-(U_{[1:r]},Y)-(X,Z_1), \label{Eq_AB_Markov_Decod2}\\
     & U_{[1:r]}-(X,Z_1)-(Y,Z_2),\label{Eq_AB_Markov_Secr1}\\
       &U_{[1:r]}-(Y,Z_2)-(X,Z_1),\label{Eq_AB_Markov_Secr2}
     \end{align}  
     $|\mathcal{U}_1|\leq |\mathcal{X}||\mathcal{Y}||\mathcal{Z}_1||\mathcal{Z}_2|+1$ and
      $|\mathcal{U}_i|\leq
 |\mathcal{X}||\mathcal{Y}||\mathcal{Z}_1||\mathcal{Z}_2|
 \prod_{j=1}^{i-1}|\mathcal{U}_j|+1
$, $\forall i>1$.

{\color{black}\noindent (ii)  If a function $(q_{XY},q_{Z_1Z_2|XY})$ is asymptotically securely computable  with privacy against both users using an $r$-round protocol, then  $(\tilde{q}_{XY},q_{Z_1Z_2|XY})$, where $\emph{supp}(\tilde{q}_{XY})\subseteq \emph{supp}(q_{XY})$, is also asymptotically securely computable with privacy against both users using an $r$-round protocol.}
\end{theorem}

\begin{remark}\label{remark:feasibility}
Notice that Alice can generate common randomness by sending some {\color{black}of her private randomness} along with the message in the first round. So, the presence or absence of common randomness should not affect the secure 
computability of a function $(q_{XY},q_{Z_1Z_2|XY})$. As expected, the condition in part~$(i)$ does not depend on common randomness.
\end{remark} 

{\color{black}\begin{remark}\label{remark:eavesdropper}

Part~$(i)$ of Theorem~\ref{Feasibility_AB} shows that, for our problem, asymptotic secure computability is equivalent to one-shot perfect secure computability. It is interesting to note that this is not the case for all secure function computation problems. Consider the problem of function computation with privacy against an eavesdropper \cite{TyagiNG11}. Tyagi et al. \cite{TyagiNG11} considers the asymptotic setting where a group of users with correlated inputs interact noiselessly to compute a common function. The privacy requirement is that the amount of information that an eavesdropper learns about the function from the communication vanishes asymptotically. \cite[Theorem~2]{TyagiNG11} states that a function $g$ is asymptotically securely computable by two users with privacy against an eavesdropper if $H(g(X,Y))<I(X;Y)$ (and only if $H(g(X,Y))\leq I(X;Y)$). In this setup, perfectly secure computability with privacy against an eavesdropper can be defined analogous to the asymptotic secure computability with privacy against an eavesdropper. Below, we give an example of a function which is computable with asymptotic security (with privacy from an eavesdropper) but not with perfect security. Furthermore, unlike part $(ii)$ of Theorem~\ref{Feasibility_AB}, asymptotic secure computability with privacy against an eavesdropper depends on the input distribution $q_{XY}$ not just through its support $\text{supp}(q_{XY})$ \cite[Theorem~2]{TyagiNG11}.

\begin{example}
Consider a doubly symmetric binary source $\text{DSBS}(a)$ with joint distribution $q_{XY}(x,y)=0.5(1-a)\mathbbm{1}_{\{x=y\}}+0.5a\mathbbm{1}_{\{x\neq y\}}$, $a\in [0,0.5]$ and $x,y\in \{0,1\}$. Let the function to be computed by both users is $g(x,y)=x\oplus y$, where `$\oplus$' is addition modulo-2. Choose $a\in (0,0.5]$ s.t. $h(a)<1-h(a)$ (where $h(\cdot)$ denotes the binary entropy function), so that $g$ is asymptotically securely computable with privacy against the eavesdropper (by \cite[Theorem 2]{TyagiNG11}). We show that there does not exist a protocol that perfectly securely computes $g$ with privacy from an eavesdropper. If $g$ is perfectly securely computable with privacy against an eavesdropper, then there exists some $r$ and a conditional p.m.f. $p(u_{[1:r]}|x,y)$ satisfying \eqref{Eq_AB_Markov1}-\eqref{Eq_AB_Markov_Decod2} with $Z_1=Z_2=G:=g(X,Y)$ (for correctness), and $I(G;U_{[1:r]})=0$ (for privacy against the eavesdropper). For simplicity, we write $U$ for $U_{[1:r]}$. Suppose there exists a conditional p.m.f. satisfying the above conditions. In particular, we have
\begin{align}
(X\oplus Y)-(U,X)-Y\label{eqn:eaves1},\\
(X\oplus Y)-(U,Y)-X\label{eqn:eaves2},\\
I(X\oplus Y;U)=0\label{eqn:eaves3}.
\end{align}
\eqref{eqn:eaves1} implies that $I(X\oplus Y;Y|U,X)=0$, which in turn implies that $H(Y|U,X)=0$ (i.e., $Y$ is a function of $(U,X)$) since $H(Y|U,X,X\oplus Y)=0$. Similarly, \eqref{eqn:eaves2} implies that $H(X|U,Y)=0$, i.e., $X$ is a function of $(U,Y)$. Now if $p(u,x,y)>0$, then we claim that $p(u)=p(u,x,y)+p(u,\bar{x},\bar{y})$ ($\bar{x}$ denotes the compliment of $x$, i.e., $\bar{x}=1-x$). To see this, if $p(u,x,y)>0$, note that $p(u,x,\bar{y})=0$ since $Y$ is a function of $(U,X)$. Similarly, $p(u,\bar{x},y)=0$ since $X$ is a function of $(U,Y)$. Hence, when $p(u)>0$, since there exists $x,y$ s.t. $p(u,x,y)>0$, we can write $p(u)=p(u,x,y)+p(u,\bar{x},\bar{y})$. Hence, $X\oplus Y$ is a function of $U$ as $x\oplus y=\bar{x}\oplus \bar{y}$, $\forall x,y\in\{0,1\}$. This is a contradiction to \eqref{eqn:eaves3} as $a\in (0,0.5]$. Therefore, $g$ is not perfectly securely computable with privacy against an eavesdropper.
\end{example}
\end{remark}
}
\begin{remark}\label{remark:equv}
 \RRB{Let us call the functions that are asymptotically securely computable in $r$ rounds, for some $r>0$, as {\em asymptotically securely computable} functions. Note that part $(i)$ of Theorem~\ref{Feasibility_AB} does not give a computable characterization of asymptotically securely computable functions since the number of auxiliary random variables to consider is unbounded.} This problem, which was partially addressed in \cite{MajiPR12,DataP17} for full support input distributions, remains open.
\end{remark}

{\color{black}\begin{proof}[Proof sketch of Theorem~\ref{Feasibility_AB}] 
We give a proof sketch here. A detailed proof can be found in Appendix~\ref{appendix:proofs_omitted}.   For part $(i)$, it is trivial to see the `if' part since \eqref{Eq_AB_Markov1}-\eqref{Eq_AB_Markov_Secr2} define an $r$-round perfectly secure protocol of blocklength one, i.e., the protocol satisfies \eqref{eqn:asymptotic_1}-\eqref{eqn:asymptotic_3} with $n=1$ and $\epsilon=0$. For the `only if' part, we first single-letterize the privacy constraints \eqref{eqn:asymptotic_2} and \eqref{eqn:asymptotic_3}. We then single-letterize \eqref{eqn:asymptotic_1} and the Markov chains that are implied by the joint distribution in \eqref{eqn:asymptoticdist} along the lines of two-way source coding of Kaspi \cite{Kaspi85}, interactive (deterministic) function computation of Ma and Ishwar \cite{MaI11}, and channel simulation of Yassaee et al. \cite{YassaeeGA15}.  Then by using the continuity of mutual information and total variation distance in the probability simplex, we show that, if a function is computable with asymptotic security, it is also computable with perfect security. For part $(ii)$, we show that a protocol which securely computes $(q_{XY},q_{Z_1Z_2|XY})$ will also securely compute the function $(\tilde{q}_{XY},q_{Z_1Z_2|XY})$, where $\text{supp}(\tilde{q}_{XY})\subseteq \text{supp}(q_{XY})$.
\end{proof}
}

{\color{black}Next theorem characterizes the rate region $\mathcal{R}^{AB-\text{pvt}}_A(r)$ }
\begin{theorem}\label{Thm_Rate_Region_AB}
If a function $(q_{XY},q_{Z_1Z_2|XY})$ is asymptotically securely computable with privacy against both users, then $\mathcal{R}^{AB-\emph{pvt}}_A(r)$ is given by the set of all non-negative rate triples $(R_0,R_{12},R_{21})$ such that
  \begin{align}
   R_{12} &\geq I(X;Z_2|Y),\label{eqn_results_Thm_Rate_Region_AB_1}\\
   R_{21} &\geq I(Y;Z_1|X),\label{eqn_results_Thm_Rate_Region_AB_2}\\
  R_{0} + R_{12} &\geq  I(X;Z_2|Y) + I(U_1;Z_1,Z_2|X,Y),\label{eqn_results_Thm_Rate_Region_AB_3}\\
    R_{0} + R_{12} + R_{21} &\geq I(X;Z_2|Y) + I(Y;Z_1|X) + \nonumber\\
    &\hspace{12pt}I(Z_1;Z_2|X,Y)\label{eqn_results_Thm_Rate_Region_AB_4},
   \end{align}
   for some conditional p.m.f. $p(u_{[1:r]}|x,y,z_1,z_2)$ satisfying \eqref{Eq_AB_Markov1}-\eqref{Eq_AB_Markov_Secr2}, $|\mathcal{U}_1|\leq |\mathcal{X}||\mathcal{Y}||\mathcal{Z}_1||\mathcal{Z}_2|+5$ and
$
|\mathcal{U}_i|\leq
 |\mathcal{X}||\mathcal{Y}||\mathcal{Z}_1||\mathcal{Z}_2|
 \prod_{j=1}^{i-1}|\mathcal{U}_j|+4
$, $\forall i>1$.
    % \end{enumerate}
 \end{theorem}

\begin{remark}
Inequality \eqref{eqn_results_Thm_Rate_Region_AB_3} on $R_0+R_{12}$ makes the rate region $\mathcal{R}^{AB-\text{pvt}}_A(r)$ possibly asymmetric. This is, in fact, due to the assumption that Alice starts the communication. 
This is similar to the possible asymmetry of the rate region observed in channel simulation \cite[Theorem 1]{YassaeeGA15}. 
\end{remark}

\begin{remark}\label{remark:wang}
Substituting $X=Y=\emptyset$ in part~$(i)$ of Theorem~\ref{Feasibility_AB} recovers a result of \cite{WangI11} which states that a distribution $q_{Z_1,Z_2}$ is securely computable (i.e., securely sampleable as there are no inputs here) if and only if $C(Z_1;Z_2)=I(Z_1;Z_2)$, where $C(Z_1;Z_2):=\underset{Z_1-W-Z_2}{\min}I(Z_1,Z_2;W)$ is Wyner common information \cite{Wyner75}. To see this, note that $C(Z_1;Z_2)=I(Z_1;Z_2)+\underset{Z_1-W-Z_2}{\min}\left(I(Z_1;W|Z_2)+I(Z_2;W|Z_1)\right).$ Furthermore, when $R_0=0$, \iftoggle{paper}{it can be shown using part $(ii)$ of Theorem~\ref{Thm_Rate_Region_AB} and \eqref{eqn:results:simplification3} that the optimal sum-rate is $R_{12}+R_{21}=C(Z_1;Z_2)=I(Z_1;Z_2)$.}{Theorem~\ref{Thm_Rate_Region_AB} implies that the optimal sum-rate is $R_{12}+R_{21}=C(Z_1;Z_2)=I(Z_1;Z_2)$. {\color{black}This follows from \eqref{eqn:results:simplification3} (proved later) and the fact that $I(U_1; Z_1,Z_2|X,Y)\leq I(U_{[1:r]}; Z_1,Z_2|X,Y)$.}}
\end{remark}
 {\color{black}\begin{remark}
For a function $(q_{XY},q_{Z_1Z_2|XY})$, which is asymptotically securely computable with privacy against both users using a one round protocol in which Alice starts the communication, \eqref{eqn_results_Thm_Rate_Region_AB_2} purports to give a lower bound on the rate of communication from Bob to Alice. However, note that this lower bound $I(Y;Z_1|X)$ is in fact zero. To see this, notice that if a function $(q_{XY},q_{Z_1Z_2|XY})$ is asymptotically securely computable with privacy against both users using a 1-round protocol in which Alice starts the communication, it follows from $U-X-Y$ and  \eqref{eqn:results:simplification2} (proved later) that $I(Y;Z_1|X)=0$.
 \end{remark}
 }
{\color{black}\begin{proof}[Proof sketch of Theorem~\ref{Thm_Rate_Region_AB}] 
We give a proof sketch here. A detailed proof can be found in Appendix~\ref{appendix:proofs_omitted}.  Our proof of achievability is along similar lines as the achievability proof of channel simulation \cite[Theorem~1]{YassaeeGA15}. We modify this proof to give a protocol which also accounts for privacy. \RRB{Specifically, we show how the Markov chains \eqref{Eq_AB_Markov_Secr1} and \eqref{Eq_AB_Markov_Secr2} can be turned into privacy constraints~\eqref{eqn:asymptotic_2} and \eqref{eqn:asymptotic_3} retaining the correctness~\eqref{eqn:asymptotic_1}}. For the converse, we first single-letterize the privacy constraints \eqref{eqn:asymptotic_2} and \eqref{eqn:asymptotic_3}. The rest of the converse is in the spirit of two-way source coding of Kaspi \cite{Kaspi85}, interactive (deterministic) function computation of Ma and Ishwar \cite{MaI11}, and channel simulation of Yassaee et al. \cite{YassaeeGA15}. \RRB{Note that such a single-letterization of privacy constraints preserving the correctness may not always be possible for any secure function computation problem (see, e.g., Remark~\ref{remark:eavesdropper}).} This gives a rate region defined by the set of non-negative rate triples $(R_0,R_{12},R_{21})$ such that
\begin{align}
   R_{12} &\geq I(X;U_{[1:r]}|Y),\label{eqn:ach1}\\
   R_{21} &\geq I(Y;U_{[1:r]}|X),\label{eqn:ach2}\\
  R_{0} + R_{12} &\geq I(X;U_{[1:r]}|Y) + I(U_1;Z_1,Z_2|X,Y),\label{eqn:ach3}\\
    R_{0} + R_{12} + R_{21} &\geq I(X;U_{[1:r]}|Y) + I(Y;U_{[1:r]}|X)\nonumber\\
    &\hspace{12pt} + I(U_{[1:r]}; Z_1,Z_2|X,Y)\label{eqn:ach4},
   \end{align}
for conditional p.m.f. $p(u_{[1:r]},z_1,z_2|x,y)$ satisfying \eqref{Eq_AB_Markov1}-\eqref{Eq_AB_Markov_Secr2}. 
 Notice that constraints \eqref{eqn:ach1}-\eqref{eqn:ach4} appear
in channel simulation \cite[Theorem 1]{YassaeeGA15} also, where the conditional p.m.f. $p(u_{[1:r]},z_1,z_2|x,y)$ satisfies \eqref{Eq_AB_Markov1}-\eqref{Eq_AB_Markov_Decod2}. The above region reduces to the form mentioned in Theorem~\ref{Thm_Rate_Region_AB} because of a simplification 
possible here due to the additional privacy constraints \eqref{Eq_AB_Markov_Secr1}-\eqref{Eq_AB_Markov_Secr2}, which gives us (as shown in \iftoggle{paper}{the Appendix}{the detailed proof of Theorem \ref{Thm_Rate_Region_AB} in Appendix~\ref{appendix:proofs_omitted}})
\begin{align}
I(X;U_{[1:r]}|Y)&=I(X;Z_2|Y),\label{eqn:results:simplification1}\\
I(Y;U_{[1:r]}|X)&=I(Y;Z_1|X),\label{eqn:results:simplification2}\\
I(U_{[1:r]}; Z_1,Z_2|X,Y)&=I(Z_1;Z_2|X,Y)\label{eqn:results:simplification3}.
\end{align} 
\end{proof}
}
\begin{remark}
\RRB{Note that the lower bounds in \eqref{eqn_results_Thm_Rate_Region_AB_1} and \eqref{eqn_results_Thm_Rate_Region_AB_2} are in fact the \emph{cut-set} lower bounds for {\em non-private} computation {(see Appendix~\ref{cutset_discussion} for details)}. Thus, Theorem~\ref{Thm_Rate_Region_AB} implies that for sufficiently large common randomness rates $R_0$, the cut-set bounds are met for securely computable functions. The intuition is as follows: from \eqref{eqn:ach1}-\eqref{eqn:ach4} in the proof of Theorem~\ref{Thm_Rate_Region_AB}, $R_{12}=I(X;U_{[1:r]}|Y)$ and $R_{21}=I(Y;U_{[1:r]}|X)$ are in the rate-region for sufficiently large $R_0$. Note that $I(X;U_{[1:r]}|Y)$ measures the rate of information Bob learns about Alice's input during the protocol. When privacy against Bob is required, the rate of information that Bob learns about Alice's input during the protocol must be equal to the rate of information about Alice's input that can be inferred just from his output, i.e., $I(X;U_{[1:r]}|Y)=I(X;Z_2|Y)$. To see this, first note that $I(X;U_{[1:r]}|Y)=I(X;U_{[1:r]},Z_2|Y)=I(X;Z_2|Y)+I(X;U_{[1:r]}|Y,Z_2)$, where the first equality follows from the correctness. Now privacy against Bob implies that $I(X,Z_1;U_{[1:r]}|Y,Z_2)=0$, which in turn implies that $I(X;U_{[1:r]}|Y,Z_2)=0$. Similarly, when privacy against Alice is required, we have $I(Y;U_{[1:r]}|X)=I(Y;Z_1|X)$.}
\end{remark}

Let us denote the minimum number of rounds required for secure computation by $r_{\text{min}}$, i.e., the smallest $r$ such that there exists auxiliary random variables $U_{[1:r]}$ which makes the function $(q_{XY},q_{Z_1Z_2|XY})$ one-shot perfectly securely computable with either Alice or Bob starting the communication. \RRB{For deterministic functions, it is known that $r_{\text{min}}<2\min\{|\mathcal{X}|,|\mathcal{Y}|\}$~\cite{Kushelvitz92}. No such bound is available for randomized functions in general~\cite{MajiPR12,DataP17}.} Note that Theorem \ref{Thm_Rate_Region_AB} is for any fixed number of rounds $r$. \RRB{As we remarked above, for securely computable functions, the cut-set lower bounds for non-private computation holds with equality for sufficiently large common randomness rate for every $r\geq r_{\text{min}}$. The following corollary shows that the optimal trade-off between the communication and common randomness rates is achieved in at most $r_{\text{min}}+1$ rounds.} %gives the region $\mathcal{R}^{AB-\text{pvt}}$.
Notice that the expression for $\mathcal{R}^{AB-\text{pvt}}$ below does not involve any auxiliary random variables.
\begin{corollary}\label{corollary_1}
If $(q_{XY},q_{Z_1Z_2|XY})$ is asymptotically securely computable with privacy against both users,
then $\mathcal{R}^{AB-\emph{pvt}}$ is given by the set of all non-negative rate triples
$(R_0,R_{12},R_{21})$ such that
\begin{align}
 R_{12} &\geq I(X;Z_2|Y),\label{eqn:optimal_region_1}\\
   R_{21} &\geq I(Y;Z_1|X),\label{eqn:optimal_region_2}\\
    R_{0} + R_{12} + R_{21} &\geq I(X;Z_2|Y) + I(Y;Z_1|X)\nonumber\\
    &\hspace{12pt} + I(Z_1;Z_2|X,Y)\label{eqn:optimal_region_3}.
\end{align} 
Furthermore, $\mathcal{R}^{AB-\emph{pvt}}(r_{\emph{min}}+1)=\mathcal{R}^{AB-\emph{pvt}}$.
\end{corollary} 
\begin{proof}[\textbf{Proof of Corollary \ref{corollary_1}}]
\RRB{It suffices to prove that $\mathcal{R}^{AB-\text{pvt}}(r_{\text{min}}+1)=\mathcal{R}_{\text{opt}}$, where $\mathcal{R}_{\text{opt}}$ is  defined to be the set of all non-negative rate triples $(R_0,R_{12},R_{21})$ such that \eqref{eqn:optimal_region_1}-\eqref{eqn:optimal_region_3} are satisfied. From Theorem~\ref{Thm_Rate_Region_AB} it is easy to see that $\mathcal{R}^{AB-\text{pvt}}(r_{\text{min}}+1)\subseteq\mathcal{R}_{\text{opt}}$. For the other direction, take a point $(R_0,R_{12},R_{21})\in\mathcal{R}_{\text{opt}}$. Without loss of generality, suppose that $r_{\text{min}}$ occurs when Alice starts the communication. Then, by Theorem~\ref{Thm_Rate_Region_AB}, there exists random variables $U_{[1:r_{\text{min}}]}$ with conditional p.m.f. $p(u_{[1:r_{\text{min}}]}|x,y,z_1,z_2)$ satisfying \eqref{Eq_AB_Markov1}-\eqref{Eq_AB_Markov_Secr2}. We find new random variables $U^\prime_{[1:r_{\text{min}}+1]}$ so that $(R_0,R_{12},R_{21})$ becomes a point in $\mathcal{R}_B^{AB-\text{pvt}}(r_{\text{min}}+1)$. Define $U_1^\prime=\emptyset$ and $U_i^\prime=U_{i-1}$ for $i>1$. This gives us that $(R_0,R_{12},R_{21})\in\mathcal{R}_B^{AB-\text{pvt}}(r_{\text{min}}+1)$. Hence $(R_0,R_{12},R_{21})\in\mathcal{R}^{AB-\text{pvt}}(r_{\text{min}}+1)$.}
\end{proof}

%For computing randomized function $(q_{XY},q_{Z_1Z_2|XY})$ \emph{without any privacy} guarantees, the \emph{cut-set} lower bounds can be shown to be \iftoggle{paper}{\ignorespaces}{(see Appendix~\ref{cutset_discussion} for details)} $R_{12}\geq I(X;Z_2|Y), R_{21}\geq I(Y;Z_1|X)$. 
%{\color{black}Part~$(i)$ of the following theorem states that if a function is securely computable with privacy against both users, then these cut-set lower bounds (for function computation without any privacy requirement) are achievable.}
As mentioned earlier, if a function is securely computable with privacy against both users, then the cut-set lower bounds (for function computation without any privacy requirement) are achievable. The converse is not true in general\footnote{\color{black}To see this, suppose $X=Y=\emptyset$, then any function $q_{Z_1Z_2}$ can be computed by using common randomness alone (see Wyner common information problem \cite{Wyner75}), i.e., by meeting the cut-set lower bounds which in this case are zero, $I(X;Z_2|Y)=0=I(Y;Z_1|X)$. Assume that $q_{Z_1Z_2}$ is such that $C(Z_1;Z_2)\neq I(Z_1;Z_2)$. Then $q_{Z_1Z_2}$ is not securely computable in view of Remark~\ref{remark:wang}.}. However, part $(ii)$ of the theorem below states that a converse holds for a class of functions including deterministic functions\footnote{\color{black}The characterization of all the functions for which the converse holds remains open.}.
 Let the rate region $\mathcal{R}_A^{\text{No-privacy}}(r)$ be defined analogous to $\mathcal{R}^{AB-\text{pvt}}_A(r)$ (except that only correctness condition~\eqref{eqn:asymptotic_1} is required).

\begin{theorem}\label{cutset}
{\color{black} (i) If a function $(q_{XY},q_{Z_1Z_2|XY})$ is securely computable in $r$ rounds with privacy against both users, then there exists common randomness rate $R_0$ such that $\big(R_0,I(X;Z_2|Y),I(Y;Z_1|X)\big)\in \mathcal{R}_A^{\emph{No-privacy}}(r)$. Furthermore, a rate of $R_0=I(Z_1;Z_2|X,Y)$ suffices for this.

\noindent (ii) Suppose the function $(q_{XY},q_{Z_1Z_2|XY})$ is such that $H(Z_1|X,Y,Z_2)=0$ \RRB{and} $H(Z_2|X,Y,Z_1)=0$ (e.g., a deterministic function). If there exists $R_0$ such that $\big(R_0,I(X;Z_2|Y),I(Y;Z_1|X)\big)\in \mathcal{R}_A^{\emph{No-privacy}}(r)$, then the function is securely computable in $r$ rounds with privacy against both users.}
\end{theorem}
We prove Theorem~\ref{cutset} in \iftoggle{paper}{the Appendix}{Appendix~\ref{appendix:proofs_omitted}}. {\color{black}Part $(i)$ will follow from Theorem~\ref{Thm_Rate_Region_AB}. We prove part $(ii)$ by showing that, for the class of functions mentioned in Theorem~\ref{cutset}}, any protocol for computation without privacy that meets the cut-set bounds must satisfy the privacy conditions as well. 
\subsection*{When privacy is required against only one user:}     
 {\color{black} Note that when privacy is required only against Alice, any function $(q_{XY},q_{Z_1Z_2|XY})$ can be securely computed using a 2-round protocol in which Alice starts the communication.}
Alice can transmit her input to Bob who can compute the function
according to $q_{Z_1Z_2|XY}$, and send $Z_1$ back to Alice. Part $(i)$ of the following theorem considers the feasibility of $1$ round protocols. Similar to part $(i)$ of Theorem~\ref{Thm_Rate_Region_AB}, it states that asymptotic secure computability implies one-shot perfectly secure computability. Part $(ii)$ characterizes the rate region for an arbitrary number of rounds $r$.
 \begin{theorem}
  \label{Thm_Rate_Region_A}
  %\begin{enumerate}[(i)]
  (i) $(q_{XY},q_{Z_1Z_2|XY})$ is asymptotically securely computable  with privacy only against Alice using a 1-round protocol in which Alice starts the communication if and only if there exists 
  a conditional p.m.f. $p(u_1|x,y,z_1,z_2)$ satisfying $(a)$ $U_1-X-Y$, $(b)$ $Z_1-(U_1,X)-(Y,Z_2)$, $(c)$ $Z_2-(U_1,Y)-(X,Z_1)$, $(d)$~$U_1-(X,Z_1)-(Y,Z_2)$. {\color{black}Furthermore, if a function $(q_{XY},q_{Z_1Z_2|XY})$ is asymptotically securely computable  with privacy only against Alice using a 1-round protocol, then  $(\tilde{q}_{XY},q_{Z_1Z_2|XY})$, where $\emph{supp}(\tilde{q}_{XY})\subseteq \emph{supp}(q_{XY})$, is also asymptotically securely computable  with privacy only against Alice using a 1-round protocol. }
   
\noindent(ii) $\mathcal{R}^{A-\emph{pvt}}_A(r)$ is given by the set of all non-negative rate triples $(R_0,R_{12},R_{21})$ such that
  \begin{align*}
   R_{12} &\geq I(X;U_{[1:r]}|Y),\\
   R_{21} &\geq I(Y;Z_1|X),\\
  R_{0} + R_{12} &\geq  I(X;U_{[1:r]}|Y) + I(U_1;Z_1|X,Y),\\
    R_{0} + R_{12} + R_{21} &\geq I(X;U_{[1:r]}|Y) + I(Y;Z_1|X) \nonumber\\
    &\hspace{12pt}+ I(U_{[1:r]}; Z_1|X,Y),
   \end{align*}
   for some conditional p.m.f. $p(u_{[1:r]}|x,y,z_1,z_2)$ satisfying \eqref{Eq_AB_Markov1}-\eqref{Eq_AB_Markov_Decod2}, \eqref{Eq_AB_Markov_Secr1}, \RRB{$|\mathcal{U}_1|\leq |\mathcal{X}||\mathcal{Y}||\mathcal{Z}_1||\mathcal{Z}_2|+4$ and
$
|\mathcal{U}_i|\leq
 |\mathcal{X}||\mathcal{Y}||\mathcal{Z}_1||\mathcal{Z}_2|
 \prod_{j=1}^{i-1}|\mathcal{U}_j|+3
$, $\forall i>1$.}
  % \end{enumerate}
 \end{theorem}

\iftoggle{paper}
{ Note that similar cardinality bounds on auxiliary random variables as in Theorem~\ref{Thm_Rate_Region_AB} and similar statements as in Remark~\ref{remark:feasibility} hold true for Theorem~\ref{Thm_Rate_Region_A} also. A theorem similar to Theorem~\ref{Thm_Rate_Region_A} holds for the case when privacy is required only against Bob and it can be found in the extended version.} 
 {When privacy is required only against Bob, any {\color{black}function} $(q_{XY},q_{Z_1Z_2|XY})$ is securely computable in at most $3$ rounds with Alice starting the communication. To see this, note that Alice may transmit nothing in the first round, Bob can transmit his input to Alice in the second round.  {\color{black}She} can {\color{black}then} compute the function
according to $q_{Z_1Z_2|XY}$, and send $Z_2$ back to Bob in the third round. Part $(i)$ of the following theorem considers the feasibility of $1$ and $2$ round protocols. Similar to part $(i)$ of Theorems~\ref{Thm_Rate_Region_AB} and \ref{Thm_Rate_Region_A}, it states that asymptotic secure computability implies perfectly secure computability. Part $(ii)$ characterizes the rate region for an arbitrary number of rounds $r$.
 \begin{theorem}
  \label{Thm_Rate_Region_B}
 % \begin{enumerate}[(i)]
   (i) $(q_{XY},q_{Z_1Z_2|XY})$ is asymptotically securely computable with privacy only against Bob using an $r$-round protocol in which 
    Alice starts the communication if and only if there exists a conditional p.m.f. $p(u_{[1:r]}|x,y,z_1,z_2)$ 
   satisfying \eqref{Eq_AB_Markov1}-\eqref{Eq_AB_Markov_Decod2} and \eqref{Eq_AB_Markov_Secr2}, for $r=1,2$. {\color{black}Furthermore, if a function $(q_{XY},q_{Z_1Z_2|XY})$ is asymptotically securely computable  with privacy only against Bob using a 1(2, resp.)-round protocol, then  $(\tilde{q}_{XY},q_{Z_1Z_2|XY})$, where $\emph{supp}(\tilde{q}_{XY})\subseteq \emph{supp}(q_{XY})$, is also asymptotically securely computable  with privacy only against Bob using a 1(2, resp.)-round protocol. }
   
  \noindent(ii) $\mathcal{R}^{B-\emph{pvt}}_A(r)$ is given by the set of all non-negative rate triples $(R_0,R_{12},R_{21})$ such that
  \begin{align*}
   R_{12} &\geq I(X;Z_2|Y),\\
   R_{21} &\geq I(Y;U_{[1:r]}|X),\\
  R_{0} + R_{12} &\geq I(X;Z_2|Y) + I(U_1;Z_2|X,Y),\\
    R_{0} + R_{12} + R_{21} &\geq I(X;Z_2|Y) + I(Y;U_{[1:r]}|X)\nonumber\\
    &\hspace{12pt} + I(U_{[1:r]}; Z_2|X,Y),
   \end{align*}
   for some conditional pmf $p(u_{[1:r]}|x,y,z_1,z_2)$ satisfying \eqref{Eq_AB_Markov1}-\eqref{Eq_AB_Markov_Decod2}, \eqref{Eq_AB_Markov_Secr2}, \RRB{$|\mathcal{U}_1|\leq |\mathcal{X}||\mathcal{Y}||\mathcal{Z}_1||\mathcal{Z}_2|+4$ and
$
|\mathcal{U}_i|\leq
 |\mathcal{X}||\mathcal{Y}||\mathcal{Z}_1||\mathcal{Z}_2|
 \prod_{j=1}^{i-1}|\mathcal{U}_j|+3
$, $\forall i>1$.}
   %\end{enumerate}
 \end{theorem} 
 }

\section{Role of Interaction and Common Randomness}
\label{sec_role_CR_Intr}

{\color{black}When no privacy is required, any function can be computed in two rounds by exchanging the inputs.
However, as Ma and Ishwar~\cite{MaI11} have shown, there are functions for which more rounds of interaction can strictly improve the communication rate.
When privacy is required against both users,
 if a function is securely computable, depending on the function, by Theorem~\ref{Feasibility_AB}, a certain minimum number of rounds of interaction is required for secure computation. Recall that we defined $r_\text{min}$ to be the smallest $r$ such that there exists auxiliary random variables $U_{[1:r]}$ which makes the function $(q_{XY},q_{Z_1Z_2|XY})$  asymptotically (and perfectly) securely computable with either Alice or Bob starting the communication.} 
  The discussion on the role of interaction below will focus on whether increasing the number of rounds beyond $r_{\text{min}}$ helps to reduce the communication rate. On common randomness, it is clear from Remark~\ref{remark:feasibility} that its absence does not affect the secure computability of a function. The discussion on the role of common randomness will focus on its effect on the communication rate.
\subsection{Privacy required against both users} 
 \label{Sec_CR_Intr_AB}% 

 \RRB{It can be inferred from Corollary~\ref{corollary_1} that interaction will not help to enlarge the rate region beyond $r_{\text{min}}+1$ rounds when privacy is required against both users. Suppose a function $(q_{XY},q_{Z_1Z_2|XY})$ is securely computable with privacy against both users. Since  $I(U_1;Z_1,Z_2|X,Y)\leq I(U_{[1:r]};Z_1,Z_2|X,Y) = I(Z_1;Z_2|X,Y)$ (equality follows from  \eqref{eqn:results:simplification3}), we get from
 Theorem~\ref{Thm_Rate_Region_AB} that, for any $r\geq r_{\text{min}}$, the optimal sum-rate $R^{AB-\text{pvt}}_{\text{sum}}(r,R_0)$ is $I(X;Z_2|Y) + I(Y;Z_1|X)+ \big[ I(Z_1;Z_2|X,Y) - R_{0}\big]_{+}$, where $[x]_+=\max \{x,0\}$. Also note that if $I(Z_1;Z_2|X,Y)=0$, the characterization of the rate region does not involve the common randomness rate $R_0$\footnote{We also note that, for $R_0=0$, the proof of achievability can be carried out assuming there is no common randomness (not just that its rate is zero). }. From this discussion, we have the following two propositions.
 \begin{proposition}\label{prop:Int}
 Interaction does not improve the minimum sum-rate, i.e., { $R^{AB-\emph{pvt}}_{\emph{sum}}(r,R_0)=R^{AB-\emph{pvt}}_{\emph{sum}}(r_{\emph{min}},R_0)$, for $r\geq~r_{\emph{min}}, R_0\geq 0$.} 
Interaction does not help to enlarge the rate region when  (i) {$I(Z_1;Z_2|X,Y)=0$}, e.g., when the functions are deterministic,  or
(ii) the common randomness rate is large enough.
 \end{proposition}
}
 \RRB{The problem of finding necessary and sufficient conditions for the interaction to not enlarge the rate region remains open.}
 \RRB{Now, since $R^{AB-\text{pvt}}_{\text{sum}}(r,R_0)=R^{AB-\text{pvt}}_{\text{sum}}(r_{\text{min}},R_0)$, for $r\geq~r_{\text{min}}, R_0\geq 0$,  we will relax the first argument $r$ of $R^{AB-\text{pvt}}_{\text{sum}}(r,R_0)$ in the sequel and simply write $R^{AB-\text{pvt}}_{\text{sum}}(R_0)$ for $R^{AB-\text{pvt}}_{\text{sum}}(r,R_0)$, for any $r\geq r_{\text{min}}$.}
 \RRB{
\begin{proposition}\label{prop:CR}
For a function {$\left(q_{XY},q_{Z_1Z_2|XY}\right)$}, common randomness can improve the minimum sum-rate {(i.e., $R^{AB-\emph{pvt}}_{\emph{sum}}(R_0)< R^{AB-\emph{pvt}}_{\emph{sum}}(0),$ for all $R_0>0 $)} if and only if {$Z_1$} and {$Z_2$} are conditionally dependent given {$(X,Y)$}. Hence, for deterministic functions, common randomness does not reduce the sum-rate.
\end{proposition}
 }
\subsection{Privacy required against one user}\label{Sec_exm_sum_rate}
\RRB{When privacy is required against only one user, in contrast to the propositions above, both interaction and common randomness may help to improve the minimum sum-rate. We first show that one extra round of communication from the minimum number of rounds required for secure computability may strictly improve the minimum sum-rate.
\begin{proposition}\label{prop3}
When privacy is required against only one user (say, Bob), {\RRB{$r_{\emph{min}}+1$}} rounds may strictly improve the minimum sum-rate, i.e., there exists a function {$(q_{XY},q_{Z_1Z_2|XY})$} such that {$R^{B-\emph{pvt}}_{\emph{sum}}(r_\emph{min},R_0)>R^{B-\emph{pvt}}_{\emph{sum}}(r_\emph{min}+1,R_0),$ for all $R_0\geq 0 $.}
\end{proposition}
 }   
\RRB{ We prove this proposition through an example where both users compute a deterministic function of $(X,Y)$, and privacy against Bob alone is required. 
% \RRB{See Example~\ref{exam_1} in Appendix~\ref{App_opt_rate}.}
 \begin{example}{\label{exam_1}}
Let $Y$ be an $m$-length vector of uniform binary random variables, $Y=(Y_1, \ldots, Y_m)$, and $X$ consists of a uniform 
 binary random variable $V$ and  
 a random variable $J$ which is uniformly distributed on $[1:m]$, i.e.,  $X=(V,J)$.
 We assume that  $Y$,  $V$, and $J$ are independent. 
 Both users want to compute the function $Z=(J,V\wedge Y_J)$, where ``$\wedge$'' represents the binary AND function, with privacy only against Bob.
 In this example, it is easy to see that \iftoggle{paper}{\ignorespaces}{(see below)} $r_{\min}$ is 2 with Bob starting the communication. 
{Note that $R^{B-\text{pvt}}_{\text{sum}}(r,R_0)$ does not depend on $R_0$ for a deterministic function (from Theorem~\ref{Thm_Rate_Region_B}) and hence we can write $R^{B-\text{pvt}}_{\text{sum}}(r)$ for $R^{B-\text{pvt}}_{\text{sum}}(r,R_0)$ by relaxing $R_0$.} We show that the optimum sum-rate for two round protocols, {$R^{B-\text{pvt}}_{\text{sum}}(2)$} is $\log m + 1/2+m$ bits (all rates in the sequel are in bits). 
 Then we give a three round protocol, with Alice starting the communication, which has a sum-rate of $R^{B-\text{pvt}}_{\text{sum}}(3)=\log m +1/2+1$.
 We also show that  $ \log m+1/2+1$ is the minimum achievable sum-rate with any $r$ ($r\geq3$) rounds of protocol. Details can be found in Appendix~\ref{App_opt_rate}.
\end{example}
}
\RRB{It remains open whether there is an example where further interaction leads to a further improvement in the sum-rate, i.e., $R^{B-\text{pvt}}_{\text{sum}}(r_\text{min}+1,R_0)>R^{B-\text{pvt}}_{\text{sum}}(r_\text{min}+k,R_0)$, for some $k>1$.
\begin{proposition}\label{newprop}
When privacy is required against only one user (say, Bob), then common randomness may strictly improve the minimum sum-rate, i.e., there exists a function $(q_{XY},q_{Z_1Z_2|XY})$ such that $R^{B-\emph{pvt}}_{\emph{sum}}(r,R_0)<R^{B-\emph{pvt}}_{\emph{sum}}(r,0)$, for sufficiently large common randomness rate $R_0$ and $r\geq r_{\emph{min}}$. 
\end{proposition}
We prove this proposition through the following example.
\begin{example}\label{example:asympt}
Let $X$ be an $m$-length ($m>1$) vector of mutually independent and uniform binary random variables, $X=(X_1,\dots,X_m)$ and $Y=\emptyset$. Also, let $Z_1=Z_2=Z=(J,X_J)$, where $J$ is a random variable uniformly distributed on $[1:m]$. This defines a function $(q_{XY},q_{Z_1Z_2|XY})$. In this example, it is easy to that $r_{\text{min}}=1$, since Alice can sample random variable $J$ uniformly distributed on $[1:m]$ and transmit $(J,X_J)$ to Bob without violating privacy. When the common randomness rate $R_0$ is sufficiently large, from Theorem~\ref{Thm_Rate_Region_B}, the minimum sum-rate is given by $I(X;Z)=1$. So, $R^{B-\text{pvt}}_{\text{sum}}(r,R_0)=1$, for sufficiently large $R_0$ and $r\geq r_{\text{min}}$. We now argue that $R^{B-\text{pvt}}_{\text{sum}}(r,0)>1$. From Theorem~\ref{Thm_Rate_Region_B}, $R^{B-\text{pvt}}_{\text{sum}}(r,0)$ is given by
\begin{align}
&R^{B-\text{pvt}}_{\text{sum}}(r,0)\nonumber\\
&\hspace{12pt}=I(X;Z)+\min_{p(u_{[1:r]}|x,z):\  \eqref{Eq_AB_Markov1}-\eqref{Eq_AB_Markov_Decod2}, \eqref{Eq_AB_Markov_Secr2}\ \text{hold}}I(U_{[1:r]};Z|X)\nonumber\\
&\hspace{12pt}=1+\min_{p(u_{[1:r]}|x,z):\  \eqref{Eq_AB_Markov1}-\eqref{Eq_AB_Markov_Decod2}, \eqref{Eq_AB_Markov_Secr2}\ \text{hold}}I(U_{[1:r]};Z|X)\label{relax}.
\end{align}
To show $R^{B-\text{pvt}}_{\text{sum}}(r,0)>1$, we first show that $R^{B-\text{pvt}}_{\text{sum}}(r,0)=R^{B-\text{pvt}}_{\text{sum}}(1,0)$.
Note that $R^{B-\text{pvt}}_{\text{sum}}(1,0)$ is given by 
\begin{align}
R^{B-\text{pvt}}_{\text{sum}}(1,0)=1+\min_{\substack{p(u|x,z):\\  Z-U-X,\\ Z-(U,X)-Z,\\ U-Z-X}}I(U;Z|X)\label{eqn:29}.
\end{align}
By relaxing the constraints \eqref{Eq_AB_Markov1}-\eqref{Eq_AB_Markov_Decod1} in the optimization problem of \eqref{relax}, it is easy to see that $R^{B-\text{pvt}}_{\text{sum}}(r,0)\geq R^{B-\text{pvt}}_{\text{sum}}(1,0)$, for $r\geq 1$. We also have $R^{B-\text{pvt}}_{\text{sum}}(r,0)\leq R^{B-\text{pvt}}_{\text{sum}}(1,0), r>1$, since the minimum sum-rate can only improve with the number of rounds. Thus, we have $R^{B-\text{pvt}}_{\text{sum}}(r,0)= R^{B-\text{pvt}}_{\text{sum}}(1,0), r>1$. By considering $U=Z$ in \eqref{eqn:29}, it is clear that $R^{B-\text{pvt}}_{\text{sum}}(1,0)\leq 1+\log(m)=\log(2m)$. In Appendix~\ref{appendixexample}, we show that $R^{B-\text{pvt}}_{\text{sum}}(1,0)=\log(2m)>1$.
\end{example}
}

\RRB{We now give a condition under which neither interaction nor common randomness helps to enlarge the rate region.
\begin{proposition}\label{prop5}
If only one user computes and privacy is required against the other user, then neither interaction nor common randomness helps to enlarge the rate region.
\end{proposition}
\begin{remark}
We mention that when privacy is required only against Bob, Example~\ref{exam_1} (Example~\ref{example:asympt}, resp.) can be modified so that only Bob computes and interaction (common randomness, resp.) helps to reduce the communication rate. Thus, it is really the combination of computation by only one user and privacy requirement against the other user that matters in Proposition~\ref{prop5}.
\end{remark}
}
\begin{proof}[Proof of Proposition~\ref{prop5}]
Let us consider the case where only Bob computes (i.e., $Z_1=\emptyset$) and privacy against Alice is required. {\color{black}Note that every function $(q_{XY},q_{Z_2|XY})$ is securely computable in one round with privacy only against Alice\footnote{\label{foot:rand-not-trivial_priv-Alice}\color{black}Note that for a deterministic function, one-sided communication implies that privacy against Alice trivially holds because conditioned on the inputs, there is only one possible output. However, for a randomized function, one-sided communication need not imply privacy against Alice, because conditioned on the inputs, the message that Alice sends to Bob can influence the output. But privacy against Alice requires that, conditioned on the inputs, Alice's message must not influence Bob's output.} since Alice may send $X$ to Bob, and Bob, having $Y$ as his input, may output $Z_2$ according to $q_{Z_2|XY}$.} By substituting $Z_1= \emptyset$ in Theorem~\ref{Thm_Rate_Region_A}, it can be observed that
the only active constraint is    $R_{12} \geq I(X;U_{[1:r]}|Y)$. 
Further, let us consider  $R^*_{12} = \min I(X;U_{[1:r]}|Y)$, where the minimization is over 
conditional p.m.f.'s $p(u_{[1:r]}|x,y,z_2)$ satisfying \eqref{Eq_AB_Markov1}-\eqref{Eq_AB_Markov2}, and
\begin{align}
  U_{[1:r]}-X-(Y,Z_2) \label{Eq_One_Priv1},\\
  Z_2-(U_{[1:r]},Y)-X  \label{Eq_One_Priv2}.
\end{align}
Now let us consider $R'_{12} = \min I(X;U_{[1:r]}|Y)$, where the minimization is only under  
\eqref{Eq_One_Priv1} and \eqref{Eq_One_Priv2}.
Then $R^*_{12} \geq R'_{12}$. Further, it can be observed that $R'_{12}$ is the minimum rate 
achievable when $r=1$. So $R^*_{12} \leq R'_{12}$.
This shows that the rate region is given by 
\begin{align*}
 \{( R_{0}, R_{12}), R_{21}:  R_{0} \geq 0, R_{12} \geq I(X;U|Y), R_{21} \geq 0   \},
\end{align*}
for some conditional p.m.f. $p(u|x,y,z_2)$ satisfying  $U-X-(Y,Z_2)$ and $ Z_2 - (U,Y) - X$.
This shows that interaction and the presence of common randomness do not help   in this case.   
\end{proof}
\section{The Communication Cost of Security - An Example}
\label{section:cost}

 {\color{black}Recall that Theorem~\ref{cutset} states that if a function is securely computable with privacy against both users, then cut-set lower bounds (for function computation without any privacy) are achievable, i.e., when a function is securely computable with privacy against both users, there is no communication cost of security. In this section, we study the communication cost of security when privacy is required against only one user and compare it with the communication cost of non-private computation.   It turns out that, for a function which is securely computable with privacy against only one user, the requirement of privacy, in general, may lead to a larger optimal rate.} We consider, arguably, the simplest special case of two-user secure computation where only Bob computes a function using a single round of transmission from Alice.
We also compare our results with function computation without any privacy requirement. For the sake of comparison, below we state a result from \cite{YassaeeGA15}, 
specialized to our setting (i.e., one-way communication from Alice to Bob and only Bob computing), for the optimal rate required
when there is no privacy requirement. We denote the optimal rate by $R^{\text{No-Privacy}}$ when there is no common randomness. Note that \cite{YassaeeGA15} considered a more general model (where Alice and Bob
communicate back and forth for multiple rounds, and both of them may produce potentially different outputs),
and obtained a single-letter expression for the optimal rate, when Alice and Bob interact for arbitrary but finite number of rounds.

\begin{theorem}{\cite[Theorem 1]{YassaeeGA15}}\label{thm:a_rate-4}
For any $(q_{XY},q_{Z|XY})$, 
\begin{center}{\em $R^{\text{No-privacy}}\quad=\quad\displaystyle \min_{\substack{p_{U|XYZ}: \\ U-X-Y \\ Z-(U,Y)-X \\}} I(X,Z;U|Y),$}\end{center}
where cardinality of $\U$ satisfies $|\U|\leq|\X|\cdot|\Y|\cdot|\Z|+2$.
\end{theorem}
We define $R^{AB-\text{pvt}}$, $R^{A-\text{pvt}}$ and $R^{B-\text{pvt}}$ as the infima of all the achievable rates in the absence of common randomness when privacy is required, respectively, against both users, only against Alice, and only against Bob. Theorems~\ref{Thm_Rate_Region_AB}, \ref{Thm_Rate_Region_A} and \ref{Thm_Rate_Region_B} give us expressions for $R^{AB-\text{pvt}}$, $R^{A-\text{pvt}}$ and $R^{B-\text{pvt}}$, respectively. Before we discuss the example, we present an alternative expression for $R^{A-\text{pvt}}$ in terms of conditional graph entropy \cite{OrlitskyR01}, $H_{\mathcal{G}}(X|Y)$ (see Definition~\ref{defn:rand-conditional-graph-entropy} in Section~\ref{Sect_Prelim}) which simplifies the computation. For a given $(q_{XY},q_{Z|XY})$, we define the following characteristic graph $\mathcal{G}=(\mathcal{X},\mathcal{E})$, where 
\begin{align*}
\mathcal{E}=\{\{x,x'\}: \exists (y,z)\text{ s.t. } q_{XY}(x,y)>0,q_{XY}(x',y)>0, \\ \text{and } q_{Z|XY}(z|x,y)\neq q_{Z|XY}(z|x',y)\}.
\end{align*} 
An {\em independent set} of $\mathcal{G}$ is a collection $\mathcal{U}\subseteq \mathcal{X}$ of vertices such that no two vertices of $\mathcal{U}$ are connected by an edge in $\mathcal{G}$.
Let $W$ denote the random variable corresponding to the independent sets in $\mathcal{G}$, and let $\varGamma(\mathcal{G})$ be the set of all independent sets in $\mathcal{G}$.

\begin{definition}[Conditional Graph Entropy, \cite{OrlitskyR01}]\label{defn:rand-conditional-graph-entropy}
For a given $(q_{XY},q_{Z|XY})$, the conditional graph entropy of $\mathcal{G}$ is defined as
\[H_{\mathcal{G}}(X|Y):= \min_{\substack{p_{W|X}: \\ W-X-Y \\ X\in W\in \varGamma(\mathcal{G})}}I(W;X|Y),\]
where minimum is taken over all the conditional distributions $p_{W|X}$ such that $p_{W|X}(w|x)>0$ only if $x\in w$.
\end{definition}
By the data-processing inequality, this minimization can be restricted to $W$'s ranging over maximal independent sets of $\mathcal{G}$ \cite{OrlitskyR01}. Note that $0\leq H_{\mathcal{G}}(X|Y) \leq H(X|Y)$ hold in general, and if $\mathcal{G}$ is a complete graph then $H_{\mathcal{G}}(X|Y)=H(X|Y)$.

\begin{proposition} \label{prop_asymptotic_alice}
Suppose $(q_{XY},q_{Z|XY})$ is asymptotically securely computable with privacy only against Alice, then $R^{A-\emph{pvt}}=H_{\mathcal{G}}(X|Y)$, where $H_{\mathcal{G}}(X|Y)$ is conditional graph entropy (see Definition~\ref{defn:rand-conditional-graph-entropy} in Section~\ref{Sec_nointeraction}).
\end{proposition}
Proposition~\ref{prop_asymptotic_alice} is proved in Appendix~\ref{appendixH}. 

\begin{example}
Consider the randomized function $(q_{XY},q_{Z|XY})$ given in Figure~\ref{fig:example1}.
%\begin{minipage}{1.0\textwidth}
%\centering
\begin{figure}
\centering
\quad \begin{tabular}{c|ccc|ccc|}
& \multicolumn{3}{c|}{$y_1$} & \multicolumn{3}{c|}{$y_2$} \\
\hline
& 0 & e & 1 & 0 & e & 1 \\ 
%\hline
\Xhline{3\arrayrulewidth}
\multirow{1}{*}{$x_1$} & {\color{magenta} 0} & {\color{magenta} 1} & {\color{magenta} 0} & ${\color{magenta} 1-p}$ & {\color{magenta}$p$} & {\color{magenta} 0} \\
\hline
\multirow{1}{*}{$x_2$} & {\color{magenta}0} & {\color{magenta}1} & {\color{magenta}0} &  & {\color{magenta}-} & \\ 
\hline
\multirow{1}{*}{$x_3$} &  & {\color{magenta}-} &  & {\color{magenta}0} & {\color{magenta}$p$} & {\color{magenta}$1-p$} \\
\hline
\end{tabular}
%\hspace{1cm}
\caption[An example for comparison]
{The randomized function $q_{Z|XY}$ described by the above matrix has  
ternary $\X$ and $\Z$ alphabets and binary $\Y$ alphabet.
Specifically, $\X=\{x_1,x_2,x_3\}$, $\Y=\{y_1,y_2\}$, and $\Z=\{0,e,1\}$, where $e$ denotes the erasure symbol.
The input distribution is as follows: $q_{XY}(x_2,y_2)=q_{XY}(x_3,y_1)=0$ and $q_{XY}(x_1,y_1)=q_{XY}(x_1,y_2)=q_{XY}(x_2,y_1)=q_{XY}(x_3,y_2)=1/4$.
When Bob's input is $y_1$, the output is always an erasure.
When Bob's input is $y_2$, the randomized function $q_{Z|X,Y=y_2}$ behaves as an erasure channel with binary input in $\{x_1,x_3\}$ and parameter $p\in[0,1]$:
$q_{Z|XY}(0|x_1,y_2)=q_{Z|XY}(1|x_3,y_2)=1-p$ and $q_{Z|XY}(e|x_1,y_2)=q_{Z|XY}(e|x_3,y_2)=p$.
}
\label{fig:example1}
\end{figure}
If $p=1$, Alice does not need to send any message, and the rate required is zero in all the four cases.
So, we consider $p<1$, and analyze the randomized function under all the four cases of privacy (including no privacy) below. \\

\paragraph{Privacy against both users}
It can be shown from Lemma~\ref{lem:characterization} (Appendix~\ref{appendixD}) that the randomized function $(q_{XY},q_{Z|XY})$ is securely computable with privacy against both users 
if and only if $p=0$ or $p=1$.
As noted earlier, if $p=1$, the rate required is zero, i.e., $R^{AB-\text{pvt}}=0$.
If $p=0$, then $R^{AB-\text{pvt}}=I(X;Z|Y)=1/2$.\\

\paragraph{Privacy only against Alice}
By Proposition~\ref{prop_asymptotic_alice}, the optimum rate is $R^{A-\text{pvt}}=H_{\mathcal{G}}(X|Y)$. The characteristic graph (see Section~\ref{Sect_Prelim}) corresponding to the randomized function $(q_{XY},q_{Z|XY})$ is $\G=(\V,\E)$, where $\V=\{x_1,x_2,x_3\}$ and $\E=\{\{x_1,x_3\}\}$, i.e., $\G$ has a single edge between $x_1$ and $x_3$.
The optimal rate $R^{A-\text{pvt}}$ is equal to the conditional graph entropy $H_{\G}(X|Y)$ of this graph $\G$.
Note that there are two maximal independent sets $\{x_1,x_2\}$ and $\{x_2,x_3\}$ in $\G$.
It can be shown easily that $I(W;X|Y)$ is minimized when $p_{W|X}(\{x_1,x_2\}|x_1)=p_{W|X}(\{x_1,x_2\}|x_2)=1$, $p_{W|X}(\{x_2,x_3\}|x_3)=1$.
This yields $I(W;X|Y)=1/2$, which implies that $H_{\G}(X|Y)=1/2$. We obtain $R^{A-\text{pvt}}=1/2$ for $0\leq p<1$.
(When $p=1$, the graph $\G$ does not have any edge, and the rate required is zero, i.e., $R^{A-\text{pvt}}=0$.) \\

\begin{figure}
\begin{center}
\input{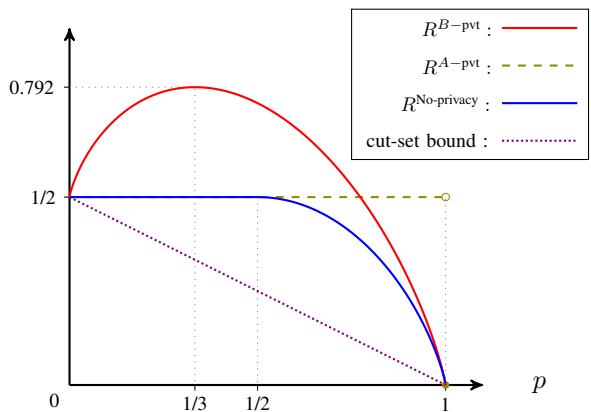}
\end{center}
\caption[Comparison of optimal rates in different privacy settings]
{\color{black} The cut-set bound (dotted violet curve on the bottom), i.e., $I(X;Z|Y)=\frac{1-p}{2}$ intersects the blue curve of no privacy case ($R^{\text{No-privacy}}$ denotes the optimal rate) at $p=0,1$, where the randomized function of Figure~\ref{fig:example1} can be securely computable with privacy against both users ($R^{AB-\text{pvt}}$ denotes the optimal rate). For $p=0$, $R^{AB-\text{pvt}}=1/2$, for $p=1$, $R^{AB-\text{pvt}}=0$. This is, in fact, consistent with Theorem~\ref{cutset}, which states that if a function is securely computable with privacy against both users, then cut-set lower bound (for function computation without any privacy) is achievable. In contrast, when the function is securely computable with privacy against only one user, the requirement of privacy, in general, may lead to a larger optimal rate. The red curve on the top is for privacy only against Bob ($R^{B-\text{pvt}}$ denotes the optimal rate). The olive green dashed curve is for privacy only against Alice ($R^{A-\text{pvt}}$ denotes the optimal rate); 
this curve does not depend on the value of $p$ for $0\leq p<1$, because $R^{A-\text{pvt}}$ is equal to the conditional graph entropy, which depends only on the input distribution for $0\leq p<1$. For $p=1$, $R^{A-\text{pvt}}=0$.
}
\label{fig:comparison}
\end{figure}

\paragraph{Privacy only against Bob}
From Theorem~\ref{Thm_Rate_Region_B}, the optimal rate for asymptotically securely computing $(q_{XY},q_{Z|XY})$ with privacy against Bob is equal to
\begin{align}
R^{B-\text{pvt}}=&\displaystyle \min_{\substack{p_{U|XYZ}: \\ U-X-Y \\ Z-(U,Y)-X \\ U-(Y,Z)-X}} I(X,Z;U|Y)\nonumber\\
=&\displaystyle \min_{\substack{p_{U|XYZ}: \\ U-X-Y \\ Z-(U,Y)-X \\ U-(Y,Z)-X}} I(Z;U|Y)\label{eq:example1-privacy-bob}
\end{align} 
First we show that this randomized function is securely computable with privacy against Bob, i.e., we give a random variable $U$ that satisfies the Markov chains in \eqref{eq:example1-privacy-bob}. For that we give an encoder-decoder pair $(p_{U|X},p_{{Z}|UY})$ 
that perfectly securely computes $(q_{XY},q_{Z|XY})$ with privacy against Bob. 
The random variable $U$ has ternary alphabet $\U=\{u_1,u_2,u_3\}$. The encoder-decoder pair $(p_{U|X},p_{{Z}|UY})$ is defined as follows: 

\begin{equation}
\begin{aligned}\label{eq:example1-ed-privacy-bob}
&p_{U|X}(u_1|x) \  && = \ 1-p  \hspace{3pt}\quad\text{for } x=x_1,x_2,\\
&p_{U|X}(u_2|x) \  && = \ p    \qquad\quad\text{for } x=x_1,x_2, \\
&p_{U|X}(u_2|x_3)\ && = \ p,  \\
&p_{U|X}(u_3|x_3)\ && = \ 1-p, \\
&p_{{Z}|UY}(e|u,y_1) &&= \ 1  \qquad\quad\text{for }u=u_1,u_2,\\
&p_{{Z}|UY}(0|u_1,y_2) &&= \ 1,\\
&p_{{Z}|UY}(e|u_2,y_2) &&= \ 1, \\
&p_{{Z}|UY}(1|u_3,y_2) &&= \ 1.
\end{aligned}
\end{equation}

\noindent It can be verified that the above encoder-decoder pair $(p_{U|X},p_{{Z}|UY})$ satisfies perfect privacy against Bob 
(i.e., it satisfies the Markov chain $U-(Y,Z)-X$) and perfect correctness (i.e., $p_{{Z}|X=x,Y=y}=q_{Z|X=x,Y=y}$ for every $x,y$ such that $q_{XY}(x,y)>0$).
We will now show that the choice of auxiliary random variable in \eqref{eq:example1-ed-privacy-bob} is optimal for the optimization problem in
\eqref{eq:example1-privacy-bob}. Consider any $p_{U|XYZ}$ which satisfies the constraints in \eqref{eq:example1-privacy-bob}.
Note that when $Y=y_1$, $Z=e$ with probability 1. This implies that $I(Z;U|Y=y_1)=0$, which simplifies the rate-expression for $R^{B-\text{pvt}}$ to the following:
\begin{align}
R^{B-\text{pvt}}=\displaystyle \min_{\substack{p_{U|XYZ}: \\ U-X-Y \\ Z-(U,Y)-X \\ U-(Y,Z)-X}} \frac{1}{2}I(Z;U|Y=y_2).\label{eq:example1-privacy-bob_new}
\end{align}
To evaluate the expression in \eqref{eq:example1-privacy-bob_new}, it is sufficient to take minimization over encoder-decoder pairs 
$(p_{U|X},p_{Z|U,Y=y_2})$ that perfectly securely compute $(q_{X|Y=y_2},q_{Z|X,Y=y_2})$ with privacy against Bob.
Consider such an encoder-decoder pair $(p_{U|X},p_{Z|U,Y=y_2})$.
Let $(z_1,z_2,z_3)=(0,e,1)$. For $i\in[1:3]$, define $\U_i:=\{u\in\U : p_{Z|U,Y=y_2}(z_i|u)>0\}$. 
It follows from Claim~\ref{claim:disjoint-messages} that $\U_1\cap\U_2\cap\U_3=\phi$
\footnote{When we restrict
Bob's input to $y_2$, then Alice's input $X$ takes values in $\{x_1,x_3\}$. For securely computing $(q_{X|Y=y_2},q_{Z|X,Y=y_2})$ 
with $p\in(0,1)$, we have $k=3$ and $\vec{\alpha_1^{(y_2)}}=(1,0)$, $\vec{\alpha_2^{(y_2)}}=(1/2,1/2)$, $\vec{\alpha_3^{(y_2)}}=(0,1)$. 
For $p=0$ we have $k=2$ and $\vec{\alpha_1^{(y_2)}}=(1,0)$, $\vec{\alpha_2^{(y_2)}}=(0,1)$.
For definitions of $k$ and $\vec{\alpha}^{(y)}$'s, see discussion on page~\pageref{eq:one-round_set-msgs}.}. 
This means that when Bob has $y_2$ as his input, then $U$ determines $Z$, i.e., $H(Z|U,Y=y_2)=0$, 
which implies $R^{B-\text{pvt}}=\frac{1}{2}I(Z;U|Y=y_2)=\frac{1}{2}H(Z|Y=y_2)=\frac{1}{2}(h(p)+(1-p))$.
Thus, any encoder-decoder pair $(p_{U|X},p_{Z|UY})$ that perfectly securely computes the randomized function of Figure~\ref{fig:example1} achieves the same value for $I(Z;U|Y)$. In particular, the pair defined in \eqref{eq:example1-ed-privacy-bob} also achieves the optimal rate 
in \eqref{eq:example1-privacy-bob_new}, which is equal to $\frac{1}{2}(h(p)+(1-p))$, for $0\leq p\leq 1$. \\

\paragraph{No privacy}
From Theorem~\ref{thm:a_rate-4}, the optimal rate is equal to 
\begin{align}
R^{\text{No-privacy}}=\displaystyle \min_{\substack{p_{U|XYZ}: \\ U-X-Y \\ Z-(U,Y)-X}} I(X,Z;U|Y). \label{eq:example1-no-privacy}
\end{align}
Note that minimization in the above expression is over all encoder-decoder pairs $(p_{U|X},p_{Z|UY})$ such that they induce 
the correct $q_{Z|XY}(.|x,y)$ for every $x\in\X,y\in\Y$ for which $q_{XY}(x,y)>0$.
Note that when $Y=y_1$, $Z=e$ with probability 1.
Now, for any $(p_{U|X},p_{Z|UY})$, define another encoder-decoder pair $(q_{U|X},q_{Z|UY})$ as follows:
\begin{align}\label{eq:example1-ed-no-privacy}
\begin{split}
q_{U|X}(u|x) &:= \begin{cases}
p_{U|X}(u|x_1) & \text{ if }x=x_1,x_2, \\ 
p_{U|X}(u|x_3) & \text{ if }x=x_3.
\end{cases} \\ 
q_{Z|UY}(z|u,y) &:= \begin{cases}
\mathbbm{1}_{\{z=e\}} & \text{ if }y=y_1, \\ 
p_{Z|UY}(z|u,y_2) & \text{ if }y=y_2.
\end{cases} 
\end{split}
\end{align}
It can be verified that $(q_{U|X},q_{Z|UY})$ is a valid encoder-decoder pair that correctly computes the randomized function of Figure~\ref{fig:example1}.
Observe that $$I(X,Z;U|Y=y_1)|_{q_{UXZ|Y=y_1}}=0$$ and 
\begin{align*}
&I(X,Z;U|Y=y_2)|_{q_{UXZ|Y=y_2}}\nonumber\\
&\hspace{12pt}=I(X,Z;U|Y=y_2)|_{p_{UXZ|{Y=y_2}}}.
\end{align*}
This implies that $I(X,Z;U|Y)_{q_{UXYZ}}\leq I(X,Z;U|Y)_{p_{UXYZ}}$, which further implies that the expression for $R^{\text{No-Privacy}}$ from 
\eqref{eq:example1-no-privacy} reduces to the following:
\begin{align}
R^{\text{No-privacy}}=\displaystyle \min_{\substack{p_{U|XYZ}: \\ U-X-Y \\ Z-(U,Y)-X}} \frac{1}{2}I(X,Z;U|Y=y_2).\label{eq:example1-no-privacy-1}
\end{align}
\begin{figure}[hbt]
\begin{center}
\begin{tikzpicture}[>=stealth']
\node at (-2.5,2.0) {$X$}; \node at (-2.5,1.5) {$x_1$}; \node at (-2.5,-1.5) {$x_3$};  
\draw [->] (-2.25,1.5) -- (-0.25,1.5); \node[scale=0.9] at (-1.25,1.7) {$1-p_1$};
\draw [->] (-2.25,1.4) -- (-0.25,0.1); \node[scale=0.9] at (-1.25,1) {$p_1$};
\draw [->] (-2.25,-1.4) -- (-0.25,-0.1); \node[scale=0.9] at (-1.25,-0.5) {$p_1$};
\draw [->] (-2.25,-1.5) -- (-0.25,-1.5); \node[scale=0.9] at (-1.25,-1.3) {$1-p_1$};
\node at (0,2.0) {$U$}; \node at (0,1.5) {0}; \node at (0,0) {$e$}; \node at (0,-1.5) {1};  
\draw [->] (0.25,1.5) -- (2.25,1.5); \node[scale=0.9] at (1.25,1.7) {$1-p_2$};
\draw [->] (0.25,1.4) -- (2.25,0.1); \node[scale=0.9] at (1.25,1) {$p_2$};
\draw [->] (0.25,0) -- (2.25,0); \node[scale=0.9] at (1.25,0.2) {1};
\draw [->] (0.25,-1.4) -- (2.25,-0.1); \node[scale=0.9] at (1.25,-0.5) {$p_2$};
\draw [->] (0.25,-1.5) -- (2.25,-1.5); \node[scale=0.9] at (1.25,-1.3) {$1-p_2$};
\node at (2.5,2.0) {$Z$}; \node at (2.5,1.5) {0}; \node at (2.5,0) {$e$}; \node at (2.5,-1.5) {1};  
\end{tikzpicture}
\caption[An optimal strategy for no privacy case]
{This figure is for $Y=y_2$. The $p_{UXZ|Y=y_2}$ which minimizes the rate-expression in \eqref{eq:example1-no-privacy-1}
is a concatenation of two symmetric erasure channels -- the first channel has parameter $p_1$, 
the second channel has parameter $p_2$, and $p_1,p_2$ satisfies the constraint that $(1-p_1)(1-p_2)=1-p$.}
\label{fig:erasure-cuff}
\end{center}
\end{figure}
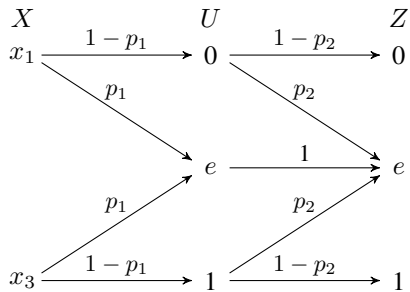
\noindent To evaluate the expression in \eqref{eq:example1-no-privacy-1}, it is sufficient to take minimization over encoder-decoder pairs 
$(p_{U|X},p_{Z|U,Y=y_2})$ that compute $(q_{X|Y=y_2},q_{Z|X,Y=y_2})$ with perfect correctness. We will argue that 
the $(p_{U|X},p_{Z|U,Y=y_2})$ shown in Figure~\ref{fig:erasure-cuff} includes a minimizer for \eqref{eq:example1-no-privacy-1}. 
It is a concatenation of two symmetric erasure channels -- first erasure channel is with parameter $p_1\in[0,1]$ 
and the second erasure channel is with parameter $p_2\in[0,1]$ -- with the constraint that $(1-p_1)(1-p_2)=1-p$. The reason that this choice of
encoder-decoder pair (which is given in Figure~\ref{fig:erasure-cuff}) is optimal for our problem is given next:

The minimization in \eqref{eq:example1-no-privacy-1} is related to Wyner's common information problem \cite{Wyner75}, 
where Alice gets an input sequence $X^n$ and Bob wants to 
generate an output sequence ${Z}^n$ such that $(X^n,{Z}^n)$ is close to the desired $(X^n,Z^n)$ in $\ell_1$-distance, where $(X_i,Z_i)$'s
are i.i.d. according to a given $q_{XZ}$. The minimum rate required for this problem is equal to the Wyner's common information, which is
defined as $\min_{p_{U|XZ}}I(X,Z;U)$, where minimization is taken over all $p_{U|XZ}$ such that $X-U-Z$ is a Markov chain.
% that satisfies $p_{XUZ}=p_{X}p_{U|X}p_{Z|U}$.
The $p_{U|XZ}$ that satisfies $X-U-Z$ and achieves the least value for $I(X,Z;U)$ was obtained by Cuff~\cite[Section II-F]{Cuff13}
and is presented in Figure~\ref{fig:erasure-cuff} (without $Y$). 
The expression that we want to compute is $\min_{p_{U|X,Y=y_2,Z}} I(X,Z;U|Y=y_2)$, 
where minimization is taken over all $p_{U|X,Y=y_2,Z}$ such that $X-U-Z$ is a Markov chain. 
Note that we are also computing the Wyner common information for the joint distribution over $\X\times\Z$ given by the conditional 
distribution $q_{XZ|Y=y_2}$. 
So, the same encoder-decoder pair $(p_{U|X},p_{Z|U,Y=y_2})$ from Figure~\ref{fig:erasure-cuff} also minimizes the rate 
expression in \eqref{eq:example1-no-privacy-1}.

With this $(p_{U|X},p_{Z|U,Y=y_2})$, the objective function in \eqref{eq:example1-no-privacy-1} becomes 
$I(X,Z;U|Y=y_2)=\frac{1}{2}[h(p)+(1-p_1)(1-h(p_2))]$. 
Note that the $(p_{U|X},p_{Z|U,Y=y_2})$ from Figure~\ref{fig:erasure-cuff} has two parameters $p_1,p_2\in[0,1]$, 
which must satisfy the constraint $(1-p_1)(1-p_2)=1-p$.
Hence, the minimization in \eqref{eq:example1-no-privacy-1} can be taken over all $p_1,p_2\in[0,1]$ which satisfy $(1-p_1)(1-p_2)=1-p$.
The simplified rate expression for $R^{\text{No-privacy}}$ is given below (for all values of $p\in[0,1]$):
\begin{align}
R^{\text{No-privacy}}=\displaystyle \min_{\substack{p_1,p_2\in[0,1]: \\ (1-p_1)(1-p_2)=1-p}} \frac{1}{2}[h(p)+(1-p_1)(1-h(p_2))].\label{eq:example1-no-privacy-2}
\end{align}

\noindent \emph{Comparison:}
For convenience, we write the optimal rate expressions for asymptotic security in all the four privacy settings,
and plot these in Figure~\ref{fig:comparison}.\\

\hrule
\begin{align*}
R^{AB-\text{pvt}}\quad &=\quad \begin{cases}1/2 &\text{ at } p=0, \\ 0 &\text{ at } p=1,\end{cases} \\
R^{A-\text{pvt}}\quad &=\quad \begin{cases}1/2 &\text{ for } 0\leq p<1, \\ 0 &\text{ at }p=1,\end{cases} \\
R^{B-\text{pvt}}\quad &=\quad \frac{1}{2}(h(p)+(1-p)) \quad\text{ for } 0\leq p\leq 1, \\
R^{\text{No-privacy}}&\\ 
&\hspace{-1cm}\quad =\min_{\substack{p_1,p_2\in[0,1]: \\ (1-p_1)(1-p_2)=1-p}} \frac{1}{2}[h(p)+(1-p_1)(1-h(p_2))] \\ &\hspace{0pt}\text{ for } 0\leq p\leq 1.
\end{align*}
\hrule
\vspace{0.5cm}
\end{example}

\section{Perfect Security: One Round of Communication; Only Bob Computes}
\label{Sec_nointeraction}
\RRB{In this section, we study perfectly secure protocols where only one user, say, Bob, produces an output. We consider three different privacy settings: one where privacy is required against both the users, and the other two where privacy is required against only one of them.
We further restrict ourselves to protocols that use a single round of communication, where Alice sends a message to Bob who then produces an output. We obtain tight upper and lower bounds on the optimal expected length of transmission. Although we consider a one-sided communication model here, as we remark later (Remark~\ref{remark:wlog}), all the results which require privacy against Alice hold in the multiple-round communication model also.}

\subsection{Preliminaries}
\label{Sect_Prelim}
Most of our results in this section are stated in terms of various information theoretic quantities that are defined on graphs 
which arise from randomized function computation.
For a given $(q_{XY},q_{Z|XY})$, we define the following characteristic graph $\mathcal{G}=(\mathcal{X},\mathcal{E})$, where 
\begin{align*}
\mathcal{E}=\{\{x,x'\}: \exists (y,z)\text{ s.t. } q_{XY}(x,y)>0,q_{XY}(x',y)>0 \\ \text{and } q_{Z|XY}(z|x,y)\neq q_{Z|XY}(z|x',y)\}.
\end{align*} 
There is a probability distribution $q_X$ on the vertices of $\mathcal{G}$.
The intuition behind this definition of a graph comes from the randomized function computation,
where Alice and Bob have inputs $X$ and $Y$, respectively, and Bob wants to produce a randomized function $q_{Z|XY}$ of these inputs.
Suppose $x,x'\in\mathcal{X},y\in\mathcal{Y},z\in\mathcal{Z}$ are such that $q_{XY}(x,y)>0,q_{XY}(x',y)>0$ and $q_{Z|XY}(z|x,y)\neq q_{Z|XY}(z|x',y)$. 
Observe that Alice cannot send the same message if her input is $x$ or $x'$: 
because if she does so, and Bob happens to have $y$ as his input (which happens with positive probability),
then there is no way Bob can produce an output with correct distribution.
This means that Alice has to distinguish through the message whether her input is $x$ or $x'$. 
This distinction is achieved by defining a graph on $\mathcal{X}$ and putting an edge between $x$ and $x'$. 
An edge $\{x',x'\}$ means that Alice cannot send the same message if her input is $x$ or $x'$.
Per contra, if $\{x,x'\}\notin \mathcal{E}$, then Alice does not need to distinguish between $x$ and $x'$,
because, if $q_{XY}(x,y)>0,q_{XY}(x',y)>0$ for some $y\in\mathcal{Y}$, then $q_{Z|XY}(z|x,y)= q_{Z|XY}(z|x',y)$ holds for every $z\in\mathcal{Z}$.
So, if $\{x,x'\}\notin\mathcal{E}$, then $x$ and $x'$ are equivalent from the function computation point of view.

We say that a coloring $c:\mathcal{X}\to\{0,1\}^*$ of the vertices of $\mathcal{G}$ is {\em proper}, if 
both the vertices of any edge $\{x,x'\}$ have distinct colors, i.e., $c(x)\neq c(x')$.
Note that $c(X)$ is a random variable with entropy
\[H(c(X))=\sum_{\gamma\ \in\ \text{image}(c)}q_X(c^{-1}(\gamma))\log_2 \frac{1}{q_{X}(c^{-1}(\gamma))},\]
where $\text{image}(c)$ is the image of the function $c$, $c^{-1}$ is the inverse of $c$, and $q_{X}(c^{-1}(\gamma))$ for a color $\gamma$ 
is defined as $q_{X}(c^{-1}(\gamma)) := \sum_{x\in\mathcal{X}:c(x)=\gamma}q_{X}(x)$.
Note that the function $c$ partitions $\mathcal{X}$ into color classes, i.e., the set of vertices assigned the same color,
and $H(c(X))$ is the entropy of this partition.
The {\em chromatic entropy} $H_{\chi}(\mathcal{G},X)$ of $\mathcal{G}$ is defined as the lowest entropy of 
such a partition produced by any proper coloring of $\mathcal{G}$:
\begin{definition}[Chromatic Entropy, \cite{AlonO96}]\label{defn:chromatic-entropy}
For a given $(q_{XY},q_{Z|XY})$, the chromatic entropy of $\mathcal{G}$ is defined as
\[H_{\chi}(\mathcal{G},X):=\min\{H(c(X)): c \text{ is a proper coloring of }\mathcal{G}\},\]
where $c(X)$ denotes the induced distribution on colors by the coloring c.
\end{definition}

\begin{figure}[htbp]
\begin{center}
\begin{tikzpicture}[>=stealth']
\draw [fill=lightgray, thick] (0,0) rectangle node {\Large A} +(1,1); \draw [->,thick] (-0.65,0.5) -- (0,0.5); \node at (-0.9,0.6) {$X$};
\draw [fill=lightgray, thick] (4,0) rectangle node {\Large B} +(1,1); \draw [<-,thick] (5,0.5) -- (5.65,0.5); \node at (6,0.55) {$Y$};
\draw [->,thick] (4.5,0) -- (4.5,-0.5); \node[right] at (4.2,-0.8) {$Z\sim q_{Z|XY}$}; %\node [right,scale=0.8] at (5,-0.5) {$Z_i\sim p_{Z|XY}$};
\draw [->,thick] (1,0.5) -- (4,0.5); \node at (2.5,0.75) {$U$};
\end{tikzpicture}
\caption{Perfectly secure computation}\label{fig:perfect-security}
\end{center}
\end{figure}
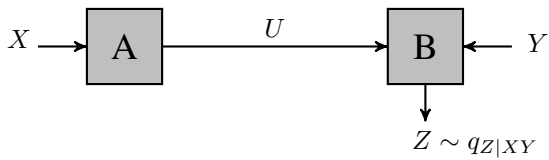
{\color{black}
As shown in Figure~\ref{fig:perfect-security}, Alice and Bob get $X$ and $Y$, respectively, as their inputs, 
where $(X,Y)$ is distributed according to $q_{XY}$;
Alice sends a message $U$ to Bob, and based on $U,Y)$, Bob produces $Z$ as the output, which is required to be distributed according to $q_{Z|XY}$ while preserving privacy. {\color{black}We allow variable length codes. Any scheme for computing $(q_{XY},q_{Z|XY})$ is defined as a pair of stochastic maps $(p_{U|X},p_{Z|UY})$, where $U$ takes values in $\{0,1\}^*$ such that all the binary strings in the support set of $U$ are prefix-free. We call $p_{U|X}$ the encoder and  $p_{Z|UY}$ the decoder. Let $L(U)$ denote the random variable corresponding to the length of $U$ in bits and the expected length $\mathbb{E}[L(U)]$ denote the rate of code.} In this section, we give bounds on the optimal rate of communication for the three cases of privacy, i.e., privacy against both users, privacy only against Alice and privacy only against Bob. In this section, we will only consider the case when there is no common randomness available to the users. However, we show below in Proposition~\ref{lemma:prvagnstAlice} that this is without loss of generality whenever privacy against Alice is required (i.e., when privacy is required against both Alice and Bob or only against Alice). When privacy is required only against Bob, we show through an example (Example \ref{example:prvagnstbob}) later in this section that common randomness helps to improve the rate of communication. 

\begin{proposition}\label{lemma:prvagnstAlice}
When privacy against Alice is required (i.e., when privacy is required against both users or when privacy is required only against Alice), common randomness does not help to improve the optimal rate.
\end{proposition} 
\begin{proof}
Let $(p_{U|XW},p_{Z|WUY})$ be an encoder-decoder pair {\color{black}(note that this pair takes in to account the common randomness also in addition to the definition of encoder-decoder pair mentioned in the paragraph above this lemma)} that securely computes $(q_{XY},q_{Z|XY})$. Consider the probability distribution,
\begin{align}
p(w,x,y,z)&=p(x,y,z)p(w|x,y,z)\\
&=p(x,y,z)p(w|x) \label{prvAlice}\\
&=p(x,y,z)p(w) \label{x_indp_w},
\end{align}
where \eqref{prvAlice} follows from privacy against Alice, i.e., $(U,W)-~X-~(Y,Z)$, \eqref{x_indp_w} follows from the independence of random variables $W$ and $X$. Thus we have,
\begin{align}
p(w,x,y,z)=p(w)p(x,y,z). \label{correctnessforprv}
\end{align}
 Now, consider the expected length,
\begin{align}
\mathbb{E}\left[ L(U) \right] &=\mathbb{E}\left[\mathbb{E}\left[L(U)|W \right] \right]\\
&=\sum_w p(w) \mathbb{E}\left[L(U)|W=w \right].
\end{align}
This implies that there exists $w^*$ such that $\mathbb{E}\left[L(U)|W=w^* \right]\leq \mathbb{E}[L(U)]$, i.e., there exists a p.m.f. $p(u,x,y,z|w^*)$ achieving the same or smaller expected length and satisfying the required Markov chains. Note that \eqref{correctnessforprv} makes sure that correctness condition holds.
\end{proof}
}

\subsection{\RRB{Privacy required against both users}}
\RRB{We consider protocols which provide privacy against both the users. Before proceeding further, we need some definitions.}
\begin{definition}\label{defn:relation-tilde}
For any $x,x'\in\mathcal{X}$, we say that $x\sim x'$, if there exists $y\in \mathcal{Y}$
and $z\in \mathcal{Z}$ such that $q_{XY}(x,y)>0,q_{XY}(x',y)>0$ and $q_{Z|XY}(z|x,y)>0$, $q_{Z|XY}(z|x',y)>0$. 
\end{definition}
\begin{definition}\label{defn:equiv-relation-1}
For any two distinct elements $x,x'$ of $\mathcal{X}$, we say that $x\equiv x'$, if there exists a sequence $x=x_1,x_2,\hdots,x_{l-1},x_l=x'$ for some integer $l$, where $x_i\sim x_{i+1}$ for every $i \in \{1,2,\hdots,l-1\}$.
\end{definition}
It can be verified easily that the above defined relation $\equiv$ is an equivalence relation. \RRB{A similar notion of equivalent inputs has been used in characterizing securely computable \emph{deterministic} functions~\cite{Kushelvitz92}.}
We know that an equivalence relation partitions the whole space into equivalence classes.
Suppose the relation $\equiv$ partitions the space $\mathcal{X}$ as  
$\mathcal{X}=\mathcal{X}_1\biguplus \mathcal{X}_2\biguplus\hdots\biguplus\mathcal{X}_k$, 
where each $\mathcal{X}_i$ is an equivalence class, $k$ is the number of equivalence classes, and $\biguplus$ stands for disjoint union.

{\color{black} Let $\mathcal{X}_{\text{EQ}}:=\{x_1,x_2,\hdots,x_k\}$, 
where $x_i$ is the representative of the equivalence class $\mathcal{X}_i$. Now, we define ($q_{X_{\text{EQ}}Y},q_{Z|X_{\text{EQ}}Y}$) as follows:}
\begin{itemize}
\item Define $q_{X_{\text{EQ}}Y}(x_i,y):=\sum_{x\in\mathcal{X}_i}q_{XY}(x,y)$ for every $(x_i,y)\in(\mathcal{X}_{\text{EQ}}\times\mathcal{Y})$. 
\item For $(x_i,y)\in(\mathcal{X}_{\text{EQ}}\times\mathcal{Y})$, if there exists $x\in\mathcal{X}_i$ s.t. $q_{XY}(x,y)>0$, then define $q_{Z|X_{\text{EQ}}Y}(z|x_i,y) :=q_{Z|XY}(z|x,y)$, for every $z\in\mathcal{Z}$. If there exists no $x\in\mathcal{X}_i$ s.t. $q_{XY}(x,y)>0$, then it does not matter what the conditional distribution $q_{Z|X_{\text{EQ}}Y}(z|x_i,y)$ is;
in particular, we can define $q_{Z|X_{\text{EQ}}Y}(.|x_i,y)$ to be the uniform distribution in $\{1,2,\hdots,|\mathcal{Z}|\}$.
\end{itemize}

\begin{lemma}\label{lem:equiv-problem-alice-bob}
Suppose $(q_{XY},q_{Z|XY})$ is perfectly securely computable with privacy against both users,
and let $(q_{X_{\text{EQ}}Y},q_{Z|X_{\text{EQ}}Y})$ be the reduced problem defined as above.
For every encoder-decoder pair $(p_{U|X},p_{Z|UY})$ that securely computes $(q_{XY},q_{Z|XY})$, there is another encoder-decoder pair
$(p_{\tilde{U}|X_{\text{EQ}}},p_{Z|\tilde{U}Y})$ that securely computes $(p_{X_{\text{EQ}}Y},p_{Z|X_{\text{EQ}}Y})$, and vice-versa.
Furthermore, $L(U)$ and $L(\tilde{U})$ have the same p.m.f., where $L(U)$ and $L(\tilde{U})$ denote the random variables corresponding to the 
lengths of $U$ and $\tilde{U}$ (in bits), respectively.
\end{lemma}

It follows from Lemma \ref{lem:equiv-problem-alice-bob} proved in Appendix~\ref{appendixD}  that, to study the communication complexity of secure computation of $(q_{XY},q_{Z|XY})$,
it is enough to study the communication complexity of secure computation of the reduced problem $(q_{X_{\text{EQ}}Y},q_{Z|X_{\text{EQ}}Y})$.
To give a rate-optimal code for $(q_{X_{\text{EQ}}Y},q_{Z|X_{\text{EQ}}Y})$, we define the following graph $\mathcal{G}_{\text{EQ}}=(\mathcal{V}_{\text{EQ}},\mathcal{E}_{\text{EQ}})$:

\begin{itemize}
\item $\mathcal{V}_{\text{EQ}}=\mathcal{X}_{\text{EQ}}=\{x_1,x_2,\hdots,x_k\}$,
\item $\mathcal{E}_{\text{EQ}}=\{\{x_i,x_j\}:$\text{ there exists} $(y,z)$\text{ such that } $q_{X_{\text{EQ}}Y}(x_i,y)>0$, $q_{X_{\text{EQ}}Y}(x_j,y) >0$ \text{and }$q_{Z|X_{\text{EQ}}Y}(z|x_i,y)\ \neq \ q_{Z|X_{\text{EQ}}Y}(z|x_j,y)\}$.
\end{itemize}
We say that a coloring $c:\mathcal{V}\to\{0,1\}^*$ of the vertices of a graph $\mathcal{G}=(\mathcal{V},\mathcal{E})$ is {\em proper}, if 
both the vertices of any edge $\{u,v\}$ have distinct colors, i.e., $c(u)\neq c(v)$.
\begin{claim}\label{claim:color-equiv-protocol}
Every proper coloring of the vertices of $\mathcal{G}_{\text{EQ}}$ corresponds to a secure code for 
$(q_{X_{\text{EQ}}Y},q_{Z|X_{\text{EQ}}Y})$, and every secure code for $(q_{X_{\text{EQ}}Y},q_{Z|X_{\text{EQ}}Y})$ corresponds to a collection of
proper colorings of the vertices of $\mathcal{G}_{\text{EQ}}$.
\end{claim}
\RRB{The above claim is proved in Appendix~\ref{appendixE}.
The equivalence between {\em non-secure protocols for source coding} and graph coloring (for an appropriately defined graph) has been obtained in \cite{AlonO96} using a straightforward argument.
To show a similar equivalence in the case of {\em secure protocols for function computation}, first we partition the input alphabet according to an equivalence relation (see Definition~\ref{defn:equiv-relation-1}), which reduces our problem to a problem on smaller alphabets such that we can show a similar equivalence between secure protocols and graph colorings. Proof of this claim is more involved than the corresponding proof in the source coding case; see Appendix~\ref{appendixE} for details.}

Now we are ready to prove upper and lower bounds on optimal expected length of the message, which is defined as $L^*_{AB-\text{pvt}}:=\min_{(p_{U|X},p_{{Z}|UY})}\mathbb{E}[L(U)]$,
where $L(U)$ denotes the random variable corresponding to the length of $U$ in bits, and 
minimization is taken over all $(p_{U|X},p_{{Z}|UY})$ that perfectly securely computes $(q_{XY},q_{Z|XY})$ with privacy against both users.
\RRB{Having established the equivalence between secure protocols and graph colorings, in the following theorem, we prove 
bounds on $L^*_{AB-\text{pvt}}$ which can be proved along the lines of the proof of \cite[Theorem 1]{AlonO96} (which is for source coding) with appropriate modifications that arise due to the {\em randomized} nature of our secure protocols.}
\begin{theorem}\label{thm:ps_rate-1}
Suppose $(q_{XY},q_{Z|XY})$ is perfectly securely computable with privacy against both users. Then
{\em\[H_{\chi}(\mathcal{G}_{\text{EQ}},X_{\text{EQ}})\leq L^*_{AB-\text{pvt}}< H_{\chi}(\mathcal{G}_{\text{EQ}},X_{\text{EQ}})+1,\]}
where chromatic entropy $H_{\chi}(.,.)$ is defined in Definition \ref{defn:chromatic-entropy}.
\end{theorem}
%{\color{magenta}A similar result for source coding (with the original input alphabets, not the reduced one as above) has been obtained in \cite[Theorem 1]{AlonO96}.}
\begin{proof}
Upper bound: Claim \ref{claim:color-equiv-protocol} implies that an optimal coloring of $\mathcal{G}_{\text{EQ}}$ corresponds to a secure code 
for $(q_{X_{\text{EQ}}Y},q_{Z|X_{\text{EQ}}Y})$. 
Now, since expected length of any optimal prefix-free binary encoding of a random variable $V$ is upper-bounded by $H(V)+1$ 
\cite[Theorem 5.4.1]{CoverJ06}, we have that the optimal expected length of the message from Alice to Bob is upper-bounded by $H_{\chi}(\mathcal{G}_{\text{EQ}},X_{\text{EQ}})+1$.

Lower bound: By Claim \ref{claim:color-equiv-protocol}, every secure (randomized) code also corresponds to 
proper (random) colorings $(c_{\vec{r}})_{\vec{r}}$ of the graph $\mathcal{G}$. 
Note that the expected length $L(c_{\vec{r}})$ of any coloring $c_{\vec{r}}$ is lower-bounded by $H(c_{\vec{r}}(X))$, 
which follows from the fact that the expected length of any prefix-free binary code 
for a random variable $V$ is lower-bounded by $H(V)$ \cite[Theorem 5.4.1]{CoverJ06}. 
Now, since entropy of any coloring is lower-bounded by $H_{\chi}(\mathcal{G}_{\text{EQ}},X_{\text{EQ}})$, entropy of the optimal coloring, 
it follows that $\mathbb{E}_{\vec{r}}[L(c_{\vec{r}})]\geq H_{\chi}(\mathcal{G}_{\text{EQ}},X_{\text{EQ}})$. 
So the expected length of the message is also lower-bounded by $H_{\chi}(\mathcal{G}_{\text{EQ}},X_{\text{EQ}})$.
\end{proof}

\RRB{\begin{remark}\label{remark:wlog}
In perfect security setting, when privacy against Alice is required, any multiple-round protocol can be turned into a single-round protocol without affecting the number of bits communicated. To see this, suppose that a function $(q_{XY},q_{Z|XY})$ is perfectly securely computable in $r$ rounds with privacy against Alice. Then there exists a conditional p.m.f. $p(u_{[1:r]}|x,y,z)$ satisfying \eqref{Eq_AB_Markov1}-\eqref{Eq_AB_Markov_Secr1} with $Z_1=\emptyset$ and $Z_2=Z$. Condition \eqref{Eq_AB_Markov_Secr1}, i.e., the Markov chain $U_{[1:r]}-X-(Y,Z)$ corresponds to privacy against Alice. This implies the Markov chain $U_{[1:r]}-X-Y$ which in turn implies that the whole transcript generated by a multi-round communication protocol can be simulated by Alice conditioned on her input in a single round. 
\end{remark}
}

\noindent \textbf{Multiple Instances.} Alice and Bob have $X^n$ and $Y^n$ as their inputs, respectively, where $(X_i,Y_i)\sim q_{XY}$, i.i.d., and Bob wants to compute $Z^n$, 
{which should be distributed according to $q_{Z^n|X^nY^n}(z^n|x^n,y^n)=\prod_{i=1}^nq_{Z|XY}(z_i|x_i,y_i)$.} 
Computation should be such that no user learns any additional information about other user's data other than what can be inferred from their own data.
It is easy to see that $(q_{X^nY^n},q_{Z^n|X^nY^n})$ is perfectly securely computable with privacy against both users 
if and only if $(q_{XY},q_{Z|XY})$ is perfectly securely computable with privacy against both users.
We can define the graph $\mathcal{G}_{\text{EQ}}^n=(\mathcal{X}_{\text{EQ}}^n,\mathcal{E}_{\text{EQ}}^n)$ by extending the definition of $\mathcal{G}_{\text{EQ}}$ in a straightforward manner: 
$\mathcal{X}_{\text{EQ}}^n$ is the $n$-fold cartesian product of $\mathcal{X}_{\text{EQ}}$, 
and $\mathcal{E}_{\text{EQ}}^n=\{\{x^n,x'^n\}\in\X_{\text{EQ}}^n\times \X_{\text{EQ}}^n:$ $\exists (y^n,z^n)\in\mathcal{Y}^n\times\mathcal{Z}^n$\text{ s.t. } $p_{X_{\text{EQ}}^nY^n}(x^n,y^n)>0$, $p_{X_{\text{EQ}}^nY^n}(x'^n,y^n)>0$ \text{and} $p_{Z^n|X_{\text{EQ}}^nY^n}(z^n|x^n,y^n)\neq p_{Z^n|X_{\text{EQ}}^nY^n}(z^n|x'^n,y^n)\}$. 
Note that we slightly abused the notation: $\mathcal{X}_{\text{EQ}}^n$ is defined to be the $n$-fold cartesian product of $\mathcal{X}_{\text{EQ}}$, but $\mathcal{E}_{\text{EQ}}^n$ is {\em not} the $n$-fold cartesian product of $\mathcal{E}_{\text{EQ}}$.

We define, analogous to the single instance case, $L_{AB-\text{pvt}-n}^*$ to be the 
optimal amortized expected length of the message that Alice sends to Bob in this setting.
The following is a simple corollary of Theorem~\ref{thm:ps_rate-1}. 
\begin{corollary}\label{cor:ps_rate-1-n}
Suppose $(q_{XY},q_{Z|XY})$ is perfectly securely computable with privacy against both users. Then
{\em\[\frac{1}{n} H_{\chi}(\mathcal{G}_{\text{EQ}}^n,X_{\text{EQ}}^n) \leq L_{AB-\text{pvt}-n}^* < \frac{1}{n} (H_{\chi}(\mathcal{G}_{\text{EQ}}^n,X_{\text{EQ}}^n) + 1).\]}
\end{corollary}
Observe that $\mathcal{G}_{\text{EQ}}^n$ is neither the {\sc and}-product nor the {\sc or}-product of the graph $\mathcal{G}_{\text{EQ}}$;
and as far as we know, no single letter expression of $\frac{1}{n}H_{\chi}(\mathcal{G}_{\text{EQ}}^n,X_{\text{EQ}}^n)$ (not even of $\lim_{n\to\infty}\frac{1}{n}H_{\chi}(\mathcal{G}_{\text{EQ}}^n,X_{\text{EQ}}^n)$) is known.  See \cite{AlonO96} for definitions of {\sc and}-product and {\sc or}-product of graphs.

\vspace{12pt}
\subsection{\RRB{Privacy required only against Alice}}
\RRB{We now consider the case where privacy only against Alice is required.
Note that every $(q_{XY},q_{Z|XY})$ is perfectly securely computable in this case, 
since Alice may send $X$ to Bob, and Bob, having $Y$ as his input, may output $Z$ according to $q_{Z|XY}$;
see Footnote~\ref{foot:rand-not-trivial_priv-Alice} (on page~\pageref{foot:rand-not-trivial_priv-Alice}) for a discussion on why one-sided communication does not automatically give privacy against Alice.}
For a given $(q_{XY},q_{Z|XY})$, we define the following characteristic graph $\mathcal{G}=(\mathcal{X},\mathcal{E})$, where the edge set is defined as follows:
%$\E=\{\{x,x'\}: \exists (y,z)\in\Y\times\Z\text{ s.t. }p_{XY}(x,y)>0,p_{XY}(x',y)>0 \text{ and }p_{Z|XY}(z|x,y)\neq p_{Z|XY}(z|x',y)\}$. 
\begin{align}
\mathcal{E}=\{\{x,x'\}: \exists (y,z)\text{ s.t. } q_{XY}(x,y)>0,q_{XY}(x',y)>0 \nonumber\\ \text{ and } q_{Z|XY}(z|x,y)\neq q_{Z|XY}(z|x',y)\}. \label{eq:normal-graph-alice}
\end{align} 
We define optimal expected length of the message, which is defined as $L^*_{A-\text{pvt}}:=\min_{(p_{U|X},p_{{Z}|UY})}\mathbb{E}[L(U)]$,
where $L(U)$ denotes the random variable corresponding to the length of $U$ in bits, and 
minimization is taken over all $(p_{U|X},p_{{Z}|UY})$ that perfectly securely computes $(q_{XY},q_{Z|XY})$ with privacy against Alice. Then, the following upper and lower bounds hold.

\begin{theorem}\label{thm:ps_rate-2}
For any $(q_{XY},q_{Z|XY})$, 
\begin{center}$H_{\chi}(\mathcal{G},X)\leq L^{*}_{A-\emph{pvt}}< H_{\chi}(\mathcal{G},X)+1$,\end{center}
where chromatic entropy $H_{\chi}(.,.)$ is defined in Definition \ref{defn:chromatic-entropy}.
\end{theorem}
\begin{proof}
We first show that every proper coloring of the vertices of $\mathcal{G}$ corresponds to a perfectly secure code for 
$(q_{XY},q_{Z|XY})$ with privacy against Alice, and every perfectly secure code for $(q_{XY},q_{Z|XY})$ with privacy against Alice
corresponds to a collection of proper colorings of the vertices of $\mathcal{G}$.
This can be proved along the lines of the proof of Claim \ref{claim:color-equiv-protocol}.
Note that Claim~\ref{claim:color-equiv-protocol} is for privacy against both users, but we did not use privacy against Bob in proving that -- 
because Claim~\ref{claim:color-equiv-protocol} is for $(q_{X_{\text{EQ}}Y},q_{Z|X_{\text{EQ}}Y})$, and when we transform our problem from $(q_{XY},q_{Z|XY})$ 
to $(q_{X_{\text{EQ}}Y},q_{Z|X_{\text{EQ}}Y})$, privacy against Bob becomes redundant. 
Now, Theorem \ref{thm:ps_rate-2} can be proved along the lines of the proof of Theorem \ref{thm:ps_rate-1}. 
\end{proof}
\RRB{Note that Remark~\ref{remark:wlog} holds verbatim here also, i.e., the above results hold for the multiple-round communication model as well.} 

\noindent \textbf{Multiple Instances.} The setting is analogous to the multiple instances setting described before, except that we do not require privacy against Bob.
Note that we want to securely compute $(q_{X^nY^n},q_{Z^n|X^nY^n})$ with privacy against Alice, where $(X_i,Y_i)$'s are i.i.d. according 
to $q_{XY}$ and $q_{Z^n|X^nY^n}$ is in product form, i.e., $q_{Z^n|X^nY^n}(z^n|x^n,y^n)=\prod_{i=1}^n q_{Z|XY}(z_i|x_i,y_i)$.
Similar to the way we defined $\mathcal{G}_{EQ}^n$, we can define the graph $\mathcal{G}^n=(\mathcal{X}^n,\mathcal{E}^n)$ by extending the definition of $\mathcal{G}$ analogously: $\mathcal{E}^n=\{\{x^n,x'^n\}\in\mathcal{X}^n\times \mathcal{X}^n:$ $\exists (y^n,z^n)\in\mathcal{Y}^n\times\mathcal{Z}^n$\text{ s.t.} $q_{X^nY^n}(x^n,y^n)>0$, $q_{X^nY^n}(x'^n,y^n)>0$ \text{and} $q_{Z^n|X^nY^n}(z^n|x^n,y^n)\neq q_{Z^n|X^nY^n}(z^n|x'^n,y^n)\}$. We define, analogous to the single instance case, $L_{A-\text{pvt}-n}^*$ to be the 
optimal amortized expected length of the message that Alice sends to Bob in this setting.
The following is a simple corollary of Theorem~\ref{thm:ps_rate-2}. 
\begin{corollary}\label{cor:ps_rate-2-n}
Suppose $(q_{XY},q_{Z|XY})$ is perfectly securely computable with privacy against both users. Then
{\em\[\frac{1}{n} H_{\chi}(\mathcal{G}^n,X^n) \leq L_{A-\text{pvt}-n}^* < \frac{1}{n} (H_{\chi}(\mathcal{G}^n,X^n) + 1).\]}
\end{corollary}
Observe that $\mathcal{G}^n$ is neither the {\sc and}-product nor the {\sc or}-product of the graph $\mathcal{G}_{\text{EQ}}$;
and as far as we know, no single letter expression of $\frac{1}{n}H_{\chi}(\mathcal{G}_{\text{EQ}}^n,X_{\text{EQ}}^n)$ (not even of $\lim_{n\to\infty}\frac{1}{n}H_{\chi}(\mathcal{G}_{\text{EQ}}^n,X_{\text{EQ}}^n)$) is known.

\vspace{12pt}
\subsection{\RRB{Privacy required only against Bob}}
\RRB{Next, we consider the case where privacy only against Bob is required.
It turns out that not all $(q_{XY},q_{Z|XY})$ are securely computable in this setting.\footnote{\RRB{For example, consider secure computation of the {\sc and} function: Alice and Bob have binary inputs $X$ and $Y$, respectively,
and Bob wants to securely compute {\sc and}($X,Y$), the logical {\sc and} of the two input bits. 
Suppose that all the four possible input pairs can occur with non-zero probability. 
If Bob's input is 0 (which means that output is always 0), Alice's message cannot distinguish whether her input is 0 or 1, 
otherwise privacy against Bob will be compromised. On the other hand, if Bob's input is 1 (which means that output is always equal to $X$), 
Alice's message must distinguish whether her input is 0 or 1. Now, since Alice does not know Bob's input 
(because there is only one-way communication), Alice cannot send a message without violating either correctness or privacy against Bob.}}
First we explicitly characterize $(q_{XY},q_{Z|XY})$ 
that are perfectly securely computable with privacy against Bob, where $q_{XY}$ has full support.
None of the proofs in this section depends on the specific distribution of $q_{XY}$ as long as it has full support. So the characterization remains the same for all $q_{XY}$ that have full support.}

For every $y\in\Y$, define a set $\Z^{(y)}=\{z\in\Z:\exists x\in\X \text{ s.t. } q_{Z|XY}(z|x,y)>0\}$. Essentially, the set $\Z^{(y)}$ discards all those elements of $\Z$ that never appear as an output when Bob's input is $y$.
\begin{definition}\label{defn:equiv-relation}
For $y\in\Y$, define a relation $\equiv_y$ on the set $\Z^{(y)}$ as follows: for $z,z'\in\Z^{(y)}$, we say that $z\equiv_y z'$, if there is 
a constant $c>0$ such that $q_{Z|XY}(z|x,y)=c\cdot q_{Z|XY}(z'|x,y)$ for every $x\in\X$.
\end{definition}
It is easy to see that $\equiv_y$ is an equivalence relation for every $y\in\Y$, which partitions $\Z^{(y)}$ into equivalent classes. Consider a $y\in\Y$, and let $\Z^{(y)}=\Z_{1}^{(y)}\biguplus\Z_{2}^{(y)}\biguplus\hdots\biguplus\Z_{k(y)}^{(y)}$, where $k(y)$ is the number of equivalence classes in the partition generated by $\equiv_y$. 
For every equivalence class $\Z_{i}^{(y)}$ in this partition, define a $|\X|\times|\Z_{i}^{(y)}|$ matrix $A_i^{(y)}$ such that $A_i^{(y)}(x,z)=q_{Z|XY}(z|x,y)$ for every $(x,z)\in\X\times\Z_{i}^{(y)}$. 
Note that, for every $i\in[k(y)]$, all the columns of $A_i^{(y)}$ are multiples of each other. 
Since $A_i^{(y)}$ is a rank-one matrix, we can write it as $A_i^{(y)}=\vec{\alpha}_i^{(y)}\vec{\gamma}_i^{(y)}$, where $\vec{\alpha}_i^{(y)}$ is a column vector (whose entries are non-negative and sum up to one, which makes it a unique probability vector) and $\vec{\gamma}_i^{(y)}$ is a row vector. Entries of $\vec{\alpha}_i^{(y)}$ and $\vec{\gamma}_i^{(y)}$ are indexed by $x\in\X$ and $z\in\Z_{i}^{(y)}$, respectively. 
So $A_i^{(y)}(x,z)=\vec{\alpha}_i^{(y)}(x)\vec{\gamma}_i^{(y)}(z)$.
Note that if $\vec{\alpha}_i^{(y)}=\vec{\alpha}_j^{(y)}$, then $i=j$. 

Suppose $(q_{XY},q_{Z|XY})$ is perfectly securely computable with privacy against Bob. This implies that there exists 
a pair of encoder and decoder $(p_{U|X},p_{{Z}|UY})$, which induces the joint distribution $p_{XYU{Z}}=q_{XY}p_{U|X}p_{{Z}|UY}$, 
that satisfies correctness (i.e., $p_{XY{Z}}=q_{XY}q_{Z|XY}$) 
and privacy against Bob (i.e., the Markov chain $U-(Y,{Z})-X$ holds).
Note that the random variable $U$ corresponds to the message that Alice sends to Bob.
We define a set $\U_{i}^{(y)}$ to be the set of all those messages that Alice can send to Bob, and when Bob, having $y$ as his input outputs an element of $\Z_{i}^{(y)}$, as follows:
\begin{equation}
\U_{i}^{(y)}=\{u\in\U:p(u,z|y)>0 \text{ for some }z\in\Z_{i}^{(y)}\}. \label{eq:one-round_set-msgs}
\end{equation}
Note that for every $y\in\Y$ and $i\in[1:k(y)]$, the probability vector $\vec{\alpha}_i^{(y)}$ corresponds to the equivalence class $\Z_{i}^{(y)}$. The following claim is proved in Appendix~\ref{appendixF}.

\begin{claim}\label{claim:disjoint-messages}
Consider any $y,y'\in\Y$ and $i\in[1:k(y)]$, $j\in[1:k(y')]$. If $\vec{\alpha}_i^{(y)}\neq \vec{\alpha}_j^{(y')}$, then $\U_{i}^{(y)}\cap\U_{j}^{(y')}=\phi$.
\end{claim}
With the help of Claim \ref{claim:disjoint-messages} we can show that for every $y,y'\in\Y$, the corresponding collections of 
probability vectors $\{\vec{\alpha}_i^{(y)}:i\in[1:k(y)]\}$ and $\{\vec{\alpha}_j^{(y')}:j\in[1:k(y')]\}$ are equal. The following claim is proved in Appendix~\ref{appendixF}.

\begin{claim}\label{claim:equal-alpha-vectors}
For all $y,y'\in\Y$, we have the following:
\begin{enumerate}
\item $k(y)=k(y')$. 
\item $\{\vec{\alpha}_1^{(y)},\vec{\alpha}_2^{(y)},\hdots,\vec{\alpha}_{k}^{(y)}\}=\{\vec{\alpha}_1^{(y')},\vec{\alpha}_2^{(y')},\hdots,\vec{\alpha}_{k}^{(y')}\}$, where $k=k(y)$ for any $y\in\Y$.
\end{enumerate}
\end{claim}
For ease of notation, without loss of generality, we rearrange the indices, to have $\vec{\alpha}_j^{(y')}=\vec{\alpha}_i^{(y)}$ if and only if $i=j$.
Now it follows from Claim \ref{claim:disjoint-messages} and Claim \ref{claim:equal-alpha-vectors} that the message set $\U$ and the alphabet $\Z^{(y)}$, for every $y\in\Y$, can be partitioned into $k$ parts as follows:
\begin{align}
\U &= \U_1\uplus\U_2\uplus\hdots\uplus\U_k, \notag \\
\Z^{(y)} &= \Z_{1}^{(y)}\uplus\Z_{2}^{(y)}\uplus\hdots\uplus\Z_{k}^{(y)}, \notag
\end{align}
where $\Z_{i}^{(y)} = \{z\in\Z^{(y)}: p(u,z|x,y)>0\text{ for some }x\in\X, u\in\U_i\}$. Note that the same $\U_i$ is used to define $\Z_{i}^{(y)}$ (and corresponds to $\vec{\alpha}_i^{(y)}$ also) for every $y\in\Y$.
Now we can state the characterization theorem, which is proved in Appendix~\ref{appendixG}.
\begin{theorem}\label{thm:one-round-characterization}
Suppose $q_{XY}$ has full support. Then $(q_{XY},q_{Z|XY})$ is perfectly securely computable with privacy against Bob if and only if the following holds for every $y,y'\in\Y${\em :}
\begin{enumerate}
\item $\{\vec{\alpha}_1^{(y)},\vec{\alpha}_2^{(y)},\hdots,\vec{\alpha}_{k}^{(y)}\}=\{\vec{\alpha}_1^{(y')},\vec{\alpha}_2^{(y')},\hdots,\vec{\alpha}_{k}^{(y')}\}$,
\item For any $i\in[k]$, $\sum_{z\in\Z_{i}^{(y)}}q_{Z|XY}(z|x,y)=\sum_{z\in\Z_{i}^{(y')}}q_{Z|XY}(z|x,y')$ for every $x\in\X$.
\end{enumerate}
\end{theorem}
We introduce a new random variable $W$, which takes values in $[k]$, and is jointly distributed with $(X,Y,Z)$ as follows: define $p_{XYWZ}(x,y,i,z):= 0$, if $p_{XYZ}(x,y,z)=0$; otherwise, define
\begin{align}
p_{XYWZ}&(x,y,i,z)\\
 &:= q_{XY}(x,y)\times p_{W|X}(i|x)\times p_{Z|WY}(z|i,y) \notag \\
&:=  q_{XY}(x,y)\times \Big(\sum_{z\in\Z_i^{(y)}}q_{Z|XY}(z|x,y)\Big)\nonumber\\ &\hspace{10pt}\times \Big(\mathbbm{1}_{\{z\in\Z_i^{(y)}\}}\times\frac{q_{Z|XY}(z|x,y)}{\sum_{z\in\Z_i^{(y)}}q_{Z|XY}(z|x,y)}\Big) \label{eq:W-interim2}.
\end{align}
Comments are in order: (i) We defined $p_{W|X}(i|x)$ to be $\sum_{z\in\Z_i^{(y)}}q_{Z|XY}(z|x,y)$ in \eqref{eq:W-interim2} -- this is a valid definition because $\sum_{z\in\Z_i^{(y)}}q_{Z|XY}(z|x,y)$ is same for all $y$ (see Theorem \ref{thm:one-round-characterization}). (ii) We defined $p_{Z|WY}(z|i,y)$ to be $\mathbbm{1}_{\{z\in\Z_i^{(y)}\}}\times\frac{q_{Z|XY}(z|x,y)}{\sum_{z\in\Z_i^{(y)}}q_{Z|XY}(z|x,y)}$ in \eqref{eq:W-interim2} -- this is also a valid definition because the matrix $A_i^{(y)}$ defined earlier is a rank-one matrix, and therefore, $\frac{q_{Z|XY}(z|\tilde{x},y)}{\sum_{z\in\Z_i^{(y)}}q_{Z|XY}(z|\tilde{x},y)}$ is same for all $\tilde{x}$'s for which $q_{Z|XY}(z|\tilde{x},y)>0$. 
It follows from \eqref{eq:W-interim2} that $p_{XYWZ}(x,y,i,z) = \mathbbm{1}_{\{z\in\Z_i^{(y)}\}}\times q_{XYZ}(x,y,z)$. 
Note that $W$ is a deterministic function of both $U$ as well as of $(Y,Z)$. 
Now we are ready to prove upper and lower bounds on optimal expected length of the message, which is defined as $L^*_{B-\text{pvt}}:=\min_{(p_{U|X},p_{{Z}|UY})}\mathbb{E}[L(U)]$,
where $L(U)$ denotes the random variable corresponding to the length of $U$ in bits, and 
minimization is taken over all $(p_{U|X},p_{{Z}|UY})$ that perfectly securely computes $(q_{XY},q_{Z|XY})$ with privacy only against Bob. 
\begin{theorem}\label{thm:ps_rate-3}
Suppose $(q_{XY},q_{Z|XY})$ is perfectly securely computable with privacy against Bob, where $q_{XY}$ has full support. Then
{\em\[H(W)\leq L^*_{B-\text{pvt}}< H(W)+1.\]}
\end{theorem}
\begin{proof}
{$L^*_{B-\text{pvt}}\geq H(W)$:} 
Fix an arbitrary pair of encoder and decoder $(p_{U|X},p_{{Z}|UY})$ that perfectly securely computes this $(q_{XY},q_{Z|XY})$ 
with privacy against Bob, i.e., the induced joint distribution $p_{XYU{Z}}=q_{XY}p_{U|X}p_{{Z}|UY}$ satisfies correctness 
(i.e., $q_{XY{Z}}=q_{XY}q_{Z|XY}$) and privacy against Bob (i.e., $U-(Y,{Z})-X$ is a Markov chain).
Fix a prefix-free encoding of $U$. 
Let $L$ denote the random variable corresponding to the length of this prefix-free encoding of $U$. 
\RRB{We show that $\mathbb{E}[L]\geq H(W)$ in Appendix~\ref{appendixG} in two steps. For this, first we show that $\mathbb{E}[L]\geq \mathbb{E}[L^\prime]$, where $L^\prime$ is the random variable corresponding to a prefix-free binary encoding of $W$ that we extract out from $L$. Then we use $\mathbb{E}[L^\prime]\geq H(W)$, which follows from the fact that the expected length of any prefix-free binary code for a random variable $W$ is lower bounded by $H(W)$~ \cite[Theorem 5.4.1]{CoverJ06}. Since this argument holds for any prefix-free encoding of $U$, we have $L^*_{B-\text{pvt}}\geq H(W)$.}

{$L^*_{B-\text{pvt}}<H(W)+1$:}
We need to show a pair of encoder and decoder $(p_{U|X},p_{{Z}|UY})$ that perfectly securely computes this $(q_{XY},q_{Z|XY})$ 
with privacy against Bob, i.e., the induced joint distribution $p_{XYU{Z}}=q_{XY}p_{U|X}p_{{Z}|UY}$ satisfies correctness 
(i.e., $p_{XY\hat{Z}}=q_{XY}q_{Z|XY}$) and privacy against Bob (i.e., $U-(Y,{Z})-X$ is a Markov chain).
Note that $U=W$ is a valid choice, where $p_{XYWZ}$ is described in \eqref{eq:W-interim2}. 
The expected length of an optimal prefix-free binary code of $W$ is upper-bounded by $H(W)+1$ \cite[Theorem 5.4.1]{CoverJ06},
which implies that $L^*_{B-\text{pvt}}<H(W)+1$.
\end{proof}

{\color{black}We conclude the single instance part of this section by giving an example where common randomness helps to improve the optimal rate, when privacy is required only against Bob (as mentioned in Section~\ref{Sect_Prelim}).
\begin{example}\label{example:prvagnstbob}
Let $X$ be an $m$-length ($m>1$) vector of mutually independent and uniform binary random variables, $X=(X_1,\dots,X_m)$ and $Y=\emptyset$. Also, let $Z=(J,X_J)$, where $J$ is a random variable uniformly distributed on $[1:m]$. This defines a function $(q_{XY},q_{Z|XY})$.

First we calculate the lower bound of $L^*_{B-\text{pvt}}$ for this example from the discussion about random variable $W$ above Theorem \ref{thm:ps_rate-3}. It follows from \eqref{eq:W-interim2} that $p_{XYWZ}(x,y,i,z) = \mathbbm{1}_{\{z\in\Z_i^{(y)}\}}\times q_{XYZ}(x,y,z)$. Marginalizing this over $x$ and using the fact that $Y=\emptyset$, we get, $p_{W,Z}(i,z)=\mathbbm{1}_{\{z\in\Z_i\}}q_Z(z)$. From the Definition \ref{defn:equiv-relation}, it is clear that $\Z_i$ is a singleton set for every $i\in [1:|\Z|]$. Now, fix an $i$ and marginalize $p_{W,Z}(i,z)$ over $z$, we get $p_W(i)=q_Z(z)$, such that $\Z_i=\{z\}$. Then we have $p_W(i)=\frac{1}{2m}$ since $q_Z(z)=\frac{1}{2m}$ for every $z$. This implies that $H(W)=\log 2m$. So, from Theorem \ref{thm:ps_rate-3} we have $L^*_{B-\text{pvt}}\geq \log 2m >1$. Whereas it turns out that a rate of $1$ is achievable when we use common randomness. To see this, let $W$ be a random variable uniformly distributed on $[1:m]$. Alice on seeing $W$ sends $X_W$ to Bob using only $1$ bit and Bob outputs $(W,X_W)$ since it already has access to $W$. It is easy to see that this protocol satisfies privacy against Bob. Thus, common randomness helps to improve the optimal rate, when privacy is required only against Bob.
\end{example}
}
\noindent \textbf{Multiple Instances.} The setting is analogous to the multiple instances setting described before, except that we require privacy against Bob here. We define, analogous to the single instance case, $L_{B-\text{pvt}-n}^*$ to be the 
optimal amortized expected length of the message that Alice sends to Bob in this setting. We give bounds on $L^*_{B-\text{pvt}-n}$ for the case when $q_{XY}$ has full support;
and we leave it open to give a tight bound on $L^*_{B-\text{pvt}-n}$ for general input distribution.
The following is a simple corollary of Theorem~\ref{thm:ps_rate-3}.
\begin{corollary}\label{cor:ps_rate-3-n}
Suppose $(q_{XY},q_{Z|XY})$ is perfectly securely computable with privacy against Bob, where $q_{XY}$ has full support. Then
{\em\[H(W)\leq L^*_{B-\text{pvt}-n}< H(W)+\frac{1}{n}.\]}
\end{corollary}

\RRB{\section{Discussion}\label{section:discussion}
We studied asymptotically secure {\em two-user} interactive function computation and gave single-letter characterizations of feasibility and rate region. Our characterization shows that asymptotically secure computability is equivalent to perfectly secure computability. In the multi-terminal setting, the communication graph plays an important role. Recently, Narayanan and Prabhakaran~\cite{NarayananP18} characterized communication graphs which allow perfectly secure computability of {\em all} functions among a subset of users. This generalizes the results of Ben-Or, Goldwasser, and Wigderson~\cite{BGW88} and Chaum, Cr\'{e}peau, and Damg{\aa}rd~\cite{CCD88} for complete communication graphs. Since the impossibility results in these works rely on a reduction of two-party secure computation, the same characterizations should hold for asymptotically secure computation.}

\RRB{However, characterizing functions which are securely computable {\em given} a communication graph remains open.  
 As we saw in Remark~\ref{remark:eavesdropper}, for the multi-terminal case, impossibility of perfectly secure computation need not imply the impossibility of asymptotically secure computation. Indeed, even for perfectly securely computation, the problem remains largely open~\cite{BlaserJLM06,Biemel07}. Finding a single-letter characterization for feasibility of asymptotically secure function computation among three or more terminals would indeed be interesting.}

%{\bf Gowtham:} please link to these papers
%\begin{verbatim}
%NarayananP18 https://link.springer.com/chapter/10.1007%2F978-3-030-03807-6_15
%BlaserJLM06 https://link.springer.com/article/10.1007/s00145-005-0329-x
%Biemel07 https://link.springer.com/article/10.1007/s00446-006-0010-0
%MajiRPTCC09 https://link.springer.com/chapter/10.1007/978-3-642-00457-5_16
%NascimentoW08 https://ieeexplore.ieee.org/document/4529286
%AhlswedeC13 https://link.springer.com/chapter/10.1007/978-3-642-36899-8_6
%PintoDMN11 https://ieeexplore.ieee.org/document/5961837
%\end{verbatim}

\RRB{We confined our attention to the case where the users have access to only private and common randomness. As mentioned earlier, it is well-known that all functions are securely computable if the users have access to non-trivially correlated random variables or non-trivial noisy channels (see~\cite{CREK88,MajiPR12}). Characterizing the rates of communication and stochastic resources needed would be an interesting direction of study. For obtaining the string oblivious transfer using a noisy channel, results on optimal rates of channel use have been found~\cite{NascimentoW08,AhlswedeC13,PintoDMN11}. 
We only considered honest-but-curious users. In the malicious setting, where a corrupt user may arbitrarily deviate from the protocol, Maji et al.~\cite{MajiRPTCC09} characterized prefectly securely computable functions. Studying feasibility and rate region for asymptotically secure computation in the malicious setting is another direction of interest.}

\RRB{We required that each user may obtain no (or only negligible) information about the other user's input and output except for what can be inferred from its own input and output. However, this may be too stringent in certain settings where only certain attributes of the data need to be kept private (see, e.g.,~\cite{BasciftciWI16,KalantariSS18}). Studying such problems under the framework of this paper is another possible direction of future work.}

\section*{Acknowledgement}
We thank the reviewers for their useful comments which helped improve the manuscript. Deepesh Data was supported by a Microsoft Research India Ph.D. Fellowship. G.~R.~Kurri was supported by a travel fellowship from the Sarojini Damodaran Foundation. J.~Ravi was supported by European Research Council (ERC) under the European Union's Horizon 2020 research and innovation programme (grant agreement number 714161).  
V. M. Prabhakaran was supported in part by a Ramanujan Fellowship from the Department
of Science and Technology, Government of India and by Information Technology Research Academy
(ITRA), Government of India under ITRA Mobile grant ITRA/15(64)/Mobile/USEAADWN/01. G. R. Kurri and V. M. Prabhakaran are supported in part by the Department of Science and Technology, Government of India, under an Indo-Israel grant DST/INT/ISR/P-16/2017. We acknowledge support of the Department of Atomic Energy, Government of India, under project no. 12-R\&{D}-TFR-5.01-0500.
\iftoggle{paper}{
	\input{appendix_5page}
}{
	%\onecolumn
	\appendices
\section{Proofs of Theorems~\ref{Feasibility_AB}, \ref{Thm_Rate_Region_AB} \& \ref{cutset} }\label{appendix:proofs_omitted}
\emph{Notation}: We use capital letter to denote a random pmf (see e.g., works \cite{YassaeeAG14}, \cite{Cuff13}), e.g.\ $P_X$. For any two sequences of random p.m.f.'s $P_X^{(n)}$ and $Q_X^{(n)}$ on $\mathcal{X}^{(n)}$ (where $\mathcal{X}^{(n)}$ is arbitrary and can differ from the cartesian product $\mathcal{X}^n$), we write $P_{X^{(n)}}\overset{\epsilon}{\approx}Q_{X^{(n)}}$ if $\mathbbm{E}\lVert P_{X^{(n)}}-Q_{X^{(n)}}\rVert_1<\epsilon$. For any two sequences of random p.m.f.'s $P_{X^{(n)}}$ and $Q_{X^{(n)}}$ on $\mathcal{X}^{(n)}$, we write $P_{X^{(n)}}\approx Q_{X^{(n)}}$ if $\lim_{n\rightarrow \infty}\mathbbm{E}\lVert P_{X^{(n)}}-Q_{X^{(n)}}\rVert_1=0$.

{\color{black}We prove Theorem~\ref{Thm_Rate_Region_AB} first and then prove Theorems~\ref{Feasibility_AB} \& \ref{cutset}.}
\begin{proof}[\textbf{Proof of Theorem \ref{Thm_Rate_Region_AB}}]
\emph{Achievability:}
The proof is broadly along the lines of the proof of interactive channel simulation of Yassaee et al.~\cite{YassaeeGA15}. We consider two protocols, protocol A and protocol B, where Protocol A corresponds to the source coding side of the problem and Protocol B corresponds to our original problem (with extra shared random variables, one for each round). Each of these protocols will induce a p.m.f. on random variables defined during the protocol. We impose a series of 
 constraints (independence constraints and constraints for the reliability of SW decoders) as done in \cite[Theorem 1]{YassaeeGA15} so that these two distributions become almost identical. Once we find such constraints making protocol A almost identical to protocol B, we investigate the correctness and privacy properties in addition to eliminating the extra shared randomness. {\color{black}The main {\em difference} with the achievability of channel simulation \cite{YassaeeGA15} is from \eqref{eqn1:5}-\eqref{eqn1:6} onwards where the requirement of privacy becomes relevant.}
 
 We start from source coding side of the problem by fixing a distribution $p(u_{[1:r]},x,y,z_{[1:2]})$ satisfying \eqref{Eq_AB_Markov1}-\eqref{Eq_AB_Markov_Secr2}. Let $(U_{[1:r]}^n,X^n,Y^n,Z_{[1:2]}^n)$ be i.i.d. with distribution $p(u_{[1:r]},x,y,z_{[1:2]})$. Messages $M_i$, extra shared randomness $F_i$ and the actual 
 shared randomness $w$ are created as bin indices of $U_{[1:r]}^n$ in the following way: random variables $F_1,M_1,w$ are three 
 independent bin indices of $U_1^n$. $F_i$ and $M_i$ are two independent bin indices of $(U_1^n,\dots,U_i^n)$ for $i\in[2:r]$. 
 The alphabets of $F_i$, $M_i$ and $w$ are $[1:2^{n\hat{R}_i}]$, $[1:2^{nR_i}]$ and $[1:2^{nR_0}]$ respectively. 
 Furthermore, we consider Slepian-Wolf decoders, for odd $i\in [1:r]$, for reconstructing $\hat{U}_i^n$ (an estimate of $U_i^n$) 
 from $(U_{[1:i-1]}^n,F_i,M_i,w,Y^n)$ and for even $i\in[1:r]$, for reconstructing $\hat{U}_i^n$ 
 from $(U_{[1:i-1]}^n,F_i,M_i,w,X^n)$. Note that the reliability and independence constraints would remain exactly the same as that of \cite[Theorem~1]{YassaeeGA15} 
 which are given by
 \begin{itemize}
  \item {\em Independence constraints:}
 \begin{align}
 R_0+\hat{R}_1&<H(U_1|X)\label{eqn1:IND1}\\
 \hat{R}_i&<H(U_i|Y,U_{[1:i-1]}), \text{for even} \ i\in[2:r],\label{eqn1:IND2}\\
  \hat{R}_i&<H(U_i|X,U_{[1:i-1]}), \text{for odd} \ i\in[2:r],\label{eqn1:IND3}\\
  \sum_{t=1}^i \hat{R}_t&<H(U_{[1:i]}|X,Y,Z_{[1:2]})\label{eqn1:IND4}.
 \end{align}
 \item {\em Constraints for the reliability of SW decoders:}
 \begin{align}
 R_1+R_0+\hat{R}_1&\geq H(U_1|Y),\label{eqn1:SW1}\\
 R_i+\hat{R}_i&\geq H(U_i|X,U_{[1:i-1]}), \text{for even} \ i\in[2:r],\label{eqn1:SW2}\\
  R_i+\hat{R}_i&\geq H(U_i|Y,U_{[1:i-1]}), \text{for odd} \ i\in[2:r].\label{eqn1:SW3}
 \end{align}
 \end{itemize}
Intuitively, one can understand the above constraints in the following way: Using the OSRB theorem \cite[Theorem~1]{YassaeeAG14} and noting that $U_1-X-Y$, the first 
independence constraint $\eqref{eqn1:IND1}$ ensures that $(F_1,w)$ are nearly uniformly distributed and mutually 
independent of $(X^n, Y^n)$. Similarly, noting that $U_i-(U_{[1:i-1]},X)-Y, \text{if}\  i\  \text{is odd}$ and 
$U_i-(U_{[1:i-1]},Y)-X, \text{if} \ i \ \text{is even}$, the independence constraints $\eqref{eqn1:IND2}-\eqref{eqn1:IND3}$ 
ensure that $F_i$ is nearly uniformly distributed and independent of $(X^n, Y^n, U_{[1:i-1]}^n)$ for $i\in[2:r]$. 
The last independence constraint ensures that $F_{[1:r]}$ is nearly independent of $(X^n, Y^n, Z^n_{[1:2]})$, which helps in eliminating the extra shared randomness. Note that $\hat{U}_1^n$ is 
reconstructed from $F_1, M_1,w$ and $Y^n$, where $F_1, M_1$ and $w$ are the random bin indices of $U_1^n$ 
created by Alice, and Bob requires a rate of $H(U_1|Y)$ from Alice to recover $U_1^n$. 
This corresponds to the first SW constraint $\eqref{eqn1:SW1}$. The other SW constraints $\eqref{eqn1:SW2}-\eqref{eqn1:SW3}$ 
are similar with $U_i^n$ being reconstructed from $F_i, M_i, U_{[1:i-1]}^n$ and $X^n$ or $Y^n$, where $F_i, M_i$ are 
the random bin indices of $U_{[1:i]}^n$ for $i\in[2:r]$.
Thus \eqref{eqn1:IND1}-\eqref{eqn1:SW3} guarantee us the following:
 \begin{align}
 P(x^n,y^n,z_{[1:2]}^n,f_{[1:r]})&\approx p^{\text{Unif}}(f_{[1:r]})p(x^n,y^n,z_{[1:2]}^n),\label{eqn1:1}
 \end{align}
 \begin{align}
  P(x^n&,y^n,z_{[1:2]}^n,m_{[1:r]},w,f_{[1:r]})\nonumber\\
  &\hspace{12pt}\approx \hat{P}(x^n,y^n,z_{[1:2]}^n,m_{[1:r]},w,f_{[1:r]}).\label{eqn1:2}
 \end{align}
 It follows from $\eqref{eqn1:1}$ and $\eqref{eqn1:2}$ that
 \begin{align}
 &\lim_{n\rightarrow\infty}\mathbb{E}_\mathcal{B}\Big[\lVert P(x^n,y^n,z_{[1:2]}^n,f_{[1:r]})\nonumber\\
 &\hspace{24pt}-p^{\text{Unif}}(f_{[1:r]})p(x^n,y^n,z_{[1:2]}^n)\rVert_1\nonumber\\
 &\hspace{3cm}+\nonumber\\
 &\lVert P(x^n,y^n,z_{[1:2]}^n,m_{[1:r]},w,f_{[1:r]})\nonumber\\
 &\hspace{24pt}-\hat{P}(x^n,y^n,z_{[1:2]}^n,m_{[1:r]},w,f_{[1:r]})\rVert_1\Big]=0,
 \end{align}
 where the expectation is taken over random binning. This implies that there exists a particular realization of the 
 random binning with the corresponding p.m.f. $p$ so that we can replace $P$ with $p$ and denote the resulting p.m.f.'s for protocols 
 A and B respectively with $p$ and $\hat{p}$, then
 \begin{align}
 p(x^n,y^n,z_{[1:2]}^n,f_{[1:r]})&\approx p^{\text{Unif}}(f_{[1:r]})p(x^n,y^n,z_{[1:2]}^n),\label{eqn1:3}
 \end{align}
 \begin{align}
 p(x^n&,y^n,z_{[1:2]}^n,m_{[1:r]},w,f_{[1:r]})\nonumber\\
 &\hspace{12pt}\approx \hat{p}(x^n,y^n,z_{[1:2]}^n,m_{[1:r]},w,f_{[1:r]})\label{eqn1:4}.
 \end{align}
{\color{black}\eqref{eqn1:3} alone was sufficient for channel simulation \cite{YassaeeGA15} as correctness was the only consideration there, but we need to utilize \eqref{eqn1:4} also in order to account for privacy. We do this with the help of the following claim.} From $\eqref{eqn1:3}$ and $\eqref{eqn1:4}$, we claim that there exists $f^*_{[1:r]}$ such that $p(f^*_{[1:r]})>0$ and
\begin{align}
 p(x^n,y^n,z_{[1:2]}^n|f^*_{[1:r]})&\approx p(x^n,y^n,z_{[1:2]}^n),\label{eqn1:5}
 \end{align}
 \begin{align}
 p(x^n&,y^n,z_{[1:2]}^n,m_{[1:r]},w|f^*_{[1:2]})\nonumber\\
 &\hspace{12pt}\approx \hat{p}(x^n,y^n,z_{[1:2]}^n,m_{[1:2]},w|f^*_{[1:2]})\label{eqn1:6}.
\end{align}
To see this, first we rewrite $\eqref{eqn1:3}$ and $\eqref{eqn1:4}$ as
\begin{align}
 p(x^n,y^n,z_{[1:2]}^n,f_{[1:r]})&\overset{\epsilon_n}{\approx} p^{\text{Unif}}(f_{[1:r]})p(x^n,y^n,z_{[1:2]}^n),\label{eqn1:3`}
 \end{align}
 \begin{align}
 p(x^n&,y^n,z_{[1:2]}^n,m_{[1:r]},w,f_{[1:r]})\nonumber\\
 &\hspace{12pt}\overset{\delta_n}{\approx} \hat{p}(x^n,y^n,z_{[1:2]}^n,m_{[1:r]},w,f_{[1:r]})\label{eqn1:4`},
\end{align}
where $\epsilon_n\rightarrow 0$ and $\delta_n\rightarrow 0$ as $n\rightarrow\infty$.

  Note that $\hat{p}(f_{[1:r]})=p^{\text{Unif}}(f_{[1:r]})$ and $\eqref{eqn1:3`}$ implies that $p(f_{[1:r]})\overset{\epsilon_n}{\approx} p^{\text{Unif}}(f_{[1:r]})$. 
  Now, using these it is easy to see that $\eqref{eqn1:3`}$ and $\eqref{eqn1:4`}$ imply 
\begin{align}
 p(f_{[1:r]})p(x^n,y^n,z_{[1:2]}^n|f_{[1:r]})&\overset{2\epsilon_n}{\approx} p(f_{[1:r]})p(x^n,y^n,z_{[1:2]}^n),\label{eqn1:7}
 \end{align}
 \begin{align}
 p(f_{[1:r]})&p(x^n,y^n,z_{[1:2]}^n,m_{[1:r]},w|f_{[1:r]})\nonumber\\
 &\hspace{12pt}\overset{\epsilon_n+\delta_n}{\approx} p(f_{[1:r]})\hat{p}(x^n,y^n,z_{[1:2]}^n,m_{[1:r]},w|f_{[1:r]})\label{eqn1:8}.
\end{align}
$\eqref{eqn1:7}$ and $\eqref{eqn1:8}$, respectively, imply
\begin{align}
\sum_{f_{[1:r]}}p(f_{[1:r]})&\lVert p(x^n,y^n,z_{[1:2]}^n|f_{[1:r]})-p(x^n,y^n,z_{[1:2]}^n)\rVert_1\nonumber\\
&\hspace{12pt}\leq 2\epsilon_n,\label{eqn1:9}
\end{align}
\begin{align}
\sum_{f_{[1:r]}}&p(f_{[1:r]})\lVert p(x^n,y^n,z_{[1:2]}^n,m_{[1:r]},w|f_{[1:r]})\nonumber\\
&\hspace{12pt}-\hat{p}(x^n,y^n,z_{[1:2]}^n,m_{[1:r]},w|f_{[1:r]})\rVert_1\leq \epsilon_n+\delta_n\label{eqn1:10}.
\end{align}
Adding $\eqref{eqn1:9}$ and $\eqref{eqn1:10}$ imply that there exists an instance $f^*_{[1:r]}$ such that $p(f^*_{[1:r]})>0$ and 
\begin{align*}
 p(x^n,y^n,z_{[1:2]}^n|f^*_{[1:r]})&\overset{3\epsilon_n+\delta_n}{\approx} p(x^n,y^n,z_{[1:2]}^n),
 \end{align*}
 \begin{align*}
 p(x^n&,y^n,z_{[1:2]}^n,m_{[1:r]},w|f^*_{[1:r]})\nonumber\\
 &\hspace{12pt}\overset{3\epsilon_n+\delta_n}{\approx} \hat{p}(x^n,y^n,z_{[1:2]}^n,m_{[1:r]},w|f^*_{[1:r]}).
\end{align*}
This proves the claim made in $\eqref{eqn1:5}-\eqref{eqn1:6}$. Marginalizing away $m_{[1:r]}$ and $w$ from $\eqref{eqn1:6}$ and using $\eqref{eqn1:5}$ gives us correctness, $\hat{p}(x^n,y^n,z_{[1:2]}^n|f^*_{[1:r]})\approx p(x^n,y^n,z_{[1:2]}^n)$.

For privacy, we first show that $p(x^n,y^n,z_{[1:2]}^n,m_{[1:r]},w,f_{[1:r]})=p(x^n,y^n,z_{[1:2]}^n)p(m_{[1:r]},w,f_{[1:r]}|x^n,z_1^n)$.
 \begin{align}
 &p(x^n,y^n,z_{[1:2]}^n,m_{[1:r]},w,f_{[1:r]})\nonumber\\
 &=\sum_{u_{[1:r]}^n}p(x^n,y^n,z_{[1:2]}^n,u_{[1:r]}^n,m_{[1:r]},w,f_{[1:r]})\nonumber\\
 &=\sum_{u_{[1:r]}^n}p(x^n,y^n,z_{[1:2]}^n)p(u_{[1:r]}^n|x^n,y^n,z_{[1:2]}^n)\nonumber\\
 &\hspace{12pt}\times p(m_{[1:r]},w,f_{[1:r]}|u_{[1:r]}^n)\nonumber\\
 &=\sum_{u_{[1:r]}^n}p(x^n,y^n,z_{[1:2]}^n)p(u_{[1:r]}^n|x^n,z_1^n)\nonumber\\
 &\hspace{12pt}\times p(m_{[1:r]},w,f_{[1:r]}|u_{[1:r]}^n)\label{ach_prv_1}\\
 &=\sum_{u_{[1:r]}^n}p(x^n,y^n,z_{[1:2]}^n)p(u_{[1:r]}^n,m_{[1:r]},w,f_{[1:r]}|x^n,z_1^n)\nonumber\\
 &=p(x^n,y^n,z_{[1:2]}^n)p(m_{[1:r]},w,f_{[1:r]}|x^n,z_1^n)\nonumber,
 \end{align}
 where \eqref{ach_prv_1} follows from the fact that $p(u_{[1:r]},x,y,z_{[1:2]})$ satisfies the Markov chain $U_{[1:r]}-(X,Z_1)-(Y,Z_2)$. Thus, we have
 \begin{align*}
 &I_p(M_{[1:r]},W,F_{[1:r]};Y^n,Z_2^n|X^n,Z_1^n)=0\\
 &\Rightarrow I_p(M_{[1:r]},W;Y^n,Z_2^n|X^n,Z_1^n,F_{[1:r]})=0\\
 &\Rightarrow I_p(M_{[1:r]},W;Y^n,Z_2^n|X^n,Z_1^n,F_{[1:r]}=f_{[1:r]})=0, \nonumber\\
 &\hspace{12pt}\forall  f_{[1:r]}, \text{s.t.}\  p(f_{[1:r]})>0\\
 &\Rightarrow I(M_{[1:r]},W;Y^n,Z_2^n|X^n,Z_1^n)|_{p(x^n,y^n,z_{[1:2]}^n,m_{[1:r]},w|f_{[1:r]})} \nonumber\\
 &\hspace{12pt}=0,\forall f_{[1:r]}, \text{s.t.}\  p(f_{[1:r]})>0\\
 &\Rightarrow I(M_{[1:r]},W;Y^n,Z_2^n|X^n,Z_1^n)|_{p(x^n,y^n,z_{[1:2]}^n,m_{[1:r]},w|f^*_{[1:r]})}\nonumber\\
 &\hspace{12pt}=0,
 \end{align*}
 where $f^*_{[1:r]}$ is fixed in $\eqref{eqn1:5}$.

 Similarly, it can be proved that $I(M_{[1:r]},W;X^n,Z_1^n|Y^n,Z_2^n)|_{p(x^n,y^n,z_{[1:2]}^n,m_{[1:r]},w|f^*_{[1:r]})}=0$ since $p(u_{[1:r]},x,y,z_{[1:2]})$ 
 satisfies the Markov chain $U_{[1:r]}-(Y,Z_2)-(X,Z_1)$ also. Now, we have
 \begin{align}
 &I(M_{[1:r]},W;Y^n,Z_2^n|X^n,Z_1^n)|_{p(x^n,y^n,z_{[1:2]}^n,m_{[1:r]},w|f^*_{[1:r]})}\nonumber\\
 &\hspace{12pt}=0,\label{eqn1:cont1}\\
 &I(M_{[1:r]},W;X^n,Z_1^n|Y^n,Z_2^n)|_{p(x^n,y^n,z_{[1:2]}^n,m_{[1:r]},w|f^*_{[1:2]})}\nonumber\\
 &\hspace{12pt}=0,\label{eqn1:cont2}\\
 &p(x^n,y^n,z_{[1:2]}^n,m_{[1:r]},w|f^*_{[1:r]})\nonumber\\
 &\hspace{12pt}\approx \hat{p}(x^n,y^n,z_{[1:2]}^n,m_{[1:r]},w|f^*_{[1:r]})\label{eqn1:cont3}.
 \end{align}
 Since mutual information is a continuous function of the probability distribution, $\eqref{eqn1:cont1}-\eqref{eqn1:cont3}$ imply 
 \begin{align*}
 &\frac{1}{n}I(M_{[1:r]},W;Y^n,Z_2^n|X^n,Z_1^n)|_{\hat{p}(x^n,y^n,z_{[1:2]}^n,m_{[1:r]},w|f^*_{[1:r]})}\nonumber\\
 &\hspace{12pt}\rightarrow 0 \ \text{as} \ n\rightarrow \infty,\\
 &\frac{1}{n}I(M_{[1:r]},W;X^n,Z_1^n|Y^n,Z_2^n)|_{\hat{p}(x^n,y^n,z_{[1:2]}^n,m_{[1:r]},w|f^*_{[1:r]})}\nonumber\\
 &\hspace{12pt}\rightarrow 0 \ \text{as} \ n\rightarrow \infty,
 \end{align*}
 using the following: if two random variables $A$ and $A'$ with 
  same support set $\mathcal{A}$  satisfy $||p_{A} -  p_{A'}||_{1} \leq \epsilon \leq 1/4 $, then it follows from \cite[Theorem 17.3.3]{CoverJ06} that $|H(A) - H(A')|\leq \eta \log |\mathcal{A}|$, where $\eta \rightarrow 0$ as $\epsilon \rightarrow 0$. Finally, eliminating $(\hat{R}_1,\dots \hat{R}_r)$ and $({R}_1,\dots {R}_r)$ by applying Fourier-Motzkin elimination to \eqref{eqn1:IND1}-\eqref{eqn1:SW3} along with $R_{12}=\sum_{i:\text{odd}}R_i$, $R_{21}=\sum_{i:\text{even}}R_i$, gives $\mathcal{R}^{AB-\text{pvt}}_A(r)$ by noticing that
\begin{align}
I(X;U_{[1:r]}|Y)&=I(X;Z_2|Y),\label{eqn:simplification_1}\\
I(Y;U_{[1:r]}|X)&=I(Y;Z_1|X),\label{eqn:simplification_2}\\
I(U_{[1:r]}; Z_1,Z_2|X,Y)&=I(Z_1;Z_2|X,Y)\label{eqn:simplification}.
\end{align}
This is because,
 \begin{align}
   I(X;U_{[1:r]}|Y)&= I(X;U_{[1:r]}|Y) +  I(X; Z_2|U_{[1:r]},Y)\label{eqn:main_proofs_achievability_simplification_2}\\
   & =  I(X; U_{[1:r]}, Z_2|Y) \nonumber\\
   & = I(X;Z_2|Y) + I(X; U_{[1:r]}|Z_2,Y)\nonumber\\
   & = I(X;Z_2|Y)\label{eqn:main_proofs_achievability_simplification_1},
  \end{align}
  where \eqref{eqn:main_proofs_achievability_simplification_2} follows from \eqref{Eq_AB_Markov_Decod2}, and \eqref{eqn:main_proofs_achievability_simplification_1} follows from \eqref{Eq_AB_Markov_Secr2}.
  
Similarly, we can show that $I(Y;U_{[1:r]}|X)=I(Y;Z_1|X)$. Also, 
\begin{align}
  &I(U_{[1:r]}; Z_1,Z_2|X,Y) \nonumber\\
  &=  I(U_{[1:r]}; Z_1|X,Y) +  I(U_{[1:r]}; Z_2|X,Y, Z_1)\nonumber\\
  & = I(U_{[1:r]}; Z_1|X,Y)\label{eqn:main_proofs_achievability_simplification_3}\\
  & = I(U_{[1:r]}; Z_1|X,Y)+ I(Z_2; Z_1|U_{[1:r]},X,Y)\label{eqn:main_proofs_achievability_simplification_4}\\
  & = I(U_{[1:r]}, Z_2; Z_1|X,Y)\nonumber\\
  & = I(Z_1; Z_2|X,Y)+ I(U_{[1:r]}; Z_1|Z_2,X,Y)\nonumber\\
  & = I(Z_1; Z_2|X,Y)\label{eqn:main_proofs_achievability_simplification_5},
\end{align}
 where \eqref{eqn:main_proofs_achievability_simplification_3} follows from \eqref{Eq_AB_Markov_Secr1}, \eqref{eqn:main_proofs_achievability_simplification_4} follows from \eqref{Eq_AB_Markov_Decod2}, and \eqref{eqn:main_proofs_achievability_simplification_5} follows from \eqref{Eq_AB_Markov_Secr2}. 
 This completes the achievability proof of the theorem.

{\em Converse:} Suppose a rate triple $(R_0,R_{12},R_{21})$ is achievable with privacy against both Alice and Bob. Then, there exists a sequence of $(n,R_0,R_{12},R_{21})$ protocols such that for every $\epsilon>0$, there exists a large enough $n$ such that
\begin{align}
\left\lVert p_{X^n,Y^n,Z_1^n,Z_2^n}-q^{(n)}_{X,Y,Z_1,Z_2}\right\rVert&\leq \epsilon, \label{Eq_Sim_Prob} \\
I(M_{[1:r]},W;Y^n,Z_2^n|X^n,Z_1^n)&\leq n\epsilon, \label{Eq_Conv_Markov1} \\
I(M_{[1:r]},W;X^n,Z_1^n|Y^n,Z_2^n)&\leq n\epsilon. \label{Eq_Conv_Markov2}
\end{align}
%   \begin{align}
% %   ||p_{ X^{n}Z_1^{n}Y^{n}Z_2^{n}} - q_{ X^{n}Z_1^{n}Y^{n}Z_2^{n}}||_{1} & \leq \epsilon \label{Eq_L1_Dist} \\
%    I(M_{[1:r]},W; Y^{n},Z_2^{n} |X^{n}Z_1^{n}) & \leq n\epsilon \label{Eq_Priv_Alice} \\
%    I(M_{[1:r]},W; X^{n},Z_1^{n} |Y^{n}Z_2^{n}) & \leq n\epsilon. \label{Eq_Priv_Bob} 
% %      M_{[1:r]}- X^{n} - Y^{n} & \label{Eq_Markov_Alice}
%   \end{align}
{\color{black}Fix an $\epsilon\in(0,\frac{1}{4}]$. Let $T$ be a random variable uniformly distributed on $[1:n]$ and independent of all other random variables. Define $U_{j}:= (M_j,W, X_{}^{T+1:n}, Y^{1:T-1},T )$, for $1 \leq j \leq r$, and $X:=X_T,Y:=Y_T,Z_1:=Z_{1T},Z_2:={Z_{2T}}$. First, we show that this choice of auxiliary random variables satisfies the privacy conditions by single-letterizing \eqref{Eq_Conv_Markov1} and \eqref{Eq_Conv_Markov2}. }
Privacy condition against Alice $\eqref{Eq_Conv_Markov1}$ implies that
  \begin{align}
   &n\epsilon \nonumber\\
    &\geq I(M_{[1:r]},W; Y^{n},Z_2^{n}  \;| \; X^{n},Z_1^{n})\nonumber\\
   & = H_{p}(Y^{n}, Z_2^{n}|X^{n}, Z_1^{n}) - H_{p}(Y^{n}, Z_2^{n} | X^{n}, Z_1^{n}, M_{[1:r]},W) \nonumber\\
   & \geq\sum_{i=1}^{n} [H_{q}(Y_{i}, Z_{2i}|X_{i}, Z_{1i}) - \epsilon_1 \nonumber\\
   &\hspace{12pt} -H_{p}(Y_{i}, Z_{2i}| Y_{}^{1:i-1}, Z_{2}^{1:i-1},X^{n}, Z_1^{n},M_{[1:r]},W )]\label{eqn:PAA_1}\\
   & \geq \sum_{i=1}^{n} [H_{q}(Y_{i}, Z_{2i}|X_{i}, Z_{1i}) - \epsilon_1 \nonumber\\
   &\hspace{12pt}-  H_{p}(Y_{i}, Z_{2i}| X_{}^{i+1:n}, Y^{1:i-1}, X_{i}, Z_{1i}, M_{[1:r]},W)]\nonumber\\
    & \geq \sum_{i=1}^{n} [H_{p}(Y_{i}, Z_{2i}|X_{i}, Z_{1i}) - \epsilon_1 - \epsilon_2 \nonumber\\
    &\hspace{12pt}  -H_{p}(Y_{i}, Z_{2i}|X_{}^{i+1:n}, Y^{1:i-1}, X_{i}, Z_{1i},M_{[1:r]},W)]\label{eqn:PAA_2}\\
   &   {=} \sum_{i=1}^{n} \big[I(Y_{i}, Z_{2i} ;  X_{}^{i+1:n}, Y^{1:i-1}, M_{[1:r]},W | X_{i}, Z_{1i}) \nonumber\\
   &\hspace{12pt}- \epsilon_1 - \epsilon_2\big]\nonumber\\
   & = n\big[ I(Y_{T}, Z_{2T}; X_{}^{T+1:n}, Y^{1:T-1}, M_{[1:r]},W |X_{T},Z_{1T},T)\nonumber\\
   &\hspace{12pt}- \epsilon_1 - \epsilon_2 \big] \nonumber\\
    & =n\big[ I(Y_{T}, Z_{2T};  X_{}^{T+1:n}, Y^{1:T-1}, M_{[1:r]},W, T|X_{T},Z_{1T}) \nonumber\\
    &\hspace{12pt}- I(Y_{T}, Z_{2T}; T|X_{T},Z_{1T})- \epsilon_1 - \epsilon_2\big]\nonumber\\
     & \geq n\big[ I(Y_{T}, Z_{2T};  X_{}^{T+1:n}, Y^{1:T-1}, M_{[1:r]},W, T|X_{T},Z_{1T}) \nonumber\\
     &\hspace{12pt}- \epsilon_1 - \epsilon_2 - \epsilon_3 \big]\label{eqn:PAA_3}\\
     & = n\left[ I(Y, Z_2; U_{[1:r]}|X,Z_1) -\delta\right]\label{eqn:PAA_4},
  \end{align}
    where \eqref{eqn:PAA_1}-\eqref{eqn:PAA_3} follow from the following fact: if two random variables $A$ and $A'$ with 
  same support set $\mathcal{A}$  satisfy $||p_{A} -  p_{A'}||_{1} \leq \epsilon \leq 1/4 $, then it follows from \cite[Theorem 17.3.3]{CoverJ06} that $|H(A) - H(A')|\leq \eta \log |\mathcal{A}|$, where $\eta \rightarrow 0$ as $\epsilon \rightarrow 0$. In \eqref{eqn:PAA_4}, $\delta:=\epsilon_1+\epsilon_2+\epsilon_3$ and $\delta\rightarrow 0$ as $\epsilon\rightarrow 0$. 
  
Privacy condition against Bob $\eqref{Eq_Conv_Markov2}$ implies that
  \begin{align}
   &n\epsilon \nonumber\\
    &\geq I(M_{[1:r]},W; X^{n},Z_1^{n}  \;| \; Y^{n},Z_2^{n})\nonumber\\
   & = H_p(X^{n}, Z_1^{n}|Y^{n}, Z_2^{n}) - H_p(X^{n}, Z_1^{n} | Y^{n}, Z_2^{n},M_{[1:r]},W) \nonumber\\
   &=   H_p(X^{n}, Z_1^{n}|Y^{n}, Z_2^{n})  \nonumber\\
   &\hspace{12pt}-\sum_{i=1}^{n} H_p(X_{i}, Z_{1i}|X_{}^{i+1:n}, Z_{1}^{i+1:n},Y^{n}, Z_2^{n},M_{[1:r]},W)]\label{eqn:proof_prv_bob_1}\\
   & \geq \sum_{i=1}^{n} [H_q(X_{i}, Z_{1i}|Y_{i}, Z_{2i}) - \epsilon_1^\prime \nonumber\\
   &\hspace{12pt}-  H_p(X_{i}, Z_{1i}|X_{}^{i+1:n}, Z_{1}^{i+1:n},Y^{n}, Z_2^{n},M_{[1:r]},W)]\label{eqn:PAB_1}\\
   & \geq \sum_{i=1}^{n} [H_q(X_{i}, Z_{1i}|Y_{i}, Z_{2i}) -\epsilon_1^\prime\nonumber\\
   &\hspace{12pt} - H_p(X_{i}, Z_{1i}|Y_{}^{i+1:n},X_{}^{i+1:n}, Y^i, Z_{2i},M_{[1:r]},W)]\nonumber\\
    &\geq\sum_{i=1}^{n} [H_p(X_{i}, Z_{1i}|Y_{i}, Z_{2i}) - \epsilon_1^\prime - \epsilon_2^\prime \nonumber\\
    &\hspace{12pt}- H_p(X_{i}, Z_{1i}|X_{}^{i+1:n}, Y^{1:i-1}, Y_{i}, Z_{2i}, M_{[1:r]},W)]\label{eqn:PAB_2}\\
   & = \sum_{i=1}^{n} [I(X_{i}, Z_{1i} ; X_{}^{i+1:n}, Y^{1:i-1}, M_{[1:r]},W | Y_{i}, Z_{2i}) \nonumber\\
   &\hspace{12pt} - \epsilon_1^\prime - \epsilon_2^\prime]\nonumber\\
   & =  n\big[ I(X_{T}, Z_{1T}; X_{}^{T+1:n}, Y^{1:T-1}, M_{[1:r]},W|Y_{T},Z_{2T},T) \nonumber\\
   &\hspace{12pt} - \epsilon_1^\prime - \epsilon_2^\prime \big]\nonumber\\
   & =  n\big[ I(X_{T}, Z_{1T};X_{}^{T+1:n}, Y^{1:T-1},  M_{[1:r]},W, T|Y_{T},Z_{2T}) \nonumber\\
   &\hspace{12pt}- I(X_{T}, Z_{1T}; T|Y_{T},Z_{2T}) - \epsilon_1^\prime - \epsilon_2^\prime \big]\nonumber\\
   & \geq n\left[ I(X_{}, Z_{1};  U_{[1:r]}|Y_{},Z_{2})  - \epsilon_1^\prime - \epsilon_2^\prime - \epsilon_3^\prime \right]\label{eqn:PAB_3}\\
   & = n\left[ I(X_{}, Z_{1};  U_{[1:r]}|Y_{},Z_{2}) -\delta^\prime\right]\label{eqn:PAB_4},
  \end{align}
where \eqref{eqn:proof_prv_bob_1} follows from the chain rule, $H(A^n)=\sum_{i=1}^nH(A_i|A^{i+1:n})$ and \eqref{eqn:PAB_1}-\eqref{eqn:PAB_4} follow due to the same reasons as that of \eqref{eqn:PAA_1}-\eqref{eqn:PAA_4} with $\delta^\prime:=\epsilon_1^\prime+\epsilon_2^\prime+\epsilon_3^\prime$ so that, $\delta^\prime\rightarrow 0$ as $\epsilon\rightarrow 0$.

It can be shown along similar lines as \cite[Theorem 1]{YassaeeGA15} that
\begin{align*}
   R_{12} &\geq I(X;U_{[1:r]}|Y),\\
   R_{21} &\geq I(Y;U_{[1:r]}|X),\\
  R_{0} + R_{12} &\geq I(X;U_{[1:r]}|Y) + I(U_1;Z_1,Z_2|X,Y)\nonumber\\
    &\hspace{12pt}-3g(\epsilon),\\
    R_{0} + R_{12} + R_{21} &\geq I(X;U_{[1:r]}|Y) + I(Y;U_{[1:r]}|X) \nonumber\\
    &\hspace{12pt}+ I(U_{[1:r]}; Z_1,Z_2|X,Y)-3g(\epsilon),
   \end{align*}
for conditional p.m.f. $p(u_{[1:r]},z_1,z_2|x,y)$ satisfying 
     \begin{align}
     \lVert p_{X,Y,Z_1,Z_2}-q_{X,Y,Z_1,Z_2}\rVert\leq \epsilon,\label{eqn_proofs_AB_single_letter_1}\\
      U_i-(U_{[1:i-1]},X)-Y, \text{if}\  i\  \text{is odd},\label{eqn_proofs_AB_single_letter_2}\\
      U_i-(U_{[1:i-1]},Y)-X, \text{if} \ i \ \text{is even},\label{eqn_proofs_AB_single_letter_3}\\
      Z_1-(U_{[1:r]},X)-(Y,Z_2),\label{eqn_proofs_AB_single_letter_4}\\
    Z_2-(U_{[1:r]},Y)-(X,Z_1)\label{eqn_proofs_AB_single_letter_5},
     \end{align}  
where $g(\epsilon)\rightarrow 0$ as $\epsilon\rightarrow 0$.

So, we have shown that $(R_0,R_{12},R_{21})\in\mathcal{S}_\epsilon(r)$ for every $\epsilon>0$, where $\mathcal{S}_\epsilon(r)$ is defined to be the set of all non-negative rate triples $(R_0,R_{12},R_{21})$ for which there exists a p.m.f. $p(u_{[1:r]},x,y,z_{[1:2]})$ such that
\begin{align}
  R_{12} &\geq I(X;U_{[1:r]}|Y)\nonumber\\
   R_{21} &\geq I(Y;U_{[1:r]}|X)\nonumber\\
  R_{0} + R_{12} &\geq I(X;U_{[1:r]}|Y)\nonumber\\
    &\hspace{12pt} + I(U_1;Z_1,Z_2|X,Y)-3g(\epsilon)\nonumber\\
    R_{0} + R_{12} + R_{21} &\geq I(X;U_{[1:r]}|Y)+ I(Y;U_{[1:r]}|X)  \nonumber\\
    &\hspace{12pt}+ I(U_{[1:r]}; Z_1,Z_2|X,Y)-3g(\epsilon),\nonumber\\
    \lVert p_{X,Y,Z_1,Z_2}-q_{X,Y,Z_1,Z_2}\rVert&\leq \epsilon,\nonumber\\
  U_i-(U_{[1:i-1]},X&)-Y, \text{if}\  i\  \text{is odd},\nonumber\\
 U_i-(U_{[1:i-1]},Y&)-X, \text{if} \ i \ \text{is even},\nonumber\\
   Z_1-(U_{[1:r]}&,X)-(Y,Z_2),\nonumber\\
 Z_2-(U_{[1:r]}&,Y)-(X,Z_1),\nonumber\\
    I(U_{[1:r]};Y,Z_2|X,Z_1)&\leq h_1(\epsilon),\nonumber\\
I(U_{[1:r]};X,Z_1|Y,Z_2)&\leq h_2(\epsilon),
\end{align}
where $g(\epsilon),h_1(\epsilon)$ and $h_2(\epsilon)\rightarrow 0$ as $\epsilon\rightarrow 0$.

Now we argue that imposing the cardinality bounds $|\mathcal{U}_1|\leq |\mathcal{X}||\mathcal{Y}||\mathcal{Z}_1||\mathcal{Z}_2|+5$,
$|\mathcal{U}_i|\leq |\mathcal{X}||\mathcal{Y}||\mathcal{Z}_1||\mathcal{Z}_2|\prod_{j=1}^{i-1}|\mathcal{U}_j|+4$, 
$\forall i>1$ on the p.m.f. does not alter the set $\mathcal{S}_\epsilon(r)$. The argument is along similar lines as \cite[Appendix D]{YassaeeGA15} which uses the Convex Cover Method \cite[Appendix C]{GamalK12}. The alphabet $\mathcal{U}_1$ should have $|\mathcal{X}||\mathcal{Y}||\mathcal{Z}_1||\mathcal{Z}_2|-1$ elements to preserve the joint distribution $p_{X,Y,Z_1,Z_2}$, which in turn preserves $H(X|Y)$, $H(Y|X)$, $H(Z_1,Z_2|X,Y)$, $H(Y,Z_2|X,Z_1)$, $H(X,Z_1|Y,Z_2),$ and 6 more elements to preserve $H(X|U_{[1:r]},Y)$, $H(Y|U_{[1:r]},X)$, $H(Z_1,Z_2|U_1,X,Y)$, $H(Z_1,Z_2|U_{[1:r]},X,Y)$, $H(Y,Z_2|U_{[1:r]},X,Z_1)$, $H(X,Z_1|U_{[1:r]},Y,Z_2)$. The alphabet $\mathcal{U}_i$, $i>1$, should have $|\mathcal{X}||\mathcal{Y}||\mathcal{Z}_1||\mathcal{Z}_2|\prod_{j=1}^{i-1}|\mathcal{U}_j|-1$ elements to preserve the joint distribution $p_{X,Y,Z_1,Z_2,U_{[1:i-1]}}$, which in turn preserves $H(X|Y)$, $H(Y|X)$, $H(Z_1,Z_2|X,Y)$, $H(Y,Z_2|X,Z_1)$, $H(X,Z_1|Y,Z_2)$, $H(Z_1,Z_2|U_1,X,Y),$ and 5 more elements to preserve $H(X|U_{[1:r]},Y)$, $H(Y|U_{[1:r]},X)$, $H(Z_1,Z_2|U_{[1:r]},X,Y)$, $H(Y,Z_2|U_{[1:r]},X,Z_1),$ and $H(X,Z_1|U_{[1:r]},Y,Z_2).$  

Using the continuity of mutual information and total variation distance in the probability simplex, 
it can be shown (similar to \cite[Lemma 6]{YassaeeGA15}) that $\bigcap_{\epsilon > 0} \cS_{\epsilon}(r) $ is equal to $\mathcal{R}^{AB-\text{pvt}}_A(r)$ by invoking \eqref{eqn:simplification_1}-\eqref{eqn:simplification}.
This concludes the proof of Theorem~\ref{Thm_Rate_Region_AB}.
\end{proof}

\begin{proof}[\textbf{Proof of Theorem~\ref{Feasibility_AB}}]

\emph{Proof of part $(i)$}: It is trivial to see the `if' part since \eqref{Eq_AB_Markov1}-\eqref{Eq_AB_Markov_Secr2} define an $r$-round perfectly secure protocol of blocklength one, i.e., the protocol satisfies \eqref{eqn:asymptotic_1}-\eqref{eqn:asymptotic_3} with $n=1$ and $\epsilon=0$. For the `only if' part, it suffices to show that every $(q_{X,Y},q_{Z_1,Z_2|X,Y})$ that is computable with asymptotic security, 
is also computable with perfect security. Assume that $(q_{X,Y},q_{Z_1,Z_2|X,Y})$ is asymptotically securely 
computable with privacy against both users. This implies that for every $\epsilon>0$, 
there exists a large enough $n$ and a protocol $\Pi_n$ that satisfies the following conditions:
\begin{align}
&\left\lVert p^{\text{(induced)}}_{X^n,Y^n,Z_1^n,Z_2^n}-q_{X^n,Y^n,Z_1^n,Z_2^n}\right\rVert_1\leq \epsilon,\nonumber\\
 &M_i-(M_{[1:i-1]},X^n)-Y^n, \text{if}\  i\  \text{is} \ odd,\nonumber\\
     & M_i-(M_{[1:i-1]},Y^n)-X^n, \text{if} \ i \ \text{is}\  \text{even},\nonumber\\
      &Z_1^n-(M_{[1:r]},X^n)-(Y^n,Z_2^n),\nonumber\\
    &Z_2^n-(M_{[1:r]},Y^n)-(X^n,Z_1^n),\nonumber\\
&I(M_{[1:r]},W;Y^n,Z_2^n|X^n,Z_1^n)\leq n\epsilon,\nonumber\\
&I(M_{[1:r]},W;X^n,Z_1^n|Y^n,Z_2^n)\leq n\epsilon.\label{eqn_proofs_feasibility_asymptotic_protocol_7}
\end{align}
 Single-letterization as done in the converse of the Theorem \ref{Thm_Rate_Region_AB} implies that for every 
 $\epsilon^\prime>0$ (by suitably selecting $\epsilon$ 
 in \eqref{eqn_proofs_feasibility_asymptotic_protocol_7}), there
 exists a p.m.f. $p(u_{[1:r]},x,y,z_1,z_2)$  which satisfies the following conditions.
\begin{align}
 &\lVert p_{X,Y,Z_1,Z_2}-q_{X,Y,Z_1,Z_2}\rVert_1\leq \epsilon^\prime,\nonumber\\
  &U_i-(U_{[1:i-1]},X)-Y, \text{if}\  i\  \text{is}\  \text{odd},\nonumber\\
     & U_i-(U_{[1:i-1]},Y)-X, \text{if} \ i \ \text{is}\  \text{even},\nonumber\\
     & Z_1-(U_{[1:r]},X)-(Y,Z_2),\nonumber\\
  &  Z_2-(U_{[1:r]},Y)-(X,Z_1),\nonumber\\
   &I(U_{[1:r]};X,Z_1|Y,Z_2)\leq \epsilon^\prime,\nonumber\\
     & I(U_{[1:r]};Y,Z_2|X,Z_1)\leq \epsilon^\prime\label{eqn_proofs_feasibility_perfect_protocol_7}.
\end{align}  
For a fixed  $(q_{X,Y},q_{Z_1,Z_2|X,Y})$, consider the following set.
\begin{align}\label{eqn_proofs_feasibility_compact}
\mathcal{S}_\epsilon:=\big\{p_{U_{[1:r]},Z_1,Z_2|X,Y}:\  &\lVert p_{X,Y,Z_1,Z_2}-q_{X,Y,Z_1,Z_2}\rVert_1\leq \epsilon,\nonumber\\
 &U_i-(U_{[1:i-1]},X)-Y, \text{if}\  i\  \text{is}\  \text{odd},\nonumber\\
    &  U_i-(U_{[1:i-1]},Y)-X, \text{if} \ i \ \text{is}\  \text{even},\nonumber\\
    &  Z_1-(U_{[1:r]},X)-(Y,Z_2),\nonumber\\
    &Z_2-(U_{[1:r]},Y)-(X,Z_1),\nonumber\\
 &  I(U_{[1:r]};X,Z_1|Y,Z_2)\leq \epsilon,\nonumber\\
    &  I(U_{[1:r]};Y,Z_2|X,Z_1)\leq \epsilon\big\}.
\end{align}
From \eqref{eqn_proofs_feasibility_perfect_protocol_7}, it is clear that $\mathcal{S}_\epsilon$ is non-empty for every $\epsilon>0$.
For a monotonically decreasing sequence $\epsilon_1>\epsilon_2>\dots$ with $\lim_{k\rightarrow\infty}\epsilon_k=0$, 
it is easy to see that $\mathcal{S}_{\epsilon_1}\supseteq\mathcal{S}_{\epsilon_2}\supseteq\dots$.
Now, in order to prove that $(q_{X,Y},q_{Z_1,Z_2|X,Y})$ is computable with perfect security it suffices to 
show that $\mathcal{S}_0$ is non-empty. We show this by first arguing that 
$\lim_{k\rightarrow\infty}\bigcap_{i=1}^k \mathcal{S}_{\epsilon_k}$ is non-empty and then proving 
that $\mathcal{S}_0=\lim_{k\rightarrow\infty}\bigcap_{i=1}^k \mathcal{S}_{\epsilon_k}$. Using the continuity 
of total variation distance and continuity of mutual information we show  below that $\mathcal{S}_\epsilon$ is compact (bounded and closed) for every $\epsilon>0$.

\emph{Boundedness of $\mathcal{S}_\epsilon$:} For a given $(q_{X,Y},q_{Z_1,Z_2|X,Y})$, for every $p(u,z_1,z_2|x,y)$ that
satisfies \eqref{eqn_proofs_feasibility_compact}, it follows from the Convex Cover Method \cite[Appendix C]{GamalK12} along similar lines as cardinality bounds of part $(i)$ that there exists another 
p.m.f. $p(u,z_1,z_2|x,y)$ with $|\mathcal{U}_1|\leq |\mathcal{X}||\mathcal{Y}||\mathcal{Z}_1||\mathcal{Z}_2|+1$, 
$|\mathcal{U}_i|\leq |\mathcal{X}||\mathcal{Y}||\mathcal{Z}_1||\mathcal{Z}_2|\prod_{j=1}^{i-1}|\mathcal{U}_j|+1$, $\forall i>1$ 
satisfying \eqref{eqn_proofs_feasibility_compact}. This implies that all the probability vectors $p(u,z_1,z_2|x,y)$ 
in $\mathcal{S}_\epsilon$ are of finite dimension and hence the set $\mathcal{S}_\epsilon$ is bounded.

\emph{Closedness of $\mathcal{S}_\epsilon$:} Consider a sequence of p.m.f.'s $p^{(k)}_{U_{[1:r]},Z_1,Z_2|X,Y}$ converging 
to a p.m.f. $r_{U_{[1:r]},Z_1,Z_2|X,Y}$, where $p^{(k)}_{U_{[1:r]},Z_1,Z_2|X,Y}\in\mathcal{S}_\epsilon,\forall k\in\mathbbm{N}$. 
For simplicity, we abbreviate $p^{(k)}_{U_{[1:r]},Z_1,Z_2|X,Y}$ and  $r_{U_{[1:r]},Z_1,Z_2|X,Y}$ by $p^{(k)}$ and $r$ respectively. 
Since mutual information is a continuous function of the distribution, it follows from the definition of the continuous 
function that $I_{p^{(k)}}(U_{[1:r]};X,Z_1|Y,Z_2)$ converges to $I_r(U_{[1:r]};X,Z_1|Y,Z_2)$ as $p^{(k)}$ converges to $r$. 
Now since $I_{p^{(k)}}(U_{[1:r]};X,Z_1|Y,Z_2)\leq\epsilon$, $\forall k\in\mathbbm{N}$, it follows by taking limit $k\rightarrow \infty$ on both sides, 
that $I_r(U_{[1:r]};X,Z_1|Y,Z_2)\leq\epsilon$. Similarly we can prove that $I_r(U_{[1:r]};Y,Z_2|X,Z_1)\leq \epsilon$. 
Since $l_1$-norm is also a continuous function, $\lVert r_{X,Y,Z_1,Z_2}-q_{X,Y,Z_1,Z_2}\rVert_1\leq \epsilon$ also follows similarly. 
So, we have $r\in\mathcal{S}_\epsilon$ and hence $\mathcal{S}_\epsilon$ is closed.

By Cantor's intersection theorem, which states that a decreasing nested sequence of non-empty compact sets 
has non-empty intersection, we have that $\lim_{k\rightarrow\infty}\bigcap_{i=1}^k \mathcal{S}_{\epsilon_k}$ is non-empty. 
To show $\mathcal{S}_0=\lim_{k\rightarrow\infty}\bigcap_{i=1}^k \mathcal{S}_{\epsilon_k}$, 
note that $\mathcal{S}_0\subseteq\lim_{k\rightarrow\infty}\bigcap_{i=1}^k \mathcal{S}_{\epsilon_k}$ holds 
trivially because $\mathcal{S}_\epsilon$ shrinks as $\epsilon$ shrinks. For the other direction, assume 
that $\gamma_{U_{[1:r]},Z_1,Z_2|X,Y}\in\lim_{k\rightarrow\infty}\bigcap_{i=1}^k \mathcal{S}_{\epsilon_k}$. 
For simplicity, we abbreviate $\gamma_{U_{[1:r]},Z_1,Z_2|X,Y}$ by $\gamma$ in the following. 
Since $\gamma\in\mathcal{S}_{\epsilon_k}$, $\forall k$, we have $I(U_{[1:r]};X,Z_1|Y,Z_2)\leq \epsilon_k$, $\forall k \in \mathbb{N}$. 
Now since $\lim_{k\rightarrow \infty}\epsilon_k=0$ and mutual information is always non-negative, we have $I_\gamma(U_{[1:r]};X,Z_1|Y,Z_2)=0$.
Similarly, we can prove that $I_\gamma(U_{[1:r]};Y,Z_2|X,Z_1)=0$ and $\lVert p_{X,Y,Z_1,Z_2}-q_{X,Y,Z_1,Z_2}\rVert_1|_\gamma=0$. 
The latter conclusion implies that $\gamma_{X,Y,Z_1,Z_2}=q_{X,Y,Z_1,Z_2}$. This shows that $\gamma\in\mathcal{S}_0$, which concludes the proof of part $(i)$.  

{\color{black}\emph{Proof of part $(ii)$}: Suppose $(q_{XY},q_{Z_1Z_2|XY})$ is asymptotically securely computable with privacy against both users using an $r$-round protocol in which Alice starts the communication. Then by part~$(i)$, there exists a conditional p.m.f. $p(u_{[1:r]}|x,y,z_1,z_2)$ satisfying \eqref{Eq_AB_Markov1}-\eqref{Eq_AB_Markov_Secr2}. So, the joint distribution can be written as
\begin{align}
p(x,y,u_{[1:r]},z_1,z_2)&=q(x,y)p(u_{[1:r]}|x,y)p(z_1|u_{[1:r]},x)\nonumber\\
&\hspace{12pt}\times p(z_2|u_{[1:r]},y)\nonumber.
\end{align}
Consider another joint distribution,
\begin{align}
\tilde{p}(x,y,u_{[1:r]},z_1,z_2)&=\tilde{q}(x,y)p(u_{[1:r]}|x,y)p(z_1|u_{[1:r]},x)\nonumber\\
&\hspace{12pt}\times p(z_2|u_{[1:r]},y)\nonumber.
\end{align}
Notice that $p(u_{[1:r]},z_1,z_2|x,y)=\tilde{p}(u_{[1:r]},z_1,z_2|x,y)=p(u_{[1:r]}|x,y)p(z_1|u_{[1:r]},x)p(z_2|u_{[1:r]},y)$. We show that $\tilde{p}(x,y,u_{[1:r]},z_1,z_2)$ also satisfies \eqref{Eq_AB_Markov1}-\eqref{Eq_AB_Markov_Secr2}, thereby making the function $(\tilde{q}_{XY},q_{Z_1Z_2|XY})$ asymptotically securely computable in $r$-rounds.  It is trivial to see that \eqref{Eq_AB_Markov1}-\eqref{Eq_AB_Markov_Decod2} depends only on $p(u_{[1:r]}|x,y), p(z_1|u_{[1:r]},x)$ and $p(z_2|u_{[1:r]},y)$ which are common in both the distributions $p(x,y,u_{[1:r]},z_1,z_2)$ and $\tilde{p}(x,y,u_{[1:r]},z_1,z_2)$. Hence, since the p.m.f. $p(x,y,u_{[1:r]},z_1,z_2)$ satisfies \eqref{Eq_AB_Markov1}-\eqref{Eq_AB_Markov_Decod2}, $\tilde{p}(x,y,u_{[1:r]},z_1,z_2)$ also satisfies \eqref{Eq_AB_Markov1}-\eqref{Eq_AB_Markov_Decod2}. It remains to show that the p.m.f. $\tilde{p}(x,y,u_{[1:r]},z_1,z_2)$ satisfies \eqref{Eq_AB_Markov_Secr1}-\eqref{Eq_AB_Markov_Secr2}. For $x,y,z_1,z_2$ s.t. $\tilde{p}(x,y,z_1,z_2)>0$, consider
\begin{align}
\tilde{p}(u_{[1:r]}|x,y,z_1,z_2)&=\frac{\tilde{p}(u_{[1:r]},z_1,z_2|x,y)}{\tilde{p}(z_1,z_2|x,y)}\nonumber\\
&=\frac{{p}(u_{[1:r]},z_1,z_2|x,y)}{{p}(z_1,z_2|x,y)}\label{eqn:support1}\\
&={p}(u_{[1:r]}|x,y,z_1,z_2)\nonumber.
\end{align}
Notice that ${p}(u_{[1:r]}|x,y,z_1,z_2)$ is well-defined as $\text{supp}(\tilde{q}_{XY})\subseteq \text{supp}(q_{XY})$. \eqref{eqn:support1} follows as $p(u_{[1:r]},z_1,z_2|x,y)=\tilde{p}(u_{[1:r]},z_1,z_2|x,y)$. Now, for $x,y,z_1,z_2$ s.t. $\tilde{p}(x,y,z_1,z_2)>0$, since $\tilde{p}(u_{[1:r]}|x,y,z_1,z_2)={p}(u_{[1:r]}|x,y,z_1,z_2)$, we have that $\tilde{p}(x,y,u_{[1:r]},z_1,z_2)$ satisfies \eqref{Eq_AB_Markov_Secr1}-\eqref{Eq_AB_Markov_Secr2} also. This completes the proof of Theorem~\ref{Feasibility_AB}.
}
\end{proof}
\begin{proof}[\textbf{Proof of Theorem~\ref{cutset}}]
{\color{black}Part $(i)$ follows directly from Theorem~\ref{Thm_Rate_Region_AB}. For part $(ii)$,
we show that, for the class of functions mentioned in the statement of theorem, }if a scheme computes $(q_{XY},q_{Z_1Z_2|XY})$ 
 with  $R_{12}$ and $R_{21}$ equal to $I(Z_2;X|Y)+\delta$ and $I(Z_1;Y|X)+\delta$
respectively, and with some $R_0$, under no privacy,
then this scheme will also satisfy 
the privacy conditions \eqref{eqn:asymptotic_2}-\eqref{eqn:asymptotic_3}.
 From the converse of \cite[Theorem~1]{YassaeeGA15}, we have $nR_{12}  \geq I(M_{[1:r]};X^n|Y^n,W)$.
Then we get
\begin{align}
 n&R_{12}\nonumber\\
  & \geq I(M_{[1:r]};X^n|Y^n,W)\nonumber\\
 & =  I(M_{[1:r]},W;X^n|Y^n)\label{eqn:cutset_1}\\
 & =  I(M_{[1:r]},W;X^n|Y^n) \nonumber\\
 &\hspace{12pt}+I(Z_2^n;X^n|M_{[1:r]},W,Y^n)\label{eqn:cutset_2}\\
 & = I(Z_2^n, M_{[1:r]}, W;X^n|Y^n)\nonumber\\
 & = I(Z_2^n;X^n|Y^n) + I(M_{[1:r]}, W;X^n|Y^n, Z_2^n) \label{eqn:forcutset}\\
 & = I(Z_2^n;X^n|Y^n)+I(M_{[1:r]}, W;X^n, Z_1^n|Y^n, Z_2^n)\nonumber\\
 &\hspace{12pt}-I(M_{[1:r]},W;Z_1^n|X^n,Y^n,Z_2^n)\nonumber\\
 & \geq I(Z_2^n;X^n|Y^n)+I(M_{[1:r]}, W;X^n, Z_1^n|Y^n, Z_2^n)\nonumber\\
 &\hspace{12pt}-H(Z_1^n|X^n,Y^n,Z_2^n)\nonumber\\
 & \geq n\left[I(Z_2;X|Y)-\epsilon_1\right]+ I(M_{[1:r]}, W;X^n, Z_1^n|Y^n, Z_2^n)\nonumber\\
 &\hspace{12pt} -n\left[H(Z_1|X,Y,Z_2)+\epsilon_2\right]\label{eqn_cutset_3},
\end{align}
where \eqref{eqn:cutset_1} is due to the independence of common randomness $W$ and $(X^n,Y^n)$, \eqref{eqn:cutset_2} follows from the Markov chain $Z_2^n-(W,Y^n,M_{[1:r]})-(X^n,Z_1^n)$. We used the following fact in \eqref{eqn_cutset_3}: if two random variables $A$ and $A'$ with 
  same support set $\mathcal{A}$  satisfy $||p_{A} -  p_{A'}||_{1} \leq \epsilon \leq 1/4 $, then it follows from \cite[Theorem 17.3.3]{CoverJ06} that $|H(A) - H(A')|\leq \eta \log |\mathcal{A}|$, where $\eta \rightarrow 0$ as $\epsilon \rightarrow 0$. Now \eqref{eqn:asymptotic_1} implies \eqref{eqn_cutset_3}, where $\epsilon_1,\epsilon_2\rightarrow 0$ as $\epsilon\rightarrow 0$.

When $H(Z_1|X,Y,Z_2)=0$, from \eqref{eqn_cutset_3} we have $I(M_{[1:r]}, W;X^n, Z_1^n|Y^n, Z_2^n)\leq \delta+\epsilon_1+\epsilon_2$ for $\delta\rightarrow 0$, and $\epsilon_1,\epsilon_2\rightarrow 0$ as $\epsilon\rightarrow 0$, which is the required privacy condition against Bob. Similar argument holds for $R_{21}$ when $H(Z_2|X,Y,Z_1)=0$. 
\end{proof}
Proofs of Theorem~\ref{Thm_Rate_Region_A} and Theorem~\ref{Thm_Rate_Region_B} are along similar lines as that of Theorem~\ref{Feasibility_AB} and Theorem~\ref{Thm_Rate_Region_AB} jointly, by noticing the following and hence omitted.
When privacy against Alice is required,
\begin{align}
I(U_{[1:r]};Z_1,Z_2|X,Y)&=I(U_{[1:r]};Z_1|X,Y)\nonumber\\
&\hspace{12pt}+I(U_{[1:r]};Z_2|X,Y,Z_1)\\
&=I(U_{[1:r]};Z_1|X,Y)\label{Thm2_simplification_1},
\end{align}
where \eqref{Thm2_simplification_1} follows from \eqref{eqn:asymptotic_2}.
\begin{align}
I(U_1;Z_1,Z_2|X,Y)&=I(U_1;Z_1|X,Y)+I(U_1;Z_2|X,Y,Z_1)\\
&=I(U_1;Z_1|X,Y)\label{Thm2_simplification_2},
\end{align}
where \eqref{Thm2_simplification_2} follows from \eqref{eqn:asymptotic_2}. Similarly, when privacy is required against Bob, we have $I(U_{[1:r]};Z_1,Z_2|X,Y)=I(U_{[1:r]};Z_2|X,Y)$ and $I(U_1;Z_1,Z_2|X,Y)=I(U_1;Z_2|X,Y)$.

\section{Cut-set Bounds for Randomized Interactive Function Computation}\label{cutset_discussion}
For a randomized interactive function computation problem, we prove lower bounds for $R_{12}$ and $R_{21}$. To lower bound $R_{12}$, note that all the inequalities upto \eqref{eqn:forcutset} in the proof of Theorem~\ref{cutset} will follow even for function computation without any privacy requirement. Then we get
\begin{align}
 nR_{12}&\geq I(Z_2^n;X^n|Y^n) + I(M_{[1:r]}, W;X^n|Y^n, Z_2^n)\\
 &\geq  I(Z_2^n;X^n|Y^n)\nonumber\\
 &\geq n[I(Z_2;X|Y)-\epsilon']\label{eqn:forcutset1},
\end{align} 
where in \eqref{eqn:forcutset1} we have used the fact: if two random variables $A$ and $A'$ with 
  same support set $\mathcal{A}$  satisfy $||p_{A} -  p_{A'}||_{1} \leq \epsilon \leq 1/4 $, then it follows from \cite[Theorem 17.3.3]{CoverJ06} that $|H(A) - H(A')|\leq \eta \log |\mathcal{A}|$, where $\eta \rightarrow 0$ as $\epsilon \rightarrow 0$. Now \eqref{eqn:asymptotic_1} implies \eqref{eqn:forcutset1}, where $\epsilon' \rightarrow 0$ as $\epsilon \rightarrow 0$. Thus, we have $R_{12}\geq I(Z_2;X|Y)-\epsilon'$, where $\epsilon'\rightarrow 0$ as $\epsilon\rightarrow 0$. So, as $\epsilon \rightarrow 0$ we have $R_{12}\geq I(X;Z_2|Y)$. Similarly, $R_{21}\geq I(Y;Z_1|X)$. These are the  cut-set lower bounds.

\section{Details Omitted from Example~\ref{exam_1}}
\label{App_opt_rate}
 \iftoggle{paper}{Details can be found in the extended version.}{Suppose there exists a two round protocol with Alice starting the communication.
 If Bob is able to compute the function with a single message from Alice, then \cite[Theorem 1]{OrlitskyR01} gives that $R_{12}$ should
 be greater than or equal to the conditional graph entropy of the confusability graph \cite{OrlitskyR01}. For this example, it can be
 easily verified that $R_{12} \geq \log m + 1$ which is strictly  greater than the cut-set bound $\log m + 1/2$.
 This shows that secure computation is not possible with a two round protocol with Alice starting the communication (because a statement similar to Theorem~\ref{cutset} holds when privacy is required only against Bob with $r=2$ rounds).
 
 {\color{black}Now let us consider a two round protocol with Bob starting the communication. From Theorem~\ref{Thm_Rate_Region_B}, there exists random variables $U_1,U_2$ satisfying 
\begin{align}
 &U_1-Y-X, \label{Eq_Markov1} \\
 &U_2-(U_1,X)-Y, \label{Eq_Markov2} \\
 &(U_1,U_2)-(Y,Z)-X,\label{Eq_Sec_Markov} \\
 &H(Z|U_1,U_2,Y)=0, \label{Eq_Dec_Markov} \\
& H(Z|U_1,U_2,X)=0, \label{Eq_Dec_Markov2}
\end{align}
such that $R_{12}\geq I(Y;U_{[1:2]}|X)=I(Y;U_1|X)$ and $R_{21}\geq I(X;Z|Y)=H(Z|Y)$.}
We show that $U_1$ must be such that $H(Y|U_1)=0$. {\color{black}Suppose not, i.e., $H(Y|U_1) > 0$. Then there exists $u_1,y$ and $ y'$, where $y\neq y'$ such that $p(y,u_1)>0$ and $p(y',u_1)>0$. Let $i\in [1:m]$ be an index such that $y_i \neq y'_i$. 
%Since inputs are independent, assume that $J=i$ and $V=1$.
Since $U_1 -Y - X$, we have 
\begin{align}
P(J=i, V=1, Y=y, U_1=u_1) > 0 ,\label{newproof_3}\\
P(J=i, V=1, Y=y', U_1=u_1) > 0\label{newproof_4}. 
\end{align}
Since $U_2 - (U_1,X) - Y$, \eqref{newproof_3} and \eqref{newproof_4} imply that there exists $u_2$ such that 
\begin{align}
P(J=i, V=1, Y=y, U_1=u_1, U_2=u_2) > 0 ,\label{newproof_1}\\
P(J=i, V=1, Y=y', U_1=u_1, U_2=u_2) > 0. \label{newproof_2}
\end{align}
For $(J=i, V=1, Y=y)$, $Z=(i,y_i):=z$, whereas for $(J=i, V=1, Y=y')$, $Z=(i,y'_i):=z'$. Note that $z \neq z'$. Since $Z$ is a deterministic function of $(U_1,U_2,X)$ (this follows from \eqref{Eq_Dec_Markov2}), \eqref{newproof_1} and \eqref{newproof_2} above imply that 
\begin{align*}
P(J=i, V=1, Y=y, U_1=u_1, U_2=u_2, Z=\tilde{z}) > 0 ,\\
P(J=i, V=1, Y=y', U_1=u_1, U_2=u_2, Z=\tilde{z}) > 0, 
\end{align*}
 for some $\tilde{z}$. This is a contradiction.}
 {\color{black}Hence,} $H(Y|U_1)=0$ for any choice of $U_1$ and
\begin{align*}
 I(Y;U_1|X ) &= H(Y|X) - H(Y|X, U_1)\\
 &\geq H(Y|X) - H(Y|U_1)\\
 & = H(Y|X)\\
 & = H(Y)\\
 &=m.
\end{align*}
Therefore, to preserve decodability and privacy, the minimum possible $R_{12}$ is $m$ which also 
shows that the function is not computable in $1$ round with Alice starting the communication.
Rate $R_{21}$ is lower bounded by $H(Z|Y)$, the cut-set bound, which is equal to $\log m + 1/2$.
So the sum-rate is lower bounded by $\log m + 1/2+m$. This can be achieved by first Bob communicating
$U_1=Y$ to Alice, and Alice computing $Z$ and communicating $U_2=Z$ to Bob.
It is easy to see that this does not violate any privacy requirements.
Thus, the optimum sum-rate for two round protocols is {\color{black}$R^{B-\text{pvt}}_{\text{sum}}(2)=I(X;Z|Y)+I(Y;U_{[1:2]}|X)=\log m + 1/2+m$}.

Now let us consider a three round protocol. Alice starts the communication by sending {\color{black}$U_1=J$}. Then Bob communicates {\color{black}$U_2=Y_J$} to Alice and in the third round Alice communicates 
{\color{black}$U_3=V \wedge Y_J$} to Bob. The sum-rate of this protocol is {\color{black}$I(X;Z|Y)+I(Y;U_{[1:3]}|X)=\log m + 1/2 +1$} which
is strictly less than the optimum sum-rate for two round protocols. 
Now, we show that $\log m + 1/2 +1$ is the optimum sum-rate for any $r$ round protocol with $r \geq 3$.}
We use the following lemma.
\begin{lemma}[{\cite{AhlswedeC93, Maurer93}}]\label{lemma}
 If $X$ and $Y$ are independent, then the following holds for any
 $U_{[1:r]}$ satisfying \eqref{Eq_AB_Markov1} and \eqref{Eq_AB_Markov2}.
 \begin{align}
  I(X;Y|U_{[1:r]}) & = 0. \label{Eq_Intr_Markov}
 \end{align}
\end{lemma}
\begin{proof}
We prove a more general statement, i.e., for any $U_{[1:r]}$ satisfying \eqref{Eq_AB_Markov1} and \eqref{Eq_AB_Markov2}, we have $I(X;Y|U_{[1:i]})\leq I(X;Y)$ for any $i\in [1:r]$. We show this by induction on $i$. The base case, for $i=1$ is true since,
\begin{align}
I(X;Y|U_1)&\leq I(X;Y|U_1)+I(U_1;Y)\nonumber\\
&=I(X,U_1;Y)\nonumber\\
&=I(X;Y)+I(U_1;Y|X)\nonumber\\
&=I(X;Y)\label{eqn_indep_lemma_1},
\end{align}
where \eqref{eqn_indep_lemma_1} follows from \eqref{Eq_AB_Markov1}. Now suppose that $I(X;Y|U_{[1:i]})\leq I(X;Y)$. Assume that $i$ is odd. Then we have
\begin{align}
I(X;Y|U_{[1:i+1]})&\leq I(X;Y|U_{[1:i+1]})+I(X;U_{i+1}|U_{[1:i]})\nonumber\\
&= I(X;U_{i+1},Y|U_{[1:i]})\nonumber\\
&=I(X;Y|U_{[1:i]})+I(X;U_{i+1}|Y,U_{[1:i]})\nonumber\\
&=I(X;Y|U_{[1:i]})\label{eqn_indep_lemma_2}\\
&\leq I(X;Y)\nonumber,
\end{align} 
where \eqref{eqn_indep_lemma_2} follows from \eqref{Eq_AB_Markov2}. One can prove that $I(X;Y|U_{[1:i+1]})\leq I(X;Y)$ for the case when $i$ is even using \eqref{Eq_AB_Markov1}. This completes the proof of lemma.
\end{proof}

The proof of $\log m + 1 + 1/2 $ being the optimal sum-rate for any arbitrary rounds follows along similar lines as that
of showing that $\log m + m + 1/2 $ to be the optimal sum-rate for two round protocols. It is outlined below. 
We show that $H(Y_J|U_{[1:r]},J)=0$. Let us assume $H(Y_J|U_{[1:r]},J) \neq 0$. 
This implies that $\exists \; u_{[1:r]}$ and $j$ 
such that 
\begin{align*}
 P(Y_j=0, U_{[1:r]} = u_{[1:r]}, J=j) & > 0\\
 P(Y_j=1,U_{[1:r]} = u_{[1:r]}, J=j) & > 0.
\end{align*}
These two further imply that $\exists \; y,y'$ such that $y(j)=0$, $y'(j)=1$, and
\begin{align}
 P(Y=y, U_{[1:r]} = u_{[1:r]}, J=j) & > 0 \label{Eq_Intr_Prob1}\\
 P(Y=y',U_{[1:r]} = u_{[1:r]}, J=j) & > 0. \label{Eq_Intr_Prob2}
\end{align}
Now we analyse two cases as before.
\begin{align}
\mbox{Case 1: } H(V|U_{[1:r]} = u_{[1:r]}, J=j) & \neq 0 \label{Eq_Intr_Case1} \\
\mbox{Case 2: } H(V|U_{[1:r]} = u_{[1:r]}, J=j) & = 0 \label{Eq__Intr_Case2}.
\end{align}

\noindent \underline{\emph{Case 1}} $(H(V|U_{[1:r]} = u_{[1:r]}, J=j) \neq 0):$\\
This implies the following:
% for some $\bar{X}=x$ and $\bar{X}=x'$ with $x(i)=0$, $ x'(i) =1 $, and 
\begin{align}
 P(V=0, U_{[1:r]} = u_{[1:r]}, J=j)&>0 \label{Eq_intr_case1_Prob1} \\
 P(V=1, U_{[1:r]} = u_{[1:r]}, J=j)&>0. \label{Eq_int_case1_Prob2} 
\end{align}
\eqref{Eq_intr_case1_Prob1}, \eqref{Eq_Intr_Prob2} and Lemma \ref{lemma} imply that 
\begin{align}
 P(V=0, J=j, U_{[1:r]} = u_{[1:r]}, Y =y' )& >0 \label{Eq_case1_Prob1} .
 \end{align}
 Similarly, \eqref{Eq_int_case1_Prob2}, \eqref{Eq_Intr_Prob2} and Lemma \ref{lemma} imply that 
\begin{align}
 P(V=1, J=j, U_{[1:r]} = u_{[1:r]}, Y =y' )& >0 \label{Eq_case1_Prob2} .
 \end{align}
From \eqref{Eq_case1_Prob1} and \eqref{Eq_case1_Prob2}, we get $H(V \wedge Y_J| U_{[1:r]}, Y, J)\neq 0$.
Then we have $H(V \wedge Y_J| U_{[1:r]}, Y)\neq 0$ which is a contradiction.

\noindent \underline{\emph{Case 2}} $(H(V|U_{[1:r]} = u_{[1:r]}, J=j) = 0):$\\
This condition implies either
\begin{align}
 P(V=0, U_{[1:r]} = u_{[1:r]}, J=j)& >0 \label{Eq_case2_Prob1} \\
 P(V=1,  U_{[1:r]} = u_{[1:r]}, J=j)& =0, \label{Eq_case2_Prob2} 
\end{align}
or 
\begin{align*}
 P(V=1,  U_{[1:r]} = u_{[1:r]}, J=j)& >0  \\
 P(V=0, U_{[1:r]} = u_{[1:r]}, J=j)& =0. 
\end{align*}

Let us consider the first case. Then
\eqref{Eq_case2_Prob1}, \eqref{Eq_case2_Prob2}, \eqref{Eq_Intr_Prob1}
and Lemma \ref{lemma} imply that
 $P( V=0,J=j, U_{[1:r]} = u_{[1:r]},  Y = y, V\wedge Y_j=0)>0$
and $P(V=1, J=j,  U_{[1:r]} = u_{[1:r]}, Y=y, V\wedge Y_j=0)= 0$. This violates the Markov chain 
$U_{[1:r]} - (Y,Z)- X$. 
Similarly, the second case also violates the Markov chain $U_{[1:r]} - (Y,Z)- X$. 
This shows that $H(V|U_{[1:r]} = u_{[1:r]}, J=j)$ can not be zero.

From Theorem~\ref{Thm_Rate_Region_B}, we have $R_{21} \geq  I(Y; U_{[1:r]}|X)$.
Then we get the following
\begin{align*}
 I(Y; U_{[1:r]}|X ) & \geq I(Y_J; U_{[1:r]}|V, J )\\
 &= H(Y_J|V, J) - H(Y_J|V,J,  U_{[1:r]})\\
 &\geq H(Y_J|V, J) - H(Y_J| J, U_{[1:r]})\\
 & = H(Y_J|V, J)=1.
\end{align*}

So we have $R_{21} \geq 1$. From the cut-set bound, we have $R_{12} \geq H(Z|Y)= \log m + 1/2$. This implies that sum-rate is lower bounded by 
$\log m + 1/2+1 $ by any $r$ round protocol.
\RRB{\section{Details Omitted from Example~\ref{example:asympt}}\label{appendixexample}
We have
\begin{align}
R^{B-\text{pvt}}_{\text{sum}}(1,0)=1+\min_{\substack{p(u|x,z):\\  Z-U-X,\\ Z-(U,X)-Z,\\ U-Z-X}}I(U;Z|X)\label{newexample1}\\
\geq 1+\min_{\substack{p(u|x,z):\\  Z-U-X,\\ U-Z-X}}I(U;Z|X)\label{newexample2}.
\end{align}
Let $U$ satisfy the Markov chains $Z-U-X$ and $U-Z-X$ and define $\mathcal{U}_z=\{u\in\mathcal{U}: p_{UZ}(u,z)>0\}$. Then we have the following claim. 
\begin{claim}
$\mathcal{U}_z\cap \mathcal{U}_{z^\prime}=\emptyset.$
\end{claim}
\begin{proof}Suppose $\mathcal{U}_z\cap \mathcal{U}_{z^\prime}\neq\emptyset$. Let $u\in \mathcal{U}_z\cap \mathcal{U}_{z^\prime}$. Then there exist $z,z^\prime$ such that $p_{UZ}(u,z), p_{UZ}(u,z^\prime)>0$.  Let $\alpha_z(x)=q_{Z|X}(z|x)$, for every $x\in\mathcal{X}$ and $\vec{\alpha}_z$ denote a vector with entries indexed by $x\in\mathcal{X}$. For this example, it is easy to see that $\text{supp}(\vec{\alpha}_z)\neq \text{supp}(\vec{\alpha}_{z^\prime})$, for $z\neq z^\prime$. Assume, without loss of generality that,  $\text{supp}(\vec{\alpha}_z)\setminus \text{supp}(\vec{\alpha}_{z^\prime})\neq \emptyset$. Let $x\in \text{supp}(\vec{\alpha}_z)\setminus \text{supp}(\vec{\alpha}_{z^\prime})$. This implies that $\vec{\alpha}_z(x)>0$ and $\vec{\alpha}_{z^\prime}(x)=0$. Since $\vec{\alpha}_{z^\prime}$ is a non-zero vector, there exists $x^\prime$ such that $\vec{\alpha}_{z^\prime}(x^\prime)>0$. Note that $\vec{\alpha}_z(x)>0$ implies $q_{Z|X}(z|x)>0$, $\vec{\alpha}_{z^\prime}(x)=0$ implies $q_{Z|X}(z^\prime|x)=0$, and $\vec{\alpha}_{z^\prime}(x^\prime)>0$ implies $q_{Z|X}(z^\prime|x^\prime)>0$. These imply, by privacy against Bob, that $p_{UZ|X}(u,z|x), p_{UZ|X}(u,z^\prime|x^\prime)>0$. Now consider $p_{UZ|X}(u,z^\prime|x^\prime)$ and expand it as follows: 
\begin{align}
p_{UZ|X}(u,z^\prime|x^\prime)&=p_{U|X}(u|x^\prime)p_{Z|UX}(z^\prime|u,x^\prime)\\
&=p_{U|X}(u|x^\prime)p_{Z|U}(z^\prime|u)
\end{align}
$p_{U|X}(u|x^\prime)>0$ and $p_{Z|U}(z^\prime|u)>0$, respectively because $p_{UZ|X}(u,z|x)>0$ and $p_{UZ|X}(u,z^\prime|x^\prime)>0$. This implies $p_{UZ|X}(u,z^\prime|x^\prime)>0$, which in turn implies $q_{Z|X}(z^\prime|x^\prime)>0$, which is a contradiction. 
\end{proof}
This shows that $Z$ is a deterministic function of $U$. Now we have $I(U;Z|X)=H(Z|X)$. This gives
\begin{align}
1+\min_{\substack{p(u|x,z):\\  Z-U-X,\\ U-Z-X}}I(U;Z|X)=1+\log(m).
\end{align}
Notice that an optimum $U$, i.e., $U=Z$ satisfies the Markov chain $Z-(U,X)-Z$ also and hence the inequality in \eqref{newexample1} is actually an equality, thereby giving us $R^{B-\text{pvt}}_{\text{sum}}(1,0)=1+\log(m)=\log(2m)$.
}

\section{Proof of Proposition~\ref{prop_asymptotic_alice}}\label{appendixH}
In view of Theorem~\ref{Thm_Rate_Region_A}, it suffices to show that
\[\displaystyle \min_{\substack{p_{U|X}: \\ U-X-(Y,Z) \\ Z-(U,Y)-X}} I(U;X|Y)\quad = 
\displaystyle \min_{\substack{p_{W|X}: \\ W-X-Y \\ X\in W\in \varGamma(\mathcal{G})}} I(W;X|Y).\]
A similar result for deterministic functions was proved by Orlitsky and Roche \cite[Theorem 2]{OrlitskyR01}.
Since we are dealing with randomized functions, the proof of this proposition is slightly more involved. 

\RRB{L.H.S. $\leq$ R.H.S.:} Suppose $p_{W|X}$ achieves the minimum on the right hand side.
Then $W$ is a random variable such that it takes values in the set $\varGamma(\mathcal{G})$ of independent sets of $\mathcal{G}$ such that $X\in W$.
\iffalse
Consider any independent set, say $w\in \mathcal{G}$. Consider a pair $(y,z)\in\mathcal{Y}\times\mathcal{Z}$. 
By definition of $\mathcal{G}$, $p_{Z|XY}(z|x,y)$ is the same for all $x\in w$ with $p_{XY}(x,y)>0$. 
Hence, $p_{Z|XY}(z|x,y)$ can be uniquely determined from $(w,y,z)$ where $x\in w$; therefore, the Markov chain $Z-(W,Y)-X$ holds. 
Observe that $p_{Z|XY}(z|x,y)$ can be uniquely determined from $(w,y,z)$ for all $w\in\varGamma(\mathcal{G})$ such that $x\in w$, 
which means that $Z-(X,Y)-W$ is a Markov chain, 
This, together with $W-X-Y$, implies $W-X-(Y,Z)$.
\fi
Now we define a joint distribution $p_{XYWZ}(x,y,w,z)$ as follows:
\[
p_{XYWZ}(x,y,w,z) := 
\begin{cases}
0  \text{ if } q_{XY}(x,y)=0, \\
0  \text{ if } q_{XY}(x,y)>0 \text{ and } x\notin w, \\
q_{XY}(x,y)p_{W|X}(w|x)q_{Z|XY}(z|x,y) \\
\hspace{12pt} \text{ if } q_{XY}(x,y)>0 \text{ and } x\in w.
\end{cases}
\]
Note that the above-defined $p_{XYWZ}$ satisfies the Markov chains $W-X-Y$ and $Z-(X,Y)-W$. These two Markov chain together implies that
$W-X-(Y,Z)$ is also a Markov chain.
Now we show that $p_{XYWZ}$ satisfies $Z-(W,Y)-X$ too.
Consider any independent set, say $w\in \G$. Consider any pair $(y,z)\in\mathcal{Y}\times\mathcal{Z}$. 
By definition of $\mathcal{G}$, $p_{Z|XY}(z|x,y)$ is the same for all $x\in w$ with $q_{XY}(x,y)>0$. 
Hence, $q_{Z|XY}(z|x,y)$ can be uniquely determined from $(w,y,z)$ where $x\in w$; therefore, the Markov chain $Z-(W,Y)-X$ holds. 
The inequality now follows by taking $U=W$.

\RRB{L.H.S. $\geq$ R.H.S.:} 
Suppose $p_{U|X}$ achieves the minimum on the left hand side such that $U-X-(Y,Z)$ and $Z-(U,Y)-X$ hold. 
Now define a random variable $W$ as a function of $U$ in the following way: $w=w(u):=\{x:p_{UX}(u,x)>0\}$. 
We need to show that the induced conditional distribution $p_{W|X}$ satisfies the following three conditions: 
1) the Markov chain $W-X-Y$ holds; 
2) $X\in W$, i.e., $p_{XW}(x,w)>0 \implies x\in w$; 
and 3) $W$ is an independent set in $\mathcal{G}=(\mathcal{X},\mathcal{E})$. We show these conditions below.
\begin{enumerate}
\item The Markov chain $W-X-Y$ holds because $U-X-(Y,Z)$ holds and $W$ is a function of $U$.
\item If $p_{WX}(w,x)>0$ then $\exists u$ s.t. $w=w(u)$ and $p_{UX}(u,x)>0$, which implies that $x\in w$.
\item To prove that $w$ is an independent set we suppose, to the contrary, that $w$ is not an independent set, 
which means that $\exists x,x',u$, where $w=w(u), x,x'\in w$, $\{x,x'\}\in \mathcal{E}$, and such that $p_{U|X}(u|x)\cdot p_{U|X}(u|x')>0$. 
By definition of $\G$, $\{x,x'\}$ being an edge in $\mathcal{E}$ implies that $\exists (y,z)\in\mathcal{Y}\times\mathcal{Z}$ such that $q_{Z|XY}(z|x,y)\neq q_{Z|X,Y}(z|x',y)$ and $q_{XY}(x,y)\cdot q_{XY}(x',y)>0$. 
Assume, without loss of generality, that $q_{Z|XY}(z|x,y)>0$.
Consider $(u,x,y,z)$. We can expand $p_{UZ|XY}(u,z|x,y)$ in two different ways. The first expansion is as follows:
\begin{align}
p_{UZ|XY}(u,z|x,y) &= p_{U|XY}(u|x,y)p_{Z|UXY}(z|u,x,y) \nonumber \\
&= p_{U|X}(u|x)p_{Z|UY}(z|u,y), \label{eq:cutset_equality1}
\end{align}
where in the second equality we used $U-X-Y$ to write $p_{U|XY}(u|x,y)=p_{U|X}(u|x)$ and $Z-(U,Y)-X$ to write 
$p_{Z|UXY}(z|u,x,y)=p_{Z|UY}(z|u,y)$ (note that $p_{Z|UXY}(z|u,x,y)$ is well-defined, because $p_{UXY}(u,x,y)$ = $q_{XY}(x,y)p_{U|X}(u|x)>0$). 
We can expand $p_{UZ|XY}(u,z|x,y)$ in the following way also:
\begin{align}
p_{UZ|XY}(u,z|x,y) &= q_{Z|XY}(z|x,y)p_{U|XYZ}(u|x,y,z) \nonumber \\
&= q_{Z|XY}(z|x,y)p_{U|X}(u|x), \label{eq:cutset_equality2}
\end{align}
where in the second equality we used the Markov chain $U-X-(Y,Z)$ to write $p_{U|XYZ}(u|x,y,z)=p_{U|X}(u|x)$ 
(note that $p_{U|XYZ}(u|x,y,z)$ is well-defined, because $p_{XYZ}(x,y,z)>0$). 
Since $p_{U|X}(u|x)>0$, comparing \eqref{eq:cutset_equality1} and \eqref{eq:cutset_equality2} gives 
\begin{align}
q_{Z|XY}(z|x,y) = p_{Z|UY}(z|u,y).\label{eq:cutset_equality25}
\end{align}
The above equality, together with $q_{Z|XY}(z|x,y)>0$, implies $p_{Z|UY}(z|u,y)>0$.
Now consider $(u,x',y,z)$, and expand $p_{UZ|XY}(u,z|x',y)$ along the first expansion above in \eqref{eq:cutset_equality1}. This gives:
\begin{align}
p_{UZ|XY}(u,z|x',y) = p_{U|X}(u|x')p_{Z|UY}(z|u,y). \label{eq:cutset_equality3}
\end{align}
Since all the terms on the RHS are greater than zero, we have $p_{UZ|XY}(u,z|x',y)>0$. Now expanding $p_{UZ|XY}(u,z|x',y)$ along the second expansion in \eqref{eq:cutset_equality2} above gives
\begin{align}
p_{UZ|XY}(u,z|x',y) = q_{Z|XY}(z|x',y)p_{U|X}(u|x'). \label{eq:cutset_equality4}
\end{align}
On comparing \eqref{eq:cutset_equality3} and \eqref{eq:cutset_equality4}, and using $p_{U|X}(u|x')>0$, we get 
$q_{Z|XY}(z|x',y) = p_{Z|UY}(z|u,y)$, which together with
\eqref{eq:cutset_equality25}, leads to a contradiction to our assumption that $q_{Z|XY}(z|x,y) \neq q_{Z|XY}(z|x',y)$.
\end{enumerate}

{\color{black} 
\section{Proof of Lemma~\ref{lem:equiv-problem-alice-bob}}\label{appendixD}
First we explicitly characterize $(q_{XY},q_{Z|XY})$ that are perfectly securely computable with privacy against both users. 
Kilian~\cite{Kilian00} gave a characterization of such $q_{Z|XY}$ with no input distribution. 
Essentially the same characterization holds for general $(q_{XY},q_{Z|XY})$ as well, and we prove it in our language in 
Lemma \ref{lem:characterization}; it will be useful in understanding the later results.}

{\color{black}We say that $C\times D$, where $C\subseteq \mathcal{X}, D\subseteq \mathcal{Y}$, is {\em column monochromatic}, if 
for every $y\in D$, if there exists $x,x'\in C$ such that $q_{XY}(x,y),q_{XY}(x',y)>0$, then $q_{Z|XY}(z|x,y)=q_{Z|XY}(z|x',y)$, $\forall z\in\mathcal{Z}$.

\begin{lemma}{\cite[Lemma 5.1]{Kilian00}}\label{lem:characterization}
Let $\mathcal{X}=\mathcal{X}_1\biguplus\mathcal{X}_2\biguplus\hdots\biguplus\mathcal{X}_k$ be the partition induced by the equivalence relation $\equiv$.
Then, $(q_{XY},q_{Z|XY})$ is perfectly securely computable with privacy against both users if and only if 
each $\mathcal{X}_i\times \mathcal{Y}$, $i\in\{1,2,\hdots,k\}$, is column monochromatic.
\end{lemma}
\begin{proof}
\RRB{\underline{`Only if' part}:}
Let $U$ denote the message sent by Alice to Bob. 
Let $(p(u|x),p(z|u,y))$ be a pair of encoder and decoder that perfectly securely computes $(q_{XY},q_{Z|XY})$ 
with privacy against both users,
i.e., the joint distribution $p(x,y,u,z)=p(x,y)p(u|x)p(z|u,y)$ satisfies the following correctness and privacy conditions:
\begin{align}
&\text{Correctness: } p(x,y,z) = q_{XY}(x,y)q_{Z|XY}(z|x,y),\quad \forall x,y,z, \label{eq:char-correctness} \\
&\text{Privacy against Alice: } U-X-(Y,Z), \label{eq:char-privacy-alice} \\
&\text{Privacy against Bob: } U-(Y,Z)-X. \label{eq:char-privacy-bob}
\end{align}
Consider an equivalence class $\mathcal{X}_i$. 
First we prove that $p(u|x)=p(u|x')$ for every message $u$ and every $x,x'\in\mathcal{X}_i$.
Since $\mathcal{X}_i$ is an equivalence class of $\equiv$, we have $x\equiv x'$, 
and by the definition of $x\equiv x'$, there exists $x=x_1,x_2,\hdots,x_{l-1},x_l=x'$ 
such that $x_i\sim x_{i+1}$ for every $0\leq i \leq l-1$.
Consider $x_i,x_{i+1}$ for some $i$ in this sequence. Since $x_i\sim x_{i+1}$, there exists $y,z$ such that 
$q(x_i,y),q(x_{i+1},y),q(z|x_i,y),q(z|x_{i+1},y)>0$. 
(These will ensure that all the conditional probabilities in the following equations are well-defined.)
Fix a message $u$ and consider the following:
\begin{align*}
p(u|x_i) &= p(u|x_i,y,z) \\
&= p(u|x_{i+1},y,z) \\
&= p(u|x_{i+1}),
\end{align*}
where the first and third equalities follow from \eqref{eq:char-privacy-alice}, and the 
second equality follows from \eqref{eq:char-privacy-bob}.
Since the above argument holds for every $i\in\{1,2,\hdots,k-1\}$, we have $p(u|x)=p(u|x')$ for every $x,x'\in\mathcal{X}_i$.
Hence, $U$ is conditionally independent of $X$, conditioned on the equivalence class to which $X$ belongs.

Now take any $y\in\mathcal{Y}$. We prove that if there exists $x,x'\in\mathcal{X}_i$ such that 
$q_{XY}(x,y)>0,q_{XY}(x',y)>0$, then $q_{Z|XY}(z|x,y)=q_{Z|XY}(z|x',y)$ for every $z\in\mathcal{Z}$.
Take any $u$ such that $p(u|x)>0$, then $p(u|x')=p(u|x)>0$. Take an arbitrary $z\in\mathcal{Z}$. Consider $(x,y,u,z)$ and expand $p(u,z|x,y)$ as follows:
\begin{align}
p(u,z|x,y) &= p(u|x)p(z|u,y). \label{eq:ps_2_one-way}
\end{align}
%In \eqref{eq:ps_2_one-way} we used \eqref{eq:char-message} to write $p(u|x,y)=p(u|x)$, and 
%\eqref{eq:char-output} to write $p(z|x,y,u)=p(z|y,u)$. 
We can expand $p(u,z|x,y)$ in another way:
\begin{align}
p(u,z|x,y) &= q(z|x,y)p(u|x,y,z) \nonumber \\
&= q(z|x,y)p(u|x). \label{eq:ps_2_another-way}
\end{align}
In \eqref{eq:ps_2_another-way} we used \eqref{eq:char-privacy-alice} to write $p(u|x,y,z)=p(u|x)$. 
Comparing \eqref{eq:ps_2_one-way} and \eqref{eq:ps_2_another-way} we get
\begin{align}
q(z|x,y)=p(z|u,y). \label{eq:ps_2_first-compare}
\end{align}
Running the same arguments with $(x',y,u,z)$ we get
\begin{align}
q(z|x',y)=p(z|u,y). \label{eq:ps_2_second-compare}
\end{align}
Comparing \eqref{eq:ps_2_first-compare} and \eqref{eq:ps_2_second-compare} gives 
$q(z|x,y)=q(z|x',y)$. Since the protocol is correct, i.e., $p(z|x,y)=q_{Z|XY}(z|x,y)$, we have our desired result that $p_{Z|XY}(z|x,y)=p_{Z|XY}(z|x',y)$. \\

\RRB{\underline{`If' part}:} If $\mathcal{X}_i\times\mathcal{Y}$, $i\in\{1,2,\hdots,k\}$, is column monochromatic, 
then we give a simple encoder-decoder pair to perfectly securely compute $(q_{XY},q_{Z|XY})$. 
%It is easy to check that the protocol is perfectly secure with privacy against both the parties.
Take $\mathcal{U}=\{1,2,\hdots,k\}$. Now define the encoder-decoder as follows:
\begin{itemize}
\item Encoder: $p_{U|X}(i|x) := \mathbbm{1}_{\{x\in\mathcal{X}_i\}}$, i.e., on input $x$, Alice sends the unique index $i\in[k]$ such that $x\in\mathcal{X}_i$.
\item Decoder: $p_{Z|UY}(z|i,y) := q_{Z|XY}(z|x,y)$, for any $x\in\mathcal{X}_i$ such that $q_{XY}(x,y)>0$. 
(If no such $x$ is there, leave $p_{Z|UY}(z|i,y)$ undefined.)
\end{itemize}
Upon observing $x$ as input, Alice sends the index $i$ such that $x\in\mathcal{X}_i$.
On observing $y\in\mathcal{Y}$ and receiving $i\in[1:k]$ from Alice, 
Bob outputs $z$ with probability $p_{Z|UY}(z|i,y)$. Note that $p_{Z|UY}(z|i,y)$ is well-defined, because $q_{XY}(x,y)>0$.
It can be easily verified that the above pair of encoder-decoder computes $(q_{XY},q_{Z|XY})$ with perfect security against both users.
\end{proof}

The above encoder-decoder may not be optimal in terms of communication complexity. 
Consider any two equivalence classes $\mathcal{X}_i$ and $\mathcal{X}_j$. Let $\mathcal{Y}_i := \{y\in\mathcal{Y} : \exists x\in\mathcal{X}_i \text{ s.t. } q_{XY}(x,y)>0\}$; $\mathcal{Y}_j$ is defined similarly. Suppose $\mathcal{Y}_i\cap \mathcal{Y}_j=\phi$.
Then, Alice can send the same message whether $x\in\mathcal{X}_i$ or $x\in\mathcal{X}_j$, 
and Bob will still be able to correctly compute $q_{Z|XY}$.

Consider an encoder-decoder pair $(p_{U|X},p_{Z|UY})$ that securely computes $(q_{XY},q_{Z|XY})$ with privacy against both users.
For $i=1,2,\hdots,k$, define $\mathcal{U}_i$ as follows:
\[\mathcal{U}_i:=\{u\in\mathcal{U}:p(u|x)>0, \text{ for some }x\in\mathcal{X}_i\}.\]
\begin{claim}\label{claim:msgs_different-classes}
For some $i,j\in[1:k],i\neq j$, if there exists $x\in\mathcal{X}_i$ and $x'\in\mathcal{X}_j$ such that $q_{XY}(x,y)>0,q_{XY}(x',y)>0$ for some $y\in\mathcal{Y}$, then $\mathcal{U}_i\cap\mathcal{U}_j=\phi$.
\end{claim}
\begin{proof}
We prove the claim by contradiction.
Suppose there exists $i,j\in[1:k],i\neq j$, and $x\in\mathcal{X}_i$, $x'\in\mathcal{X}_j$ such 
that $q_{XY}(x,y)>0,q_{XY}(x',y)>0$ for some $y\in\mathcal{Y}$, and $\mathcal{U}_i\cap\mathcal{U}_j\neq\phi$.
 Since $q_{XY}(x,y)>0$, there is a $z\in\mathcal{Z}$ with $q_{Z|XY}(z|x,y)>0$. Since $x$ and $x'$ are not in the same equivalence class, 
we have $q_{Z|XY}(z|x',y)=0$.

Suppose $u\in\mathcal{U}_i\cap\mathcal{U}_j$. It follows from the definition of $\mathcal{U}_i$ and the proof of Lemma \ref{lem:characterization} that $p(u|x)>0$ and $p(u|x')>0$. Now consider $(x,y,u,z)$ and by 
expanding $p(u,z|x,y)$ in two different ways (similar to what we did in the proof of 
Lemma \ref{lem:characterization}) we get $q_{Z|XY}(z|x,y)=p(z|u,y)$. 
Applying the same arguments with $(x',y,u,z)$ gives $q_{Z|XY}(z|x',y)=p(z|u,y)$. 
These imply $q_{Z|XY}(z|x,y)=q_{Z|XY}(z|x',y)$, which is a contradiction, because by assumption $q_{Z|XY}(z|x,y)>0$ and $q_{Z|XY}(z|x',y)=0$.
Hence, $\mathcal{U}_i\cap\mathcal{U}_j$ must be empty.
\end{proof}

From the proof of Lemma \ref{lem:characterization} we have that the $U$ is conditionally independent 
of the input $X$ conditioned on the equivalence class to which $X$ belongs. This, together with 
Claim \ref{claim:msgs_different-classes} suggests replacing all the elements in $\mathcal{X}_i, i\in[1:k]$, by a 
single element $x_i$, which induces a new random variable, which we denote by $X_{\text{EQ}}$. 
Note that $X_{\text{EQ}}$ takes values in the set $\mathcal{X}_{\text{EQ}}:=\{x_1,x_2,\hdots,x_k\}$, 
where $x_i$ is the representative of the equivalence class $\mathcal{X}_i$.
Thus we can define an equivalent problem ($q_{X_{\text{EQ}}Y},q_{Z|X_{\text{EQ}}Y}$) as follows:
\begin{itemize}\label{equiv-graph}
\item Define $q_{X_{\text{EQ}}Y}(x_i,y):=\sum_{x\in\mathcal{X}_i}q_{XY}(x,y)$ for every $(x_i,y)\in(\mathcal{X}_{\text{EQ}}\times\mathcal{Y})$. 
\item For $(x_i,y)\in(\mathcal{X}_{\text{EQ}}\times\mathcal{Y})$, if there exists $x\in\mathcal{X}_i$ s.t. $q_{XY}(x,y)>0$, then define $q_{Z|X_{\text{EQ}}Y}(z|x_i,y) :=q_{Z|XY}(z|x,y)$, for every $z\in\mathcal{Z}$. If there exists no $x\in\mathcal{X}_i$ s.t. $q_{XY}(x,y)>0$, then it does not matter what the conditional distribution $q_{Z|X_{\text{EQ}}Y}(z|x_i,y)$ is;
in particular, we can define $q_{Z|X_{\text{EQ}}Y}(.|x_i,y)$ to be the uniform distribution in $\{1,2,\hdots,|\mathcal{Z}|\}$.
\end{itemize}
}
{\color{black}Note that $U-X_{\text{EQ}}-X$ is a Markov chain.
The above definition is for $(q_{XY},q_{Z|XY})$ which are perfectly securely computable with privacy against both users as per Lemma \ref{lem:characterization}.
\begin{remark}\label{remark:bob-always-learn}
{\em Suppose $(q_{XY},q_{Z|XY})$ is securely computable. 
%Then for any code for securely computing $(p_{XY},p_{Z|XY})$ 
Observe that Bob {\em always} learns which equivalence class $\mathcal{X}_i$ Alice's input lies in. 
This is because $\forall (y,z)\in\mathcal{Y}\times\mathcal{Z}$ such that $q_{YZ}(y,z)>0$, there is a 
unique $x\in\mathcal{X}_{\text{EQ}}$ s.t. $q_{X_{\text{EQ}}Y}(x,y)>0$ and $q_{Z|X_{\text{EQ}}Y}(z|x,y)>0$.}
\end{remark}
Now we are ready to prove Lemma~\ref{lem:equiv-problem-alice-bob}.
\begin{proof}[Proof of Lemma~\ref{lem:equiv-problem-alice-bob}]
Let $(p_{U|X},p_{Z|UY})$ be an encoder-decoder pair that securely computes $(q_{XY},q_{Z|XY})$. Consider the following 
encoder-decoder pair $(p_{\tilde{U}|X_{\text{EQ}}},p_{Z|\tilde{U}Y})$ for securely computing the reduced function, that is $(q_{X_{\text{EQ}}Y},q_{Z|X_{\text{EQ}}Y})$.
The alphabet of $X_{\text{EQ}}$ is $\mathcal{X}_{\text{EQ}}=\{x_1,x_2,\hdots,x_k\}$, where $x_i$ is the representative element of the equivalence class $\mathcal{X}_i$
in the partition $\mathcal{X}=\mathcal{X}_1\biguplus\mathcal{X}_2\biguplus\hdots\biguplus\mathcal{X}_k$ generated by the equivalence relation $\equiv$;
see Definition \ref{defn:equiv-relation-1}.
Encoder is defined as follows: for every $i\in[1:k]$, define $p_{\tilde{U}|X_{\text{EQ}}}(u|x_i):=p_{U|X}(u|x)$ for any $x\in\mathcal{X}_i$. This is well-defined,
because $p_{U|X}(u|x)$ is identical for every $x$ in an equivalence class (see the proof of Lemma \ref{lem:characterization}).
Decoder remains the same, i.e., $p_{Z|\tilde{U}Y}=p_{Z|UY}$.
Note that the encoder $p_{\tilde{U}|X_{\text{EQ}}}$ is defined on a subset of $\mathcal{X}$, and on that subset it is identical to $p_{U|X}$.
This implies that, since $(p_{U|X},p_{Z|UY})$ securely computes $(q_{XY},q_{Z|XY})$, the above-defined $(p_{\tilde{U}|X_{\text{EQ}}},p_{Z|\tilde{U}Y})$
will also securely compute $(q_{X_{\text{EQ}}Y},q_{Z|X_{\text{EQ}}Y})$.

For the other direction, let $(p_{\tilde{U}|X_{\text{EQ}}},p_{Z|\tilde{U}Y})$ securely computes $(q_{X_{\text{EQ}}Y},q_{Z|X_{\text{EQ}}Y})$. Now consider the 
following encoder-decoder pair $(p_{U|X},p_{Z|UY})$ for securely computing $(q_{XY},q_{Z|XY})$.
Encoder is defined as follows: for any $x\in\mathcal{X}$, define $p_{U|X}(u|x):=p_{\tilde{U}|X_{\text{EQ}}}(u|x_i)$, where $x\in\mathcal{X}_i$. 
Decoder remains the same, i.e., $p_{Z|UY}=p_{Z|\tilde{U}Y}$. Note that $p_{U|X}(u|x)$ must be identical for every $x$ in an 
equivalence class (see the proof of Lemma \ref{lem:characterization}). Therefore, since 
$p_{\tilde{U}|X_{\text{EQ}}}(u|x_i)$ is a valid choice for $p_{U|X}(u|x_i)$, where $x_i$ is the representative element of the 
equivalence class $\mathcal{X}_i$, it follows that the above-defined encoder-decoder pair $(p_{U|X},p_{Z|UY})$ will also securely compute $(q_{XY},q_{Z|XY})$.

In the following we show that $L(U)$ and $L(\tilde{U})$ have the same p.m.f. For simplicity, let $L=L(U)$ and $\tilde{L}=L(\tilde{U})$. 
%{\allowdisplaybreaks
\begin{align*}
p_{L}(l) &= \sum_{x\in\mathcal{X}}q_X(x)p_{L|X}(l|x) \\
&\stackrel{\text{(a)}}{=} \sum_{x\in\mathcal{X}}p_X(x)\sum_{u\in\mathcal{U}:l(u)=l}p_{U|X}(u|x) \\
&= \sum_{u\in\mathcal{U}:l(u)=l}\sum_{x\in\mathcal{X}}q_X(x)p_{U|X}(u|x) \\
&= \sum_{u\in\mathcal{U}:l(u)=l}\sum_{i=1}^k\sum_{x\in\mathcal{X}_i}q_X(x)p_{U|X}(u|x) \\
&\stackrel{\text{(b)}}{=} \sum_{u\in\mathcal{U}:l(u)=l}\sum_{i=1}^k\sum_{x\in\mathcal{X}_i}q_X(x)p_{\tilde{U}|X_{\text{EQ}}}(u|x_i) \\
&= \sum_{u\in\mathcal{U}:l(u)=l}\sum_{i=1}^kp_{\tilde{U}|X_{\text{EQ}}}(u|x_i)\sum_{x\in\mathcal{X}_i}q_X(x) \\
&\stackrel{\text{(c)}}{=} \sum_{u\in\mathcal{U}:l(u)=l}\sum_{i=1}^kp_{\tilde{U}|X_{\text{EQ}}}(u|x_i)p_{X_{\text{EQ}}}(x_i) \\
&= \sum_{i=1}^kp_{X_{\text{EQ}}}(x_i)\sum_{u\in\mathcal{U}:l(u)=l}p_{\tilde{U}|X_{\text{EQ}}}(u|x_i) \\
&= \sum_{i=1}^kp_{X_{\text{EQ}}}(x_i)p_{\tilde{L}|X_{\text{EQ}}}(l|x_i) \\
&= p_{\tilde{L}}(l).
%
%&= \sum_{i=1}^k \sum_{x\in\mathcal{X}_i}\mathbb{E}_{(p_{U|X},p_{Z|UY})}[L(U)|X=x]\cdot p_X(x) \\
%&= \sum_{i=1}^k \sum_{x\in\mathcal{X}_i}\mathbb{E}_{(p_{U|X},p_{Z|UY})}[L(U)|X\in\mathcal{X}_i]\cdot p_X(x) \\
%&\stackrel{\text{(a)}}{=} \sum_{i=1}^k \mathbb{E}_{(p_{U|X},p_{Z|UY})}[L(U)|X\in\mathcal{X}_i]\sum_{x\in\mathcal{X}_i}\cdot p_X(x) \\
%&\stackrel{\text{(b)}}{=} \sum_{i=1}^k \mathbb{E}_{(p_{U|X_{\text{EQ}}},p_{Z|UY})}[L(U)|X_{\text{EQ}}=x_i]\cdot p_{X_{\text{EQ}}}(x_i) \\
%&= \mathbb{E}_{(p_{U|X_{\text{EQ}}},p_{Z|UY})}[L(U)],
\end{align*}
Equality (a) follows from the fact that $l(u)$, the length of $u$, is a deterministic function of $u$.
In (b) we used our definition of the encoder $p_{\tilde{U}|X_{\text{EQ}}}(u|x_i)=p_{U|X}(u|x)$, 
where $x_i$ is the representative element of the equivalence class $\mathcal{X}_i$.
%$p_{\tilde{U}|X_{\EQ}}(u|x_i)=p_{U|X}(u|x)$.
We also used the definition of $p_{X_{\text{EQ}}}(x_i)=\sum_{x\in\mathcal{X}_i}q_{X}(x)$ in (c).
%}
\end{proof}
It follows from Lemma \ref{lem:equiv-problem-alice-bob} that, to study the communication complexity of secure computation of $(q_{XY},q_{Z|XY})$,
it is enough to study the communication complexity of secure computation of the reduced problem $(q_{X_{\text{EQ}}Y},q_{Z|X_{\text{EQ}}Y})$.
}
{\color{black}
\section{Proof of Claim~\ref{claim:color-equiv-protocol}}\label{appendixE}
\RRB{\underline{`Only if' part}:} Let $c:\mathcal{X}_{\text{EQ}}\to\{0,1\}^*$ be a proper coloring of the vertices of $\mathcal{G}_{\text{EQ}}$ that Alice and Bob agree upon. 
We will give a pair of encoder-decoder $(p_{U|X_{\text{EQ}}},p_{Z|UY})$ that securely computes $(q_{X_{\text{EQ}}Y},q_{Z|X_{\text{EQ}}Y})$.
Our encoder will be a deterministic map, which is defined as $p_{U|X_{\text{EQ}}}(u|x):=\mathbbm{1}_{\{u=c(x)\}}$. 
The alphabet of $U$ is $\mathcal{U}=\{a\in\{0,1\}^*:\exists x\in\mathcal{X}_{\text{EQ}} \text{ s.t. } c(x)=a\}$. 
We define our decoder as $p_{Z|UY}(z|u,y):= p_{Z|X_{\text{EQ}}Y}(z|x,y)$, for any $x\in\mathcal{X}_{\text{EQ}}$ such that $c(x)=u$ and $q_{X_{\text{EQ}}Y}(x,y)>0$.
The decoder is well-defined, since $c$ is a proper coloring 
%and all the vertices that have the same color form an independent set 
%in the graph $\G_{\EQ}$, and thus are equivalent from the function computation point of view. 
(if $x,x'$ are such that $c(x)=c(x')$,
then $\{x,x'\}\notin\mathcal{E}_{\text{EQ}}$; therefore, for every $(y,z)\in\mathcal{Y}\times\mathcal{Z}$, $q_{Z|X_{\text{EQ}}Y}(z|x,y)=q_{Z|X_{\text{EQ}}Y}(z|x',y)$,
whenever $q_{X_{\text{EQ}}Y}(x,y)>0,q_{X_{\text{EQ}}Y}(x',y)>0$).
Note that this pair of encoder-decoder $(p_{U|X_{\text{EQ}}},p_{Z|UY})$ correctly computes $(q_{X_{\text{EQ}}Y},q_{Z|X_{\text{EQ}}Y})$.
It also satisfies privacy against both users.
Privacy against Alice follows from the fact that our encoder is a deterministic function, in which $U$ is deterministic function
of $X_{\text{EQ}}$, which implies that the Markov chain $U-X_{\text{EQ}}-(Y,Z)$ trivially holds.
Privacy against Bob follows from the fact that $X_{\text{EQ}}$ is a deterministic function of $(Y,Z)$; see  Remark \ref{remark:bob-always-learn}.\\

\RRB{\underline{`If' part}:} 
Fix a code $\mathcal{C}=(p_{U|X_{\text{EQ}}},p_{Z|UY})$ that securely computes $(q_{X_{\text{EQ}}Y},q_{Z|X_{\text{EQ}}Y})$,
i.e., the induced joint distribution 
\begin{equation}\label{eq:color-equiv-interim}
p_{X_{\text{EQ}}YUZ}=q_{X_{\text{EQ}}Y}p_{U|X_{\text{EQ}}}p_{Z|UY}
\end{equation}
satisfies the two Markov chains 
$U-X_{\text{EQ}}-(Y,Z)$ and $U-(Y,Z)-X_{\text{EQ}}$, which correspond to privacy against Alice and privacy against Bob, respectively.
%Code $\C$ induces the following joint distribution: $p_{X_{\EQ}YUZ}=p_{X_{\EQ}Y}p_{U|X_{\EQ}}p_{Z|UY}$. 
Now we show that, without loss of generality, we can always take the encoder $p_{U|X_{\text{EQ}}}$ to be a deterministic map.

First note that privacy against Alice $U-X_{\text{EQ}}-(Y,Z)$ implies that $U-(X_{\text{EQ}},Y)-Z$ is a Markov chain.
This follows from $0 = I(U;Y,Z|X_{\text{EQ}}) \geq I(U;Z|X_{\text{EQ}},Y)$ and the fact that 
conditional mutual information is always non-negative.
The Markov chain $U-(X_{\text{EQ}},Y)-Z$ implies that, for every $x\in\mathcal{X}_{\text{EQ}},y\in\mathcal{Y}$ such that $q_{X_{\text{EQ}}Y}(x,y)>0$, we have
\begin{align}
q_{Z|X_{\text{EQ}}Y}(z|x,y) &= p_{Z|UX_{\text{EQ}}Y}(z|u,x,y), \nonumber\\
&\hspace{12pt}\text{ for all } u\in\mathcal{U} \text{ s.t. } p_{U|X}(u|x)>0 \label{eq:color-equiv-interim1} \\
&= p_{Z|UY}(z|u,y). \label{eq:color-equiv-interim2}
\end{align}
In \eqref{eq:color-equiv-interim2} we used the \eqref{eq:color-equiv-interim} to write $p_{Z|UX_{\text{EQ}}Y}(z|u,x,y)=p_{Z|UY}(z|u,y)$.
%Fix an arbitrary $x\in\X_{\EQ}$
It follows from \eqref{eq:color-equiv-interim1}-\eqref{eq:color-equiv-interim2} that,
for each $x\in\mathcal{X}_{\text{EQ}}$, if we pick an arbitrary $\hat{u}\in\mathcal{U}$ such that $p_{U|X}(\hat{u}|x)>0$, and define our encoder as $E(x):=\hat{u}$ 
(and leave the decoder as before), then the resulting encoder-decoder pair $(E(X_{\text{EQ}}),p_{Z|UY})$ computes 
$(q_{X_{\text{EQ}}Y},q_{Z|X_{\text{EQ}}Y})$ with zero-error. Note that $(E(X_{\text{EQ}}),p_{Z|UY})$ satisfies privacy against both users:
privacy against Alice comes from the fact that the above-defined encoder $U=E(X_{\text{EQ}})$ is a deterministic function of Alice's input $X_{\text{EQ}}$,
which implies that the Markov chain $U-X_{\text{EQ}}-(Y,Z)$ trivially holds; privacy against Bob follows from the fact that $X_{\text{EQ}}$ is a 
deterministic function of $(Y,Z)$ (see Remark \ref{remark:bob-always-learn}), which implies that the Markov chain $U-(Y,Z)-X_{\text{EQ}}$ trivially holds.

Note that the encoder $p_{U|X_{\text{EQ}}}$ in the given code $\mathcal{C}$ may be a randomized function of Alice's input,
but once we fix the random coins of the encoder, the encoder becomes deterministic. 
For fixed random coins $\vec{r}$ of the encoder, let $E_{\vec{r}}:\mathcal{X}_{\text{EQ}}\to\mathcal{U}$ denote the resulting deterministic encoder, and
let $\mathcal{C}_{\vec{r}}=(E_{\vec{r}}(X_{\text{EQ}}),p_{Z|UY})$ denote the resulting secure code with this deterministic encoder.
As argued above, for every choice of random coins $\vec{r}$ of the encoder, the code $\mathcal{C}_{\vec{r}}$ securely computes 
$(q_{X_{\text{EQ}}Y},q_{Z|X_{\text{EQ}}Y})$.
\iffalse
, which means that Alice's message $U$ 
cannot influence the output $Z$ conditioned on the inputs $(X_{\EQ},Y)$; 
hence, all that Alice can send is the information about her input, so that Bob can sample the output with correct distribution.
%Note that the encoder $p_{U|X_{\EQ}}$ may be a randomized function of Alice's input,
%but once we fix the random coins of the encoder, the encoder becomes deterministic. 
Therefore, since $\C$ computes $(p_{X_{\EQ}Y},p_{Z|X_{\EQ}Y})$ with zero-error, all possible random coins of the encoder give a 
zero-error code. Note that we crucially used privacy against Alice to say this (that, for all possible random coins, the resulting code
with the deterministic encoder is a zero-error code). 
For fixed random coins $\vec{r}$, let $E_{\vec{r}}:\X_{\EQ}\to\U$ denote the resulting deterministic encoder, and
let $\C_{\vec{r}}=(E_{\vec{r}}(x),p_{Z|UY})$ denote the resulting secure code with this deterministic encoder; 
note that the decoder remains the same $p_{Z|UY}$, which may be randomized map.
Now, fix random coins $\vec{r}$ of the encoder, and let $E_{\vec{r}}(x)$ denote the resulting deterministic encoder.
%Suppose Alice's input is $x\in\X_{\EQ}$, and she sends the message $E_{\vec{r}}(x)$ to Bob. 
\fi
Since $\mathcal{C}_{\vec{r}}$ is a secure code with zero-error, the coloring defined by $c_{\vec{r}}(x):=E_{\vec{r}}(x), \forall x\in\mathcal{X}_{\text{EQ}}$ 
will be a proper coloring of the vertices of $\mathcal{G}_{\text{EQ}}$. Run through the possible random coins $\vec{r}$ of the encoder; 
this will produce a random coloring $(c_{\vec{r}})_{\vec{r}}$ of the vertices, where for any randomness $\vec{r}$, 
the corresponding coloring $c_{\vec{r}}$ is a proper coloring.
}
{\color{black}
\section{Proofs of Claims~\ref{claim:disjoint-messages} \& \ref{claim:equal-alpha-vectors}}\label{appendixF}
\begin{proof}[Proof of Claim~\ref{claim:disjoint-messages}]
We prove this by contradiction. Suppose $\vec{\alpha}_i^{(y)}\neq\vec{\alpha}_j^{(y')}$ and $\U_{i}^{(y)}\cap\U_{j}^{(y')}\neq\phi$. Let $u\in\U_{i}^{(y)}\cap\U_{j}^{(y')}$.
Define two sets $\text{Supp}(\vec{\alpha}_i^{(y)})=\{x\in\X: \vec{\alpha}_i^{(y)}(x)>0\}$ and $\text{Supp}(\vec{\alpha}_j^{(y')})=\{x\in\X: \vec{\alpha}_j^{(y')}(x)>0\}$. We analyze two cases, one when these two sets are equal, and the other, when they are not.

{\bf Case 1.} $\text{Supp}(\vec{\alpha}_i^{(y)})=\text{Supp}(\vec{\alpha}_j^{(y')})$: 
since $\vec{\alpha}_i^{(y)}$ and $\vec{\alpha}_j^{(y')}$ are distinct probability vectors with same support, there must exist $x$ and $x'$ such that $\vec{\alpha}_i^{(y)}(x)>0$, $\vec{\alpha}_j^{(y')}(x)>0$, $\vec{\alpha}_i^{(y)}(x')>0$, $\vec{\alpha}_j^{(y')}(x')>0$, and $\frac{\vec{\alpha}_i^{(y)}(x')}{\vec{\alpha}_i^{(y)}(x)}\neq\frac{\vec{\alpha}_j^{(y')}(x')}{\vec{\alpha}_j^{(y')}(x)}$.
Since $\vec{\alpha}_i^{(y)}(x)>0, \vec{\alpha}_i^{(y)}(x')>0$, we have $q_{Z|XY}(z|x,y)>0$ and $q_{Z|XY}(z|x',y)>0$ for every $z\in\Z_{i}^{(y)}$ . Similarly, since $\vec{\alpha}_j^{(y')}(x)>0,\vec{\alpha}_j^{(y')}(x')>0$, we have $q_{Z|XY}(z'|x,y')>0$ and $q_{Z|XY}(z'|x',y')>0$ for every $z'\in\Z_{j}^{(y')}$.
Note that $\frac{\vec{\alpha}_i^{(y)}(x')}{\vec{\alpha}_i^{(y)}(x)}=\frac{q_{Z|XY}(z|x',y)}{q_{Z|XY}(z|x,y)}$ and $\frac{\vec{\alpha}_j^{(y')}(x')}{\vec{\alpha}_j^{(y')}(x)}=\frac{q_{Z|XY}(z'|x',y')}{q_{Z|XY}(z'|x,y')}$, where, by hypothesis, $\frac{q_{Z|XY}(z|x',y)}{q_{Z|XY}(z|x,y)}\neq\frac{q_{Z|XY}(z'|x',y')}{q_{Z|XY}(z'|x,y')}$.

Since $u\in\U_{i}^{(y)}$, there exists $z\in\Z_{i}^{(y)}$ such that $p(u,z|y)>0$. This implies -- by privacy against Bob -- that $p(u,z|x,y)>0$ and $p(u,z|x',y)>0$.
Consider $p_{UZ|XY}(u,z|x,y)$ and expand it as follows:
\begin{align}
p_{UZ|XY}(u,z|x,y) &= q_{Z|XY}(z|x,y)p_{U|XYZ}(u|x,y,z) \notag \\
&= q_{Z|XY}(z|x,y)p_{U|YZ}(u|y,z) \label{eq:disjoint-messages-interim1}
\end{align}
Since $p_{UZ|XY}(u,z|x,y)>0$, all the terms above are non-zero and well-defined. In the last equality we used privacy against Bob to write $p_{U|XYZ}(u|x,y,z)=p_{U|YZ}(u|y,z)$. Now expand $p_{UZ|XY}(u,z|x,y)$ in another way as follows:
\begin{align}
p_{UZ|XY}(u,z|x,y) &= p_{U|XY}(u|x,y)p_{Z|UXY}(z|u,x,y) \notag \\
&= p_{U|XY}(u|x)p_{Z|UY}(z|u,y) \label{eq:disjoint-messages-interim2}
\end{align}
Again, all the terms above are non-zero and well-defined because $p_{UZ|XY}(u,z|x,y)>0$. We used $U-X-Y$ to write $p_{U|XY}(u|x,y)=p_{U|X}(u|x)$ and $Z-(U,Y)-X$ to write $p_{Z|UXY}(z|u,x,y)=p_{Z|UY}(z|u,y)$ in \eqref{eq:disjoint-messages-interim2}. Now comparing \eqref{eq:disjoint-messages-interim1} and \eqref{eq:disjoint-messages-interim2} gives the following:
\begin{align}
q_{Z|XY}(z|x,y)p_{U|YZ}(u|y,z)=p_{U|X}(u|x)p_{Z|UY}(z|u,y) \label{eq:disjoint-messages-interim3}
\end{align}
Since $p_{UZ|XY}(u,z|x',y)>0$, we can apply the same arguments as above with $p_{UZ|XY}(u,z|x',y)$ and get the following:
\begin{align}
q_{Z|XY}(z|x',y)p_{U|YZ}(u|y,z)=p_{U|X}(u|x')p_{Z|UY}(z|u,y) \label{eq:disjoint-messages-interim4}
\end{align}
Note that all the terms on both sides of \eqref{eq:disjoint-messages-interim3} and \eqref{eq:disjoint-messages-interim4} are non-zero. Dividing \eqref{eq:disjoint-messages-interim3} by \eqref{eq:disjoint-messages-interim4} gives the following:
\begin{align}
\frac{q_{Z|XY}(z|x,y)}{q_{Z|XY}(z|x',y)} = \frac{p_{U|X}(u|x)}{p_{U|X}(u|x')}. \label{eq:first-ratio}
\end{align}
Similarly, since $u\in\U_{j}^{(y')}$, there exists $z'\in\Z_{j}^{(y')}$ such that $p(u,z'|y')>0$. This implies -- by privacy against Bob -- that $p(u,z'|x,y')>0$ and $p(u,z'|x',y')>0$. Applying the above arguments with $p(u,z'|x,y')>0$ and $p(u,z'|x',y')>0$ gives
\begin{align}
\frac{q_{Z|XY}(z'|x,y')}{q_{Z|XY}(z'|x',y')}=\frac{p_{U|X}(u|x)}{p_{U|X}(u|x')}. \label{eq:second-ratio}
\end{align}
Comparing \eqref{eq:first-ratio} and \eqref{eq:second-ratio} gives $\frac{q_{Z|XY}(z|x,y)}{q_{Z|XY}(z|x',y)}=\frac{q_{Z|XY}(z'|x,y')}{q_{Z|XY}(z'|x',y')}$, which is a contradiction.\\
{\bf Case 2.} $\text{Supp}(\vec{\alpha}_i^{(y)})\neq \text{Supp}(\vec{\alpha}_j^{(y')})$: 
assume, without loss of generality, that $\text{Supp}(\vec{\alpha}_i^{(y)})\setminus \text{Supp}(\vec{\alpha}_j^{(y')})\neq\phi$. Let $x\in \text{Supp}(\vec{\alpha}_i^{(y)})\setminus \text{Supp}(\vec{\alpha}_j^{(y')})$. This implies that $\vec{\alpha}_i^{(y)}(x)>0$ and $\vec{\alpha}_j^{(y')}(x)=0$. Since $\text{Supp}(\vec{\alpha}_j^{(y')})\neq \phi$ (because $\vec{\alpha}_j^{(y')}$ is not a zero vector), there exists $x'\in \text{Supp}(\vec{\alpha}_j^{(y')})$ such that $\vec{\alpha}_j^{(y')}(x')>0$.
Note that $\vec{\alpha}_i^{(y)}(x)>0$ implies $q_{Z|XY}(z|x,y)>0$ for every $z\in\Z_{i}^{(y)}$; $\vec{\alpha}_j^{(y')}(x)=0$  implies $q_{Z|XY}(z'|x,y')=0$ for every $z'\in\Z_{j}^{(y')}$; and $\vec{\alpha}_j^{(y')}(x')>0$ implies $q_{Z|XY}(z'|x',y')>0$ for every $z'\in\Z_{j}^{(y')}$.
Since $u\in\U_{i}^{(y)}\cap\U_{j}^{(y')}$, there exists $z\in\Z_{i}^{(y)}$ and $z'\in\Z_{j}^{(y')}$ such that $p(u,z|y)>0$ and $p(u,z'|y')>0$. 
These imply -- by privacy against Bob -- that $p(u,z|x,y)>0$ and $p(u,z'|x',y')>0$.
Now consider $p(u,z'|x,y')$ and expand it as follows:
\begin{align}
p(u,z'|x,y') &= p(u|x,y')p(z'|x,y',u) \label{eq:second-comparison1} \\
&= p(u|x)p(z'|y',u) \label{eq:second-comparison2}
\end{align}
We used the Markov chain $U-X-Y$ to write $p(u|x,y')=p(u|x)$, where $p(u|x)>0$ because $p(u,z|x,y)>0$. We used the Markov chain $Z-(Y,U)-X$ to write $p(z'|x,y',u)=p(z'|y',u)$ in \eqref{eq:second-comparison1}, which is greater than zero because $p(u,z'|x',y')>0$. Putting all these together in \eqref{eq:second-comparison2} gives $p(u,z'|x,y')>0$, which implies $q_{Z|XY}(z'|x,y')>0$, a contradiction.
\end{proof}
\begin{proof}[Proof of Claim~\ref{claim:equal-alpha-vectors}]
Since the Markov chain $U-X-Y$ holds, we have that the set of messages that Alice sends to Bob are same for all inputs of Bob. Now it follows from Claim \ref{claim:disjoint-messages} that for every $y\in\Y$ we can partition the set of all possible messages $\U$ as follows: $\U=\U_1^{(y)}\biguplus\U_2^{(y)}\biguplus\hdots\biguplus\U_{k(y)}^{(y)}$, where $\U_i^{(y)}$ is as defined earlier in \eqref{eq:one-round_set-msgs}. Consider any two $y,y'\in\Y$. We have
\begin{align}
\U_1^{(y)}\biguplus\U_2^{(y)}\biguplus\hdots\biguplus\U_{k(y)}^{(y)}=\U_1^{(y')}\biguplus\U_2^{(y')}\biguplus\hdots\biguplus\U_{k(y')}^{(y')}.\label{eq:equal-partition_messages}
\end{align}
First observe that $k(y)=k(y')$. Otherwise, there exists $i\in[k(y)]$ and $j\in [k(y')]$ such that $\U_i^{(y)}\cap \U_j^{(y')}\neq\phi$ and $\vec{\alpha}_i^{(y)}\neq \vec{\alpha}_j^{(y')}$, which contradicts Claim \ref{claim:disjoint-messages}. From now on we denote $k(y)$ by $k$ for every $y\in\Y$.

Suppose $\{\vec{\alpha}_1^{(y)},\vec{\alpha}_2^{(y)},\hdots,\vec{\alpha}_{k}^{(y)}\}\neq\{\vec{\alpha}_1^{(y')},\vec{\alpha}_2^{(y')},\hdots,\vec{\alpha}_{k}^{(y')}\}$. It means that there exists an $i\in [k]$ such that $\vec{\alpha}_i^{(y)}\notin\{\vec{\alpha}_1^{(y')},\vec{\alpha}_2^{(y')},\hdots,\vec{\alpha}_{k}^{(y')}\}$. Since $\U_i^{(y)}$ is associated with $\vec{\alpha}_i^{(y)}$, we have by \eqref{eq:equal-partition_messages} that $\U_i^{(y)}\notin\U_1^{(y')}\biguplus\U_2^{(y')}\biguplus\hdots\biguplus\U_{k}^{(y')}$. This contradicts \eqref{eq:equal-partition_messages}.
\end{proof}
}
\section{Proofs of Theorems~\ref{thm:one-round-characterization} and \ref{thm:ps_rate-3}}\label{appendixG}
\begin{proof}[Proof of Theorem~\ref{thm:one-round-characterization}]
\RRB{\underline{`Only if' part}:} (1) has been shown in Claim \ref{claim:equal-alpha-vectors}. (2) follows from the following argument:
take any $x\in\X$ and consider the following set of equalities:
\begin{align*}
\sum_{z\in\Z_{i}^{(y)}}q_{Z|XY}(z|x,y) &\stackrel{\text{(a)}}{=} \sum_{u\in\U_i} p_{U|XY}(u|x,y) \\
&\stackrel{\text{(b)}}{=} \sum_{u\in\U_i} p_{U|XY}(u|x,y') \\
&\stackrel{\text{(c)}}{=} \sum_{z\in\Z_{i}^{(y')}}q_{Z|XY}(z|x,y'),
\end{align*}
where (a) and (c) follow from the fact that Bob's output $Z$ belongs to $\Z_{i}^{(y)}$, when his input is $y$ (or belongs to $\Z_{i}^{(y')}$, when his input is $y'$) if and only if the message $U$ that Alice sends to Bob belongs to $\U_{i}$. (b) follows from the Markov chain $U-X-Y$. \\

\RRB{\underline{`If' part}:} We show this direction by giving a secure protocol in Figure~\ref{fig:one-round-protocol}. Now we prove that this protocol is perfectly secure, i.e., it satisfies perfect correctness and perfect privacy. \\

\begin{figure}
\hrule
\vspace{0.20cm}
{\bf One-round secure protocol}
\vspace{0.20cm}
\hrule
\vspace{0.25cm}
\noindent {\bf Input:} Alice has $x\in\X$ and Bob has $y\in\Y$. \\
\noindent {\bf Output:} Bob outputs $z$ with probability $q_{Z|XY}(z|x,y)$.
\begin{center}{\bf Protocol}\end{center}
\begin{enumerate}
\item Both Alice and Bob agree on an element in the alphabet $\Y$, say $y_1$, beforehand.
\item Alice sends $u_i$ to Bob with probability $\sum_{z\in\Z_i^{(y)}}q_{Z|XY}(z|x,y_1)$.
\item Upon receiving $u_i$, Bob fixes any element $x'$ for which $\alpha_i^{(y)}(x')>0$, and outputs $z$ with probability
$\frac{q_{Z|XY}(z|x',y)}{\sum_{z\in\Z_{i}^{(y)}}q_{Z|XY}(z|x',y)}$.
\end{enumerate}
\hrule
\vspace{0.20cm}
\caption{A one-round secure protocol.}
\label{fig:one-round-protocol}
\end{figure}
{Correctness:} Suppose the protocol of Figure~\ref{fig:one-round-protocol} produces an output according to the p.m.f. $p_{Z|XY}(z|x,y)$. We show below that $p_{Z|XY}(z|xy)$ is equal to $q_{Z|XY}(z|x,y)$.
{\allowdisplaybreaks
\begin{align}
&p_{Z|XY}(z|x,y)\\
 &= p_{ZU|XY}(z,u_i|x,y) \label{eq:char_corr-interim1} \\
&= p_{U|XY}(u_i|x,y)\times p_{Z|UXY}(z|u_i,x,y) \notag \\
&= p_{U|X}(u_i|x)\times p_{Z|UY}(z|u_i,y) \label{eq:char_corr-interim2} \\
&= \Big(\sum_{z\in\Z_i^{(y)}}q_{Z|XY}(z|x,y_1)\Big)\times\Big(\frac{q_{Z|XY}(z|x',y)}{\sum_{z\in\Z_{i}^{(y)}}q_{Z|XY}(z|x',y)}\Big) \label{eq:char_corr-interim3} \\
&= \Big(\sum_{z\in\Z_i^{(y)}}q_{Z|XY}(z|x,y)\Big)\times\Big(\frac{q_{Z|XY}(z|x,y)}{\sum_{z\in\Z_{i}^{(y)}}q_{Z|XY}(z|x,y)}\Big) \label{eq:char_corr-interim4} \\
&= q_{Z|XY}(z|x,y) \notag
\end{align}
In \eqref{eq:char_corr-interim1} we assume that $z\in\Z_i^{(y)}$, and \eqref{eq:char_corr-interim1} is an equality because the message that Alice sends to Bob is a deterministic function of Bob's input and output. In \eqref{eq:char_corr-interim2} we used the Markov chain $U-X-Y$ to write $p(u_i|x,y)=p(u_i|x)$ and $Z-(U,Y)-X$ to write $p(z|u_i,x,y)=p(z|u_i,y)$. In \eqref{eq:char_corr-interim3} we substituted the values of $p(u_i|x,y)$ and $p(z|u_i,y)$ from the protocol of Figure~\ref{fig:one-round-protocol}. In \eqref{eq:char_corr-interim4} we used the facts that $\sum_{z\in\Z_i^{(y)}}q_{Z|XY}(z|x,y_1)=\sum_{z\in\Z_i^{(y)}}q_{Z|XY}(z|x,y)$ (which follows from the assumption -- see the second item in the theorem statement) and $\frac{q_{Z|XY}(z|x',y)}{\sum_{z\in\Z_{i}^{(y)}}q_{Z|XY}(z|x',y)}=\frac{q_{Z|XY}(z|x,y)}{\sum_{z\in\Z_{i}^{(y)}}q_{Z|XY}(z|x,y)}$ (which follows from the fact that the matrix $A_i^{(y)}=\vec{\alpha}_i^{(y)}\times\vec{\gamma}_i^{(y)}$, defined earlier, is a rank-one matrix). \\
}
{Privacy:} we need to show that if $p_{Z|XY}(z|x_1,y)>0$ and $p_{Z|XY}(z|x_2,y)>0$ for some $x_1,x_2,y,z$, then for every $u_i$, $p(u_i|x_1,y,z)=p(u_i|x_2,y,z)$ must hold. This follows because $p(u_i|x,y,z)=\mathbbm{1}_{\{z\in\Z_i^{(y)}\}}$, i.e., $u_i$ is a deterministic function of $(y,z)$ and is independent of Alice's input.
\end{proof}
\balance
\RRB{\begin{proof}[Remaining proof of Theorem~\ref{thm:ps_rate-3}]
Here we only show that $\mathbb{E}[L]\geq H(W)$, where $L,W$ are defined in the proof of Theorem~\ref{thm:ps_rate-3}.
{\allowdisplaybreaks
\begin{align}
\mathbb{E}[L] &= \sum_{x\in\X}p_X(x)\mathbb{E}[L|X=x] \notag \\
&= \sum_{x\in\X}p_X(x)\sum_{l\in\mathbb{N}} l \times p_{L|X}(l|x) \notag \\
%&= \sum_{x\in\X}p_X(x)\sum_{l\in\mathbb{N}} l \sum_{u\in\U:l(u)=l}p_{UL|X}(u,l(u)|x) \notag \\
&= \sum_{x\in\X}p_X(x)\sum_{l\in\mathbb{N}} l \sum_{u\in\U:l(u)=l}p_{U|X}(u|x) \label{eq:ps_privacy-bob-interim1}\\
&= \sum_{x\in\X}p_X(x)\sum_{u\in\U}l(u)p_{U|X}(u|x) \notag \\
&= \sum_{x\in\X}p_X(x)\sum_{i=1}^k \sum_{u\in\U_i} l(u) p(u|x) \notag \\
&\geq \sum_{x\in\X}p_X(x) \sum_{i=1}^k \sum_{u\in\U_i} l_i\times p(u|x),\text{ where }l_i=\min_{u\in\U_i}l(u) \notag \\
&= \sum_{x\in\X}p_X(x) \sum_{i=1}^k l_i \sum_{u\in\U_i} p(u|x) \notag \\
&= \sum_{x\in\X}p_X(x) \sum_{i=1}^k l_i \sum_{u\in\U_i} p(u|x,y_1) \label{eq:ps_privacy-bob-interim15} \\
&= \sum_{x\in\X}p_X(x) \sum_{i=1}^k l_i \sum_{z\in\Z_i^{(y_1)}} p_{Z|XY}(z|x,y_1) \label{eq:ps_privacy-bob-interim2} \\
&= \sum_{x\in\X}p_X(x) \sum_{i=1}^k l_i\times p_{W|X}(i|x) \label{eq:ps_privacy-bob-interim3} \\
&= \sum_{i=1}^k l_i \sum_{x\in\X}p_X(x)p_{W|X}(i|x) \notag \\
&= \sum_{i=1}^k l_i\times p_{W}(i) \notag \\
&= \mathbb{E}[L'] \label{eq:ps_privacy-bob-interim4} \\
&\geq H(W). \label{eq:ps_privacy-bob-interim5}
\end{align}
\eqref{eq:ps_privacy-bob-interim1} holds because $l(u)$ is deterministic function of $u$.
\eqref{eq:ps_privacy-bob-interim15} follows from the Markov chain $U-X-Y$. \eqref{eq:ps_privacy-bob-interim2} follows from the 
fact that when Bob has $y_1$ as his input, his output $Z$ belongs to $\Z_i^{(y)}$ if and only if $U$ belongs to $\U_i$.
In \eqref{eq:ps_privacy-bob-interim3} we used the equality $p_{W|X}(i|x)=\sum_{z\in\Z_i^{(y)}} q_{Z|XY}(z|x,y)$ (see \eqref{eq:W-interim2}).
Since $l(u)_{i\in\U}$'s are lengths of a valid prefix-free encoding of $U$, and $l_i=\min_{u\in\U_i}l(u)$ for $i\in[k]$, 
it can be verified that $l_1,l_2,\hdots,l_k$ are lengths of a valid prefix-free encoding of $W$. 
The random variable $L'$ in \eqref{eq:ps_privacy-bob-interim4} corresponds to this prefix-free encoding of $W$, where $\Pr\{L'=l_i\}=p_W(i)$.
Inequality \eqref{eq:ps_privacy-bob-interim5} follows from the fact that the expected length of any prefix-free binary code 
for a random variable $W$ is lower-bounded by $H(W)$ \cite[Theorem 5.4.1]{CoverJ06}.  
Since this argument holds for any prefix-free encoding of $U$, we have $L^*_{B-\text{pvt}}\geq H(W)$.
}
\end{proof}
}
}
 %\balance
\iftoggle{paper}
{
\bibliographystyle{IEEEtran}
\bibliography{Bibliography}}
{\bibliographystyle{IEEEtran}
\bibliography{Bibliography}
}

\end{document}